\documentclass[journal, romanappendices, final]{IEEEtran}
\usepackage{latexsym,amsfonts,amssymb,amsthm, amsmath, setspace, subfig, verbatim,cite,mathtools,graphicx,bm,stfloats, tikz,xcolor, bm, algorithmic, algorithm, enumerate, paralist,url, flushend}

\newcommand{\be}{\begin{equation}}
\newcommand{\ee}{\end{equation}}
\newcommand{\ben}{\begin{equation*}}
\newcommand{\een}{\end{equation*}}
\newcommand{\mc}{\mathcal}
\newcommand{\mbf}{\mathbf}

\newtheorem{lemma}{Lemma}

\newtheorem{defi}{Definition}[section]
\newtheorem{thm}{Theorem}
\newtheorem{fact}{Fact}
\newtheorem{note}{Note}
\newtheorem{rem}{Remark}

\newtheorem{corr}{Corollary}

\newcommand{\e}{\epsilon}

\newcommand{\abs}[1]{\lvert#1\rvert}
\newcommand{\norm}[1]{\lVert#1\rVert}
\newcommand{\expec}{\mathbb{E}}

\newcommand{\bst}{\boldsymbol{\theta}}
\newcommand{\bsth}{\hat{\boldsymbol{\theta}}}
\newcommand{\by}{\mathbf{y}}

\begin{document}
\title{Cluster-Seeking James-Stein Estimators}
\author{K. Pavan Srinath
and Ramji Venkataramanan
\thanks{ This work was supported in part by a Marie Curie Career Integration Grant (Grant Agreement No. 631489) and an Early Career Grant from the Isaac Newton Trust. This paper was presented in part at the 2016 IEEE International Symposium on Information Theory.}%
\thanks{K.~P.~Srinath and R.~Venkataramanan are with Department of Engineering, University of Cambridge, Cambridge CB2 1PZ, UK (e-mail: \{pk423, rv285\}@cam.ac.uk).}
}
\maketitle

\begin{abstract}
This paper considers the problem of estimating a high-dimensional vector of parameters $\boldsymbol{\theta} \in \mathbb{R}^n$ from a noisy observation. The noise vector is i.i.d. Gaussian with known variance. For a squared-error loss function, the James-Stein (JS) estimator is known to dominate the simple maximum-likelihood (ML) estimator when the dimension $n$ exceeds two. The JS-estimator shrinks the observed vector towards the origin, and the  risk reduction over the ML-estimator is greatest for  $\boldsymbol{\theta}$ that lie close to the origin. JS-estimators can be generalized to shrink the data towards any target subspace. Such estimators also dominate the ML-estimator, but the  risk reduction is significant only when $\boldsymbol{\theta}$ lies close to the subspace.  
This leads to the question: in the absence of prior information about $\boldsymbol{\theta}$, how do we design estimators that give significant risk reduction over the ML-estimator for a wide range of $\boldsymbol{\theta}$? 

In this paper, we propose shrinkage estimators that attempt to infer the structure of $\boldsymbol{\theta}$ from the observed data in order to construct a good attracting subspace. In particular,  the components of the observed vector are separated into clusters, and the elements in each cluster shrunk towards a common attractor. The number of clusters and the attractor for each cluster are determined from the observed vector. We provide concentration results for the squared-error loss and convergence results for the risk of the proposed estimators. The results show that the estimators give significant risk reduction over the ML-estimator for a wide range of $\boldsymbol{\theta}$, particularly for large $n$. Simulation results are provided to support the theoretical claims.
\end{abstract}

\begin{IEEEkeywords}
High-dimensional estimation, Large deviations bounds, Loss function estimates, Risk estimates, Shrinkage estimators  \end{IEEEkeywords}

\section{Introduction}
\label{sec:intro}
\IEEEPARstart{C}onsider the problem of estimating a vector of parameters $\bst \in \mathbb{R}^n$ from a noisy observation 
$\mbf{y}$ of the form
\[ \mbf{y} = \bst + \mbf{w}. \]
The noise vector $\mbf{w} \in \mathbb{R}^n$ is  distributed as $\mc{N}(\mbf{0}, \sigma^2 \mbf{I})$, i.e., its components are i.i.d. Gaussian random variables with mean zero and variance $\sigma^2$.   We emphasize that $\bst$ is deterministic, so the joint probability density function of $\mbf{y}=[y_1, \ldots, y_n]^T$  for a given $\bst$ is 
\begin{equation}\label{eq_pdf_y}
 p_{\bst}(\mathbf{y}) = \frac{1}{\left(2\pi\sigma^2\right)^{\frac{n}{2}}}e^{-\frac{ \norm{\mathbf{y}-\boldsymbol{\theta}}^2}{2\sigma^2}}.
\end{equation}
The performance of an estimator $\bsth$ is measured using the squared-error loss function given by 
\begin{equation*}\label{eq_loss_function}
 L(\boldsymbol{\theta}, \hat{\boldsymbol{\theta}}(\mbf{y}) ) \vcentcolon= \norm{\hat{\boldsymbol{\theta}}(\mbf{y})- \boldsymbol{\theta}}^2,
\end{equation*}
where $\norm{\cdot}$ denotes the Euclidean norm. The \emph{risk} of the estimator for a given $\bst$ is the expected value of the loss function:
\begin{equation*}\label{eq_risk_function}
 R(\boldsymbol{\theta}, \hat{\boldsymbol{\theta}} ) \vcentcolon= \mathbb{E}\left[ \norm{ \hat{\boldsymbol{\theta}}(\mbf{y}) - \boldsymbol{\theta}}^2 \right],
\end{equation*}
where the expectation is computed using the density in \eqref{eq_loss_function}. The {\it normalized risk} 
is $R(\boldsymbol{\theta}, \hat{\boldsymbol{\theta}})/n$.

Applying the maximum-likelihood (ML) criterion to \eqref{eq_pdf_y} yields the ML-estimator  $\hat{\boldsymbol{\theta}}_{ML} = \mathbf{y}$. The ML-estimator is an unbiased estimator, and its risk is  $R(\boldsymbol{\theta}, \hat{\boldsymbol{\theta}}_{ML}) = n \sigma^2$.  The goal of this paper is to design  estimators that give significant risk reduction over  $\bsth_{ML}$ for a wide range of $\bst$, without any prior assumptions about its structure. 

In 1961 James and Stein published a surprising result \cite{stein}, proposing an estimator that uniformly achieves lower risk than $\hat{\boldsymbol{\theta}}_{ML}$ for any $\boldsymbol{\theta} \in \mathbb{R}^{n}$, for $n \geq 3$. Their estimator 
$\hat{\boldsymbol{\theta}}_{JS}$ is given by
\begin{equation}\label{eq_regular_JS}
 \hat{\boldsymbol{\theta}}_{JS} = \left[1-\frac{(n-2)\sigma^2}{\Vert\mathbf{y}\Vert^2} \right]\mathbf{y},
\end{equation}
and  its risk is \cite[Chapter $5$, Thm. 5.1]{lehmannCas98}
 \begin{equation}\label{eq_theta_JS_1}
 R\left( \boldsymbol{\theta}, \hat{\boldsymbol{\theta}}_{JS}\right) = n\sigma^2 - (n-2)^2\sigma^4\mathbb{E} 
 \left[\frac{1}{\norm{\mathbf{y}}^2} \right].
\end{equation}
Hence  for $n \geq 3$,  
\begin{equation}\label{eq_js_ml_risk}
R(\boldsymbol{\theta}, \hat{\boldsymbol{\theta}}_{JS}) < R(\boldsymbol{\theta}, \hat{\boldsymbol{\theta}}_{ML})=n\sigma^2, ~~\forall \boldsymbol{\theta} \in \mathbb{R}^n.
\end{equation}
An estimator $\hat{\boldsymbol{\theta}}_1$ is said to {\it dominate} another estimator $\hat{\boldsymbol{\theta}}_2$ if 
\begin{equation*}
 R(\boldsymbol{\theta}, \hat{\boldsymbol{\theta}}_{1}) \leq R(\boldsymbol{\theta}, \hat{\boldsymbol{\theta}}_{2}),  ~~\forall \boldsymbol{\theta} \in \mathbb{R}^n,
\end{equation*}
with the inequality being strict for at least one $\boldsymbol{\theta}$. 
Thus \eqref{eq_js_ml_risk} implies that the James-Stein estimator  (JS-estimator)  dominates the ML-estimator. Unlike the ML-estimator, the JS-estimator  is non-linear and biased. However, the risk reduction over  the ML-estimator can be significant, making it an attractive option in many situations --- see, for example, \cite{efron_morris2}.

By evaluating the  expression in \eqref{eq_theta_JS_1}, it can be shown that the risk of the JS-estimator depends on $\bst$ only via  $\norm{\bst}$ \cite{stein}. Further,  the risk decreases as $\norm{\bst}$ decreases. (For intuition about this,  note in \eqref{eq_theta_JS_1} that for large $n$, $\norm{\mbf{y}}^2 \approx n\sigma^2 + \norm{\bst}^2$.)  The dependence of the risk on 
$\norm{\bst}$ is illustrated in Fig. \ref{fig:js_comp},  where the average loss of the JS-estimator  is plotted versus $\norm{\bst}$, for two different choices of $\bst$.

The JS-estimator in \eqref{eq_regular_JS} shrinks each element of $\mbf{y}$ towards the origin. Extending this idea, JS-like estimators can be defined by shrinking $\mbf{y}$ towards any vector, or more generally, towards a target subspace $\mathbb{V} \subset \mathbb{R}^{n}$.  Let $P_\mathbb{V}(\mathbf{y})$ denote the projection of $\mathbf{y}$ onto $\mathbb{V}$, so that $\Vert \mathbf{y} - P_\mathbb{V}(\mathbf{y}) \Vert^2 = \min_{\mathbf{v} \in \mathbb{V}}\Vert \mathbf{y} - \mathbf{v}\Vert^2$. Then the JS-estimator that shrinks $\mathbf{y}$ towards the subspace $\mathbb{V}$ is
\begin{equation}\label{JS_attractor_gen}
\hat{\boldsymbol{\theta}} = P_\mathbb{V}(\mathbf{y})   + \left[1 - \frac{(n-d-2)\sigma^2}{\left\Vert \mathbf{y} - P_\mathbb{V}(\mathbf{y})  \right \Vert^2} \right] \left( \mathbf{y} - P_\mathbb{V}(\mathbf{y}) \right ),
\end{equation}
where $d$ is the dimension of  $\mathbb{V}$.\footnote{The dimension $n$ has to be greater than $d+2$ for the estimator to achieve lower risk than $\bsth_{ML}$.}
A classic example of such an  estimator is Lindley's estimator \cite{lindley}, which shrinks $\mathbf{y}$ towards the one-dimensional subspace defined by the all-ones vector $\mbf{1}$. It is given by 
 \begin{equation}\label{eq_one_part_estimator}
   \hat{\boldsymbol{\theta}}_{L} = \bar{y}\mathbf{1}   + \left[1 - \frac{(n-3)\sigma^2}{\norm{\mathbf{y} - \bar{y}\mathbf{1}}^2} \right] \left( \mathbf{y} - \bar{y}\mathbf{1}  \right ),
  \end{equation}
where $\bar{y} \vcentcolon= \frac{1}{n}\sum_{i=1}^ny_i$ is the empirical mean of $\mbf{y}$. 

\begin{figure}[t]
    \centering
    \begin{subfloat}[\label{subfig_n_10_comp_1}]{
       \includegraphics[width=3.25in]{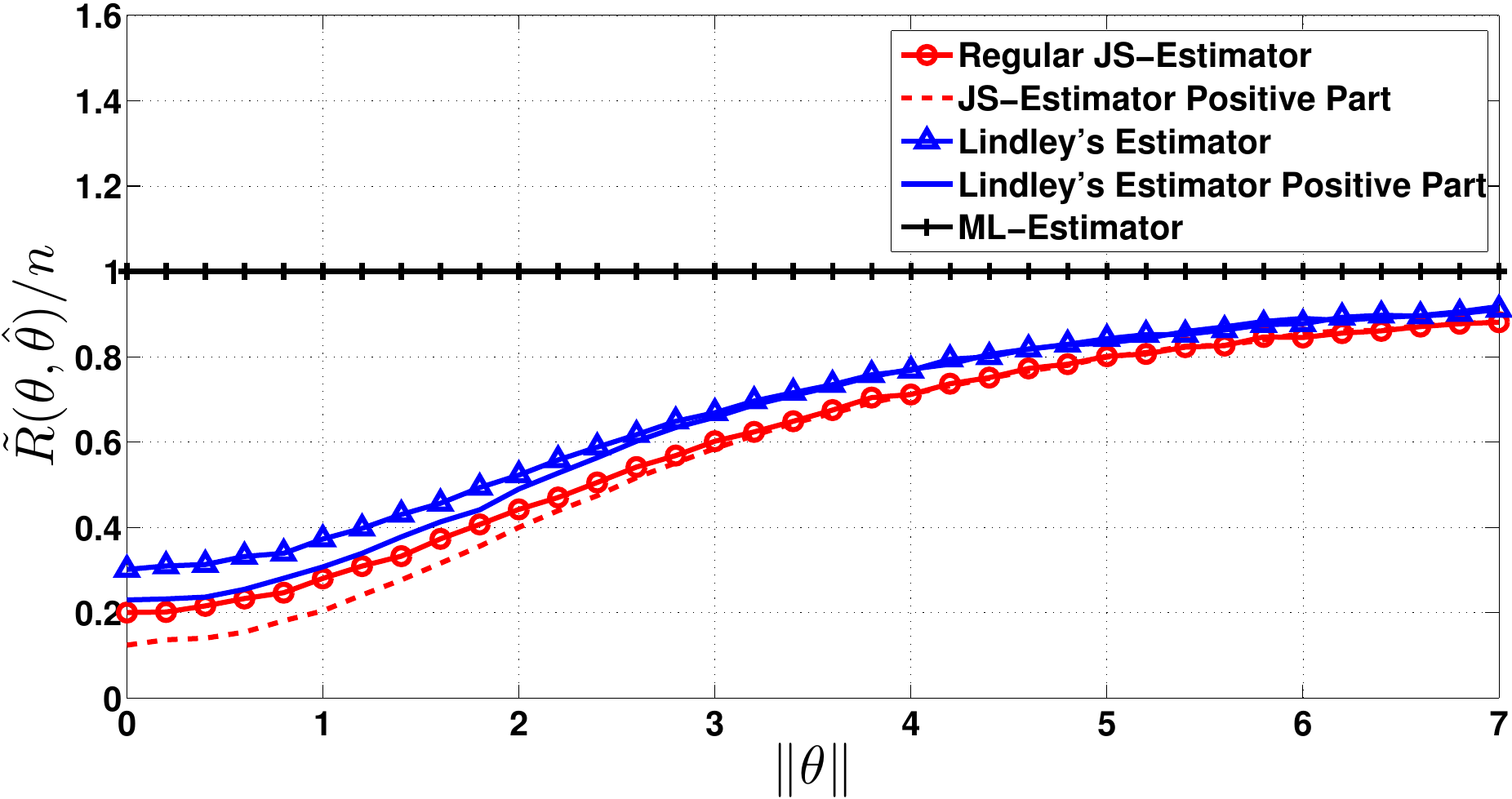}   
       }
    \end{subfloat}
  \quad
    \begin{subfloat}[\label{subfig_n_10_comp_2}]{
       \includegraphics[width=3.25in]{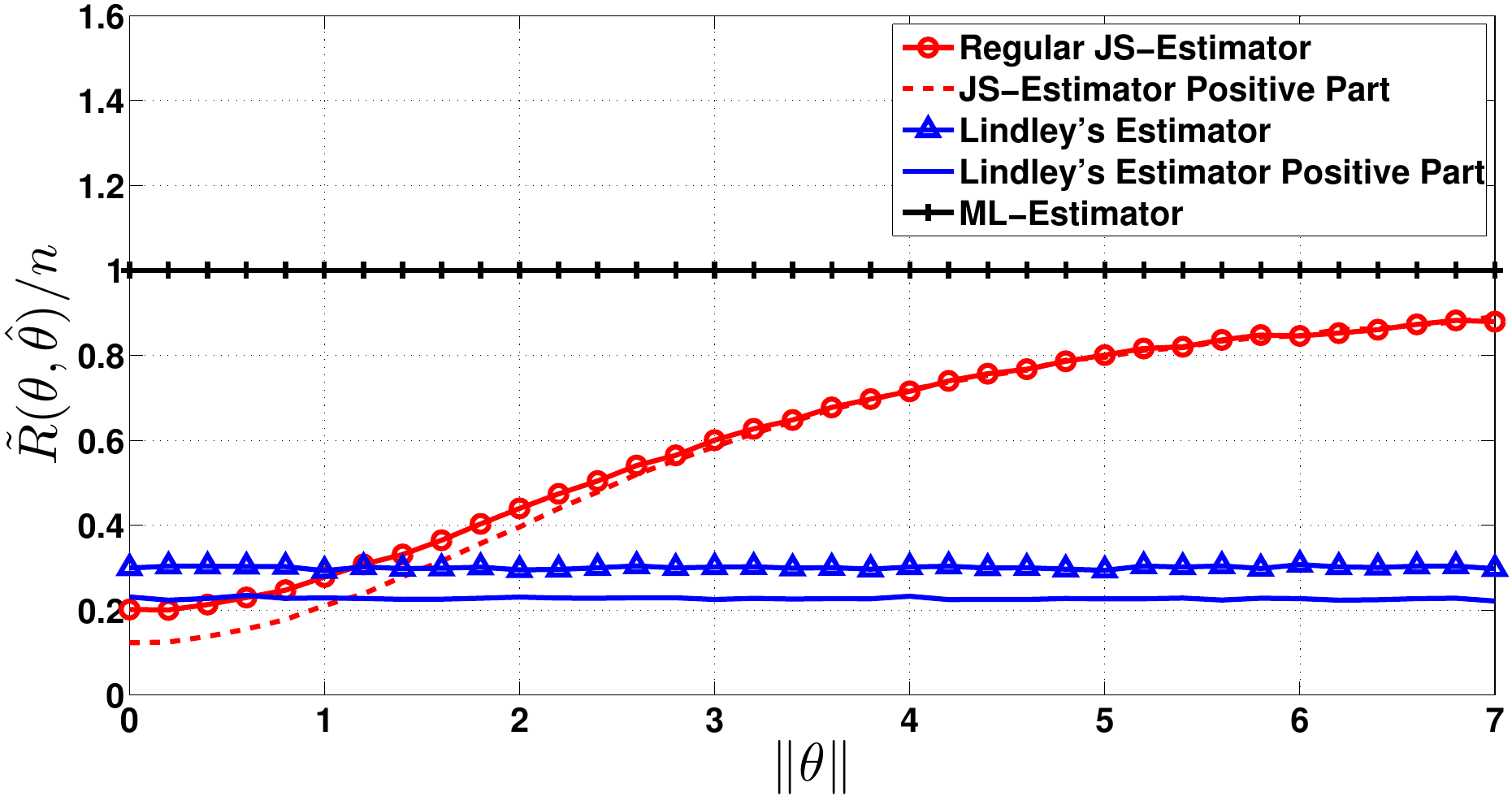}
       }
    \end{subfloat}    
    \caption{\small Comparison of the average normalized loss of the regular JS-estimator, Lindley's estimator, and their positive-part versions for $n=10$ as a function of $\Vert \boldsymbol{\theta} \Vert$. The loss of the ML-estimator is $\sigma^2=1$.  In (a) $\theta_i = \norm{\bst}/{\sqrt{10}}$, $i=1,\cdots,5$, and     $\theta_i = - \norm{\bst}/{\sqrt{10}}$, $i=6,\cdots,10$. In (b), $\theta_i =  \norm{\bst}/{\sqrt{10}}$, $\forall i$.}
    \label{fig_JS_vs_Lindley}
    \label{fig:js_comp}
\end{figure}

It can be shown that the different variants of the JS-estimator  such as \eqref{eq_regular_JS},\eqref{JS_attractor_gen},\eqref{eq_one_part_estimator}  all dominate the ML-estimator.\footnote{The risks of JS-estimators of the form \eqref{JS_attractor_gen} can  usually be computed using Stein's lemma \cite{stein2}, which states that $\expec[Xg(X)]=E[g'(X)]$, where $X$ is a standard normal random variable, and $g$ a weakly differentiable function.}   Further, all JS-estimators share the following key property \cite{efron_morris,george3, george1}: \textbf{the smaller the Euclidean distance between $\bst$ and the attracting vector, the smaller the risk}.

Throughout this paper,  the term ``attracting vector" refers to the vector   that $\by$ is shrunk towards.
 For $\bsth_{JS}$ in \eqref{eq_regular_JS}, the attracting vector is $\mbf{0}$,  and the risk reduction over $\bsth_{ML}$ is larger when $\norm{\bst}$ is close to zero. Similarly, if the components of $\bst$ are clustered around some value $c$, a JS-estimator with attracting vector $c\mathbf{1}$ would give significant risk reduction over $\bsth_{ML}$.  
One motivation for Lindley's estimator in \eqref{eq_one_part_estimator} comes from a guess  that the components of  $\bst$ are close to its empirical mean $\bar{\theta}$ --- since we do not know $\bar{\theta}$, we approximate it by $\bar{y}$ and use the attracting vector $\bar{y} \mathbf{1}$.

Fig. \ref{fig:js_comp} shows how the performance of $\bsth_{JS}$ and  $\bsth_L$ depends on the structure of  $\bst$. In the left panel of the figure, the empirical mean $\bar{\theta}$ is always  $0$, so the risks of both estimators increase monotonically with  
$\norm{\bst}$. In the right panel, all the components of $\bst$ are all equal to $\bar{\theta}$. In this case, the distance from the attracting vector for $\bsth_L$ is $\norm{\bst-\bar{y} \mbf{1}} = \sqrt{(\sum_{i=1}^nw_i)^2/n}$, so the risk  does not vary with  $\norm{\bst}$; in contrast the risk of $\bsth_{JS}$ increases with $\norm{\bst}$ as its attracting vector is $\mbf{0}$.

The risk reduction obtained by using a JS-like shrinkage estimator over $\bsth_{ML}$  crucially  depends on the choice of attracting vector.  To achieve significant risk reduction for a wide range of $\bst$, in this paper, we  infer the structure of  $\bst$ from the data $\mbf{y}$ and choose attracting vectors tailored to this structure.  
The idea is to partition $\mbf{y}$ into clusters, and shrink the components in each cluster towards a common element (attractor). Both the number of clusters and the attractor for each cluster are to be determined based on the data $\mbf{y}$.

As a motivating example, consider a $\bst$ in which half the components are equal to $\norm{\bst}/\sqrt{n}$ and the other half  are equal to  $-\norm{\bst}/\sqrt{n}$. Fig. \ref{fig_JS_vs_Lindley}(a) shows that the risk reduction of both $\bsth_{JS}$ and $\bsth_{L}$ diminish as $\norm{\bst}$ gets larger. This is because the empirical mean $\bar{y}$ is close to zero, hence $\bsth_{JS}$ and  $\bsth_{L}$ both shrink $\mbf{y}$ towards $\textbf{0}$.  An ideal JS-estimator would shrink the   $y_i$'s corresponding to  $\theta_i = \norm{\bst}/\sqrt{n}$ towards the attractor $\norm{\bst}/\sqrt{n}$, and the remaining  observations  towards   $-\norm{\bst}/\sqrt{n}$. Such an estimator would give handsome gains over $\bsth_{ML}$ for all $\bst$ with the above structure. 
On the other hand, if $\bst$ is such that all its components are equal (to $\bar{\theta}$),  Lindley's estimator $\bsth_L$ is an excellent choice,  with significantly smaller risk than $\bsth_{ML}$ for all values of $\norm{\bst}$ (Fig. \ref{fig_JS_vs_Lindley}(b)). 

We would like an intelligent estimator that can correctly distinguish  between different $\bst$ structures (such as the two above) and choose an appropriate attracting vector, based only on  $\mbf{y}$. We propose such estimators in Sections \ref{sec_2_hybrid}  and \ref{sec_multi_attractor}. For reasonably large $n$, these estimators choose a good attracting subspace tailored to the structure of $\bst$, and use an approximation of the best attracting vector within the subspace.

The main contributions of our paper are as follows.
\begin{itemize}
 \item We construct a two-cluster JS-estimator, and provide concentration results for the squared-error loss, and asymptotic convergence results for its risk. Though this estimator does not dominate the ML-estimator, it is shown to provide significant risk reduction over Lindley's estimator and the regular JS-estimator when the components of $\bst$ can be approximately separated into two clusters. 
 \item We present a hybrid JS-estimator that, for any $\boldsymbol{\theta}$ and for large $n$, has risk close to the minimum of that of Lindley's estimator and the proposed two-cluster JS-estimator. Thus the hybrid estimator  asymptotically dominates both the ML-estimator and Lindley's estimator, and gives significant risk reduction over the ML-estimator for a wide range of $\bst$.
 
 \item We generalize the above idea to define general multiple-cluster hybrid JS-estimators, and provide concentration and convergence results for the squared-error loss and risk, respectively. 
 \item We provide simulation results that support the theoretical results on the loss function. The simulations indicate that the hybrid estimator gives significant risk reduction over the ML-estimator for a wide range of $\bst$ even for modest values of $n$, e.g. $n=50$. The empirical risk of the  hybrid estimator converges  rapidly to the theoretical value with growing $n$. 
\end{itemize}

\subsection{Related work}

George \cite{george3, george1} proposed  a ``multiple shrinkage estimator", which is a convex   combination of multiple subspace-based JS-estimators of the form \eqref{JS_attractor_gen}. The coefficients defining the convex combination give larger weight to the estimators whose target subspaces are closer to $\mbf{y}$. Leung and Barron  \cite{LeungBarron06,LeungThesis} also studied similar ways of combining estimators and their risk properties. Our proposed estimators also seek to emulate the best among a class of subspace-based estimators, but there are some key differences. In \cite{george3, george1}, the target subspaces are fixed a priori, possibly based on  prior knowledge about where $\bst$ might lie. In the absence of such prior knowledge, it may not be possible to choose good target subspaces. This motivates the estimators proposed in this paper, which use a target subspace constructed from the data $\mbf{y}$.  The nature  of clustering in $\bst$ is inferred from  $\mbf{y}$, and used to define a suitable subspace. 

Another difference from earlier work is in how the attracting vector is determined given a target subspace $\mathbb{V}$. Rather than choosing the attracting vector as the projection of $\mbf{y}$ onto $\mathbb{V}$, we use an approximation of  the projection of $\bst$ onto 
$\mathbb{V}$. This approximation is computed from $\mbf{y}$, and concentration inequalities are  provided to guarantee the goodness of the approximation. 

The risk of a JS-like estimator is typically computed using Stein's lemma \cite{stein2}. However, the data-dependent subspaces we use result in estimators that are hard to analyze using this technique. We therefore use concentration inequalities to bound the loss function of the proposed estimators. Consequently, our theoretical bounds get sharper as the dimension $n$ increases, but may not be accurate for small $n$.   However, even for relatively small $n$, simulations indicate that the risk reduction over the ML-estimator is significant  for a wide range of $\bst$.

Noting that the shrinkage factor multiplying $\mbf{y}$ in \eqref{eq_regular_JS} could be negative, Stein proposed the following positive-part JS-estimator \cite{stein}:
 \begin{equation}\label{pp_JS}
  \hat{\boldsymbol{\theta}}_{JS_{+}} = \left[1-\frac{(n-2)\sigma^2}{\Vert\mathbf{y}\Vert^2} \right]_+\mathbf{y}, 
 \end{equation}
 where $X_+$ denotes $\max(0,X)$.  We can similarly define positive-part versions of  JS-like estimators such as \eqref{JS_attractor_gen} and \eqref{eq_one_part_estimator}. The positive-part Lindley's estimator is given by 
 \begin{equation}\label{pp_lindley}
   \hat{\boldsymbol{\theta}}_{L_+} = \bar{y}\mathbf{1}   + \left[1 - \frac{(n-3)\sigma^2}{\norm{\mathbf{y} - \bar{y}\mathbf{1}}^2} \right]_+ \left( \mathbf{y} - \bar{y}\mathbf{1}  \right ).
  \end{equation}
 Baranchik \cite{baranchik} proved that $\hat{\boldsymbol{\theta}}_{JS_{+}}$ dominates $\hat{\boldsymbol{\theta}}_{JS}$, and his result also proves that $\hat{\boldsymbol{\theta}}_{L_+}$ dominates $\hat{\boldsymbol{\theta}}_{L}$.      Estimators that dominate $\hat{\boldsymbol{\theta}}_{JS_{+}}$ are discussed in \cite{shao,maruyama2}. 
 Fig. \ref{fig_JS_vs_Lindley} shows  that  the positive-part versions can give noticeably lower loss than the regular JS and Lindley estimators.  However, for large $n$,  the shrinkage factor  is positive with high probability, hence the positive-part estimator is nearly always identical to the regular JS-estimator. Indeed,  for large $n$, $\frac{\norm{\mbf{y}}^2}{n} \approx \frac{\norm{\bst}^2}{n} +\sigma^2$, and the shrinkage factor is 
 \[ \left( 1- \frac{(n-2) \sigma^2}{\norm{\mbf{y}}^2} \right) \approx  \left( 1- \frac{(n-2) \sigma^2}{\norm{\mbf{\bst}}^2 + n \sigma^2} \right) >0. \]
 
 We analyze the positive-part version of the proposed hybrid estimator using concentration inequalities. Though we cannot guarantee that the hybrid estimator dominates the positive-part JS or Lindley estimators for any finite $n$,  we show that for large $n$, the loss of the hybrid estimator is equal to the minimum of that of the positive-part Lindley's estimator and the cluster-based estimator with high probability (Theorems \ref{thm3} and \ref{thm_hybrid_partition}).

The rest of the paper is organized as follows. In Section \ref{sec_2_attractor}, a two-cluster JS-estimator is proposed and its performance analyzed. Section \ref{sec_2_hybrid} presents a hybrid JS-estimator along with its performance analysis. General multiple-attractor JS-estimators are discussed in Section \ref{sec_multi_attractor}, and  simulation results to corroborate the theoretical analysis are provided in Section \ref{sec:simulations}. The proofs of the main results are given in Section \ref{sec:proofs}.  Concluding remarks and possible directions for future research constitute Section \ref{sec_conc}. 

\subsection{Notation} \label{subsec:notation}
Bold lowercase letters are used to denote vectors, and plain lowercase letters for their entries.  For example, the entries of  $\mathbf{y} \in \mathbb{R}^n$ are  $y_i$, $i=1,\cdots,n$. All vectors have length $n$ and are column vectors, unless otherwise mentioned. For vectors $\mathbf{y},  \mbf{z} \in \mathbb{R}^{n }$, $\langle \mathbf{y},\mathbf{z} \rangle$ denotes their Euclidean inner product. The all-zero vector and the all-one vector of length $n$ are denoted by  $\mathbf{0}$ and $\mathbf{1}$, respectively. The complement of a set $A$ is denoted by $A^c$. For a finite set $A$ with real-valued elements, $\min(A)$ denotes the minimum of the elements in $A$. We use  $\mathsf{1}_{\{\mathcal{E}\}}$ to denote the indicator function of an event $\mathcal{E}$. A central chi-squared distributed random variable with $n$ degrees of freedom is denoted by $\mathcal{X}^2_{n}$. 
The $Q$-function is given by $Q(x) = \int_x^\infty \frac{1}{\sqrt{2\pi}}\exp(-\frac{x^2}{2})dx$, and $Q^c(x):=1-Q(x)$. For a random variable $X$, $X_+$ denotes $\max(0,X)$. For real-valued functions $f(x)$ and $g(x)$, the notation $f(x) = o(g(x))$ means that $\lim_{x \to 0} [{f(x)}/{g(x)}] = 0$, and $f(x) = O(g(x))$ means that $\lim_{x \to \infty} [{f(x)}/{g(x)}] = c$ for some positive constant $c$. 

For a sequence of random variables $\{X_n\}_{n=1}^{\infty}$, $X_n \overset{P}{\longrightarrow} X$, $X_n \overset{a.s.}{\longrightarrow} X$, and $X_n \overset{\mathcal{L}^1}{\longrightarrow} X$ respectively denote convergence in probability, almost sure convergence, and convergence in $\mathcal{L}^1$ norm to the random variable $X$. 

We use the following shorthand for concentration inequalities. Let $\{X_n(\bst), \bst \in \mathbb{R}^n \}_{n=1}^\infty$ be a sequence of random variables. The notation $X_n(\bst) \doteq X$, where $X$ is either a random variable or a constant, means that for any $\epsilon > 0$, 
\begin{equation}\label{notation_doteq}
 \mathbb{P}\left(\vert X_n(\bst) - X \vert\geq \epsilon \right) \leq Ke^{-\frac{nk \min(\epsilon^2,1)}{\max(\Vert \bst\Vert^2/n,1)}},
\end{equation}
where $K$ and $k$ are positive constants that do not depend on $n$ or $\bst$. The exact values  of $K$ and $k$ are not specified. 

The shrinkage estimators we propose have the general form
\[ \bsth =  \boldsymbol{\nu} +    \left[1 - \frac{n\sigma^2}{\left\Vert \mathbf{y} - \boldsymbol{\nu} \right \Vert^2} \right]_+ \left( \mathbf{y} - \boldsymbol{\nu}  \right ). \]
For $1 \leq i \leq n$,  the $i$th component of the attracting vector $\boldsymbol{\nu}$ is the attractor for $y_i$  (the point towards  which it is shrunk).


\section{A two-cluster James-Stein estimator} \label{sec_2_attractor}

Recall the example in Section \ref{sec:intro} where $\bst$ has half its components equal to $\norm{\bst}/\sqrt{n}$, and the other half equal to $\norm{\bst}/\sqrt{n}$. Ideally, we would like to shrink the $y_i$'s corresponding to the first group towards 
$\norm{\bst}/\sqrt{n}$, and the remaining  points towards $-\norm{\bst}/\sqrt{n}$. However, without an oracle, we cannot accurately guess which point each $y_i$ should be shrunk towards.  We would like to obtain an estimator that identifies separable clusters in $\mbf{y}$, constructs a suitable attractor for each cluster, and shrinks the $y_i$ in each cluster towards its attractor. 

We start by dividing the observed data  into two clusters based on a separating point $s_{\mbf{y}}$, which is obtained from $\mbf{y}$. A natural choice for the $s_{\mbf{y}}$ would be the empirical mean $\bar{\theta}$; since this is unknown we use $s_{\mbf{y}}=\bar{y}$. Define the clusters
\begin{align*}
\mathcal{C}_1 & \vcentcolon= \{y_i,  \ 1 \leq i \leq n  \mid y_i > \bar{y} \}, \\
\mathcal{C}_2 & \vcentcolon= \{y_i,  \ 1 \leq i \leq n  \mid y_i \leq \bar{y} \}. 
\end{align*}
The points in  $\mathcal{C}_1$ and $\mathcal{C}_2$ will be shrunk towards attractors $a_1(\mathbf{y})$ and $a_2(\mathbf{y})$, respectively, where $a_1, a_2: \mathbb{R}^n \to \mathbb{R}$ are defined in \eqref{eq_attractor1} later in this section. For brevity, we henceforth do not indicate the dependence of the attractors on $\mathbf{y}$. Thus the attracting vector is
\be \label{two_partition_attractor}
 \boldsymbol{\nu}_{2} \vcentcolon= a_1\begin{bmatrix}
  \mathsf{1}_{\{y_1 > \bar{y}\}}\\
  \mathsf{1}_{\{y_2 > \bar{y}\}}\\
  \vdots \\
  \mathsf{1}_{\{y_n  > \bar{y}\}}
  \end{bmatrix}
+ a_2\begin{bmatrix}
  \mathsf{1}_{\{y_1 \leq \bar{y}\}}\\
  \mathsf{1}_{\{y_2 \leq \bar{y}\}}\\
  \vdots \\
  \mathsf{1}_{\{y_n  \leq \bar{y}\}}
  \end{bmatrix},
 \ee
 with $a_1$ and $a_2$ defined in \eqref{eq_attractor1}. The proposed estimator is 
 \begin{align} \nonumber
   \hat{\boldsymbol{\theta}}_{JS_{2}} &= \boldsymbol{\nu}_{2}   + \left[1 - \frac{n\sigma^2}{\left\Vert \mathbf{y} - \boldsymbol{\nu}_{2} \right \Vert^2} \right]_+ \left( \mathbf{y} - \boldsymbol{\nu}_{2}  \right ) \\ \label{eq_two_part_estimator}
   &= \boldsymbol{\nu}_{2}   + \left[1 - \frac{\sigma^2}{g\left({\left\Vert \mathbf{y} - \boldsymbol{\nu}_{2} \right \Vert^2}/{n}\right)} \right] \left( \mathbf{y} - \boldsymbol{\nu}_{2}  \right ),
\end{align}
where  the function $g$ is defined as
 \begin{equation}\label{eq_g_func}
   g(x) \vcentcolon= \max(\sigma^2,x), \quad x \in \mathbb{R}.
 \end{equation}
 
 The attracting vector $\boldsymbol{\nu}_{2}$ in \eqref{two_partition_attractor} lies in a {two}-dimensional subspace defined by the orthogonal vectors 
$[\mathsf{1}_{\{y_1 > \bar{y} \}},  \cdots,  \mathsf{1}_{\{y_n  > \bar{y}\}} ]^T$ and 
 $[\mathsf{1}_{\{y_1  \leq \bar{y} \}},  \cdots,  \mathsf{1}_{\{y_n  \leq \bar{y}} ]^T$. 
To derive the values of $a_1$ and $a_2$ in \eqref{two_partition_attractor}, it is useful to compare $\boldsymbol{\nu}_{2}$ to the attracting vector of Lindley's estimator in \eqref{eq_one_part_estimator}. Recall that Lindley's attracting vector lies in the one-dimensional subspace spanned by $\mbf{1}$. The vector lying in this subspace that is closest in Euclidean distance to $\bst$ is its projection $\bar{\theta}\mbf{1}$. Since $\bar{\theta}$ is unknown, we use the approximation $\bar{y}$ to define the attracting vector $\bar{y} \mbf{1}$.

Analogously, the vector in the two-dimensional subspace defined by \eqref{two_partition_attractor} that is closest to $\bst$ is the {projection} of  $\bst$ onto this subspace. Computing this projection,  the desired values for $a_1, a_2$ are found to be
\be a^{des}_1 =  \frac{\sum_{i=1}^n\theta_i\mathsf{1}_{\{y_i > \bar{y}\}}}{\sum_{i=1}^n\mathsf{1}_{\{y_i > \bar{y} \}}}, \quad
a^{des}_2= \frac{\sum_{i=1}^n\theta_i\mathsf{1}_{\{y_i \leq \bar{y}\}}}{\sum_{i=1}^n\mathsf{1}_{\{y_i \leq \bar{y} \}}}.  \label{eq:a1a2_des}
\ee
As the $\theta_i$'s are not available, we define the attractors $a_1, a_2$ as approximations of $a^{des}_1, a^{des}_2$, obtained using the following concentration results. 

\begin{lemma}\label{prop_y1_y2}
 We have 
 \begin{align}
\label{eq:lem1_yigeq} 
&   \frac{1}{n} \sum_{i=1}^n y_i \mathsf{1}_{\{y_i > \bar{y}\}} \doteq  \frac{1}{n} \sum_{i=1}^n\theta_i\mathsf{1}_{\{y_i > \bar{y}\}}+ 
\frac{\sigma}{n \sqrt{2\pi}}\sum_{i=1}^n e^{-\frac{\left(\bar{\theta}-\theta_i\right)^2}{2\sigma^2}}, \\
\label{eq:lem1_yileq}
&  \frac{1}{n} \sum_{i=1}^n y_i\mathsf{1}_{\{y_i \leq \bar{y}\}} \doteq  \frac{1}{n} \sum_{i=1}^n\theta_i\mathsf{1}_{\{y_i \leq \bar{y}\}} - \frac{\sigma}{n \sqrt{2\pi}}\sum_{i=1}^n e^{-\frac{\left(\bar{\theta}-\theta_i\right)^2}{2\sigma^2}},\\ 
\label{eq:lem1_theti_yi}
&  \frac{1}{n} \sum_{i=1}^n \theta_i \mathsf{1}_{\{y_i > \bar{y}\}} \doteq  \sum_{i=1}^n \theta_i Q\left(\frac{\bar{\theta}-\theta_i}{\sigma}\right), \\ 
\label{eq:lem1_theti_yi1}
& \frac{1}{n} \sum_{i=1}^n \theta_i \mathsf{1}_{\{y_i \leq \bar{y}\}} \doteq \frac{1}{n} \sum_{i=1}^n \theta_i Q^c\left(\frac{\bar{\theta}-\theta_i}{\sigma}\right), \\
\label{eq:lem1_1yi}
&\mathbb{P}\left(  \frac{1}{n}\left\vert \sum_{i=1}^n \mathsf{1}_{\{y_i > \bar{y}\}} -  \sum_{i=1}^n Q\left(\frac{\bar{\theta}-\theta_i}{\sigma}\right) \right \vert \geq \epsilon \right) \leq Ke^{-nk\epsilon^2}, \\ 
\label{eq:lem1_1yi1}
&\mathbb{P}\left( \frac{1}{n}\left\vert  \sum_{i=1}^n \mathsf{1}_{\{y_i \leq \bar{y}\}}  -   \sum_{i=1}^n Q^c\left(\frac{\bar{\theta}-\theta_i}{\sigma}\right) \right \vert \geq \epsilon \right)  \leq Ke^{-nk\epsilon^2}.
 \end{align}
 where $Q^c\left(\frac{\bar{\theta}-\theta_i}{\sigma} \right ) \vcentcolon= 1- Q\left(\frac{\bar{\theta}-\theta_i}{\sigma} \right )$. Recall from Section \ref{subsec:notation} that the symbol  $\doteq$ is shorthand for a concentration inequality of the form \eqref{notation_doteq}.
\end{lemma}
The proof is given in Appendix \ref{pf_prop_y1_y2}.

Using Lemma \ref{prop_y1_y2}, we can obtain estimates for $a^{des}_1, a^{des}_2$ in \eqref{eq:a1a2_des} provided we have an estimate for the term   $\frac{\sigma}{n \sqrt{2\pi}}\sum_{i=1}^n e^{-\frac{\left(\bar{\theta}-\theta_i\right)^2}{2\sigma^2}}$.  This is achieved via the following concentration result.

\begin{lemma}\label{prop_bias}
Fix $\delta > 0$.  Then for any $\epsilon >0$, we have
\begin{align} \nonumber
&\mathbb{P}\left(\Bigg\vert \frac{\sigma^2}{2n\delta}\sum_{i=0}^{n} \mathsf{1}_{\left\{\left\vert y_i - \bar{y}\right\vert \leq \delta \right\}} - \left(\frac{\sigma}{n\sqrt{2\pi}}\sum_{i=0}^{n} e^{-\frac{\left(\bar{\theta}-\theta_i\right)^2}{2\sigma^2}} + \kappa_n\delta  \right) \Bigg\vert \right.\\ \label{eq_prop_bias}
& \geq \epsilon \Bigg) \leq 10e^{-{nk\epsilon^2}},
\end{align}
where $k$ is a positive constant and $\vert \kappa_n \vert \leq \frac{1}{\sqrt{2\pi e}}$.
\end{lemma}
The proof is given in Appendix \ref{pf_prop_bias}.
\begin{note}\label{note_delta_bound}
Henceforth in this paper, $\kappa_n$ is used to denote a generic bounded constant (whose exact value is not needed) that is a coefficient of $\delta$ in expressions of the form $f(\delta) = a + \kappa_n\delta + o(\delta)$ where $ a$ is some constant. As an example to illustrate its usage, let $f(\delta) = \frac{1}{a + b \delta}$, where $a > 0$ and $\vert b \delta \vert < a$. Then, we have $f(\delta) = \frac{1}{a(1 + b \delta/a)} = \frac{1}{a}(1 + \frac{b}{a}\delta + o(\delta)) = \frac{1}{a} + \kappa_n\delta + o(\delta)$. 
\end{note}
Using Lemmas \ref{prop_y1_y2} and \ref{prop_bias}, the two attractors are defined to be
\begin{equation}
\begin{split}
\label{eq_attractor1}
  {a}_1 = \frac{\sum_{i=1}^n y_i\mathsf{1}_{\left\{y_i > \bar{y}\right\}} - \frac{\sigma^2}{2\delta}\sum_{i=0}^{n} \mathsf{1}_{\left\{\vert y_i - \bar{y} \vert \leq \delta \right\}} }{\sum_{i=1}^n \mathsf{1}_{\left\{y_i > \bar{y}\right\}}}, \\ 
  {a}_2 = \frac{\sum_{i=1}^n y_i\mathsf{1}_{\left\{y_i \leq \bar{y}\right\}} + \frac{\sigma^2}{2\delta}\sum_{i=0}^{n} \mathsf{1}_{\left\{\vert y_i - \bar{y} \vert \leq \delta \right\}} }{\sum_{i=1}^n \mathsf{1}_{\left\{y_i \leq \bar{y}\right\}}}.
  \end{split}
\end{equation}
With $\delta >0$ chosen to be a small positive number, this completes the specification of the attracting vector in \eqref{two_partition_attractor}, and hence the two-cluster JS-estimator in \eqref{eq_two_part_estimator}.

Note that $\boldsymbol{\nu}_2$, defined by \eqref{two_partition_attractor}, \eqref{eq_attractor1}, is an approximation of the projection of $\bst$ onto the two-dimensional subspace $\mathbb{V}$ spanned by the  vectors $[1_{y_1 > \bar{y}},\cdots, 1_{y_n > \bar{y}}]^T$ and $[1_{y_1 \leq \bar{y}},  \cdots,1_{y_n \leq \bar{y}}]^T$. We remark that $\boldsymbol{\nu}_2$, which approximates the vector in $\mathbb{V}$ that is closest  to $\bst$, is distinct from the projection of $\mbf{y}$ onto $\mathbb{V}$. While the analysis is easier (there would be no terms involving $\delta$) if $\boldsymbol{\nu}_2$ were chosen to be a projection of $\by$ (instead of $\bst$) onto $\mathbb{V}$, our numerical simulations suggest that this choice yields significantly higher risk. The intuition behind choosing the projection of  $\bst$ onto $\mathbb{V}$ is that if all the $y_i$ in a group are to be attracted to a common point (without any prior information), a natural choice would be the mean of the $\theta_i$ within the group, as  in \eqref{eq:a1a2_des}. This mean is determined by the term $\mathbb{E}(\sum_{i=1}^n \theta_i \mathsf{1}_{\{y_i \geq \bar{y} \}})$, which  is different from $\mathbb{E}(\sum_{i=1}^n y_i \mathsf{1}_{\{y_i \geq \bar{y} \}})$   because  
$$\mathbb{E}\left( \sum_{i=1}^n (y_i - \theta_i) \mathsf{1}_{\{y_i \geq \bar{y} \}} \right)  = \mathbb{E}\left( \sum_{i=1}^n w_i \mathsf{1}_{\{y_i \geq \bar{y} \}} \right) \neq 0.
$$ The term involving $\delta$ in \eqref{eq_attractor1} approximates $\mathbb{E}(\sum_{i=1}^n w_i \mathsf{1}_{\{y_i \geq \bar{y} \}})$.

\begin{note}
The attracting vector  $\boldsymbol{\nu}_{2} $ is dependent not just on $\mathbf{y}$ but also on $\delta$, through the two attractors $a_1$ and $a_2$. In Lemma \ref{prop_bias}, for the deviation probability in \eqref{eq_prop_bias} to fall exponentially in $n$, $\delta$ needs to be held constant and independent of $n$. From a practical design point of view,  what is needed is $n\delta^2 \gg 1$. Indeed , for $\frac{\sigma^2}{2n\delta}\sum_{i=0}^{n} \mathsf{1}_{\left\{\left\vert y_i - \bar{y}\right\vert \leq \delta \right\}}$ to be a reliable approximation of the term $\frac{\sigma}{n\sqrt{2\pi}}\sum_{i=1}^ne^{-\frac{\left(\bar{\theta}-\theta_i\right)^2}{2\sigma^2}}$, it is shown in Appendix \ref{pf_prop_bias}, specifically in \eqref{prop_bias_eq1} that we need $n\delta^2 \gg 1$. Numerical experiments suggest  a value of $5/\sqrt{n}$ for $\delta$ to be  large enough for a good approximation.  
\end{note}
We now present the first main result of the paper.
\begin{thm}\label{thm1}
 The loss function of the two-cluster JS-estimator in \eqref{eq_two_part_estimator} satisfies the following:
 \begin{enumerate}[(1)]
 \item For any $\epsilon > 0$, and for any fixed $\delta > 0$ that is independent of $n$,
 \begin{align}
 \nonumber
 &\mathbb{P}\left( \Bigg \vert \frac{1}{n}\norm{\boldsymbol{\theta} - \hat{\boldsymbol{\theta}}_{JS_{2}}}^2 - \left[\min\left(\beta_n,\frac{\beta_n\sigma^2}{\alpha_n+\sigma^2}\right) + \kappa_n\delta  \right. \right. \\ 
 \label{eq1_thm1_statement}
 & \hspace{0.5in}+ o(\delta) \bigg]\Bigg \vert\geq  \epsilon \Bigg) \leq Ke^{-\frac{nk \min(\epsilon^2,1)}{\max(\Vert \bst\Vert^2/n,1)}},
 \end{align}
 where $\alpha_n, \beta_n$ are given by \eqref{eq_alpha2} and \eqref{eq_beta2} below, and $K$ is a positive constant that is independent of $n$ and $\delta$, while $k = \Theta(\delta^2)$ is another positive constant  that is independent of $n$ (for a fixed $\delta$).
 \item For a sequence of $\bst$ with increasing dimension $n$, if $\limsup_{n \to \infty} \Vert\bst \Vert^2/n < \infty$, we have
\begin{align} %
 \nonumber
& \hspace{-18pt} \lim_{n\to \infty} \bigg \vert\frac{1}{n}R( \boldsymbol{\theta}, \hat{\boldsymbol{\theta}}_{JS_{2}})- \bigg[\min\left(\beta_n,\frac{\beta_n\sigma^2}{\alpha_n+\sigma^2}\right) + \kappa_n\delta + o(\delta)\bigg]\bigg \vert \\ 
\label{eq2_thm1_statement}
  &= 0.
\end{align} 
 \end{enumerate}
The constants $\beta_n, \alpha_n$ are given by 
 \begin{equation}\label{eq_beta2}
 \beta_n \vcentcolon= \frac{\Vert \boldsymbol{\theta} \Vert^2}{n} - \frac{ c_1^2}{n}\sum_{i=1}^nQ\left(\frac{\bar{\theta}-\theta_i}{\sigma}\right) - \frac{c_2^2}{n}\sum_{i=1}^nQ^c\left(\frac{\bar{\theta}-\theta_i}{\sigma}\right),
\end{equation}
\begin{equation}\label{eq_alpha2}
 \alpha_n \vcentcolon=\beta_n- \left(\frac{2\sigma}{n\sqrt{2\pi}} \right)\left(\sum_{i=1}^ne^{-\frac{\left(\bar{\theta}-\theta_i\right)^2}{2\sigma^2}} \right)\left(c_1 - c_2\right),
 \end{equation}
where
\begin{equation}\label{eq_c_1_c_2}
  c_1 \vcentcolon = \frac{\sum_{i=1}^n \theta_iQ\left(\frac{\bar{\theta}-\theta_i}{\sigma}\right)}{\sum_{i=1}^nQ\left(\frac{\bar{\theta}-\theta_i}{\sigma}\right)}, ~~c_2 \vcentcolon = \frac{\sum_{i=1}^n \theta_iQ^c\left(\frac{\bar{\theta}-\theta_i}{\sigma}\right)}{\sum_{i=1}^nQ^c\left(\frac{\bar{\theta}-\theta_i}{\sigma}\right)}.                                                                         
\end{equation}
\end{thm}
The proof of the theorem is given in Section \ref{subsec:thm1_proof}. 

\begin{rem}\label{note_beta_n}
 In Theorem \ref{thm1}, $\beta_n$ represents the concentrating value for the distance between $\bst$ and the attracting vector $\boldsymbol{\nu}_2$. (It is shown in Sec. \ref{subsec:thm1_proof}  that $\Vert \bst - \boldsymbol{\nu}_2 \Vert^2/n$ concentrates around $\beta_n +\kappa_n \delta$.) Therefore, the closer $\bst$ is to the attracting subspace, the lower the normalized asymptotic  risk $R( \boldsymbol{\theta}, \hat{\boldsymbol{\theta}}_{JS_{2}})/n$.   The term $\alpha_n + \sigma^2$  represents the concentrating value for the distance between $\by$ and $\boldsymbol{\nu}_2$. (It is shown in  Sec. \ref{subsec:thm1_proof} that $ \Vert \by - \boldsymbol{\nu}_2 \Vert^2/n$ concentrates around $\alpha_n + \sigma^2 + \kappa_n \delta$.)
 \end{rem}
 
 \begin{rem}
Comparing $\beta_n$ in \eqref{eq_beta2} and $\alpha_n$ in \eqref{eq_alpha2}, we note that $\beta_n \geq  \alpha_n$ because
 \be 
 \label{eq:c1minc2}
 c_1 - c_2 = \frac{- n \sum_{i=1}^n (\theta_i - \bar{\theta}) Q\left(\frac{\theta_i-\bar{\theta}}{\sigma}\right)}{\left(\sum_{i=1}^nQ\left(\frac{\theta_i-\bar{\theta}}{\sigma}\right)\right)\left( \sum_{i=1}^nQ^c\left(\frac{\theta_i-\bar{\theta}}{\sigma}\right)\right)} \geq 0. 
 \ee
To see \eqref{eq:c1minc2},  observe that in the sum in the numerator, the $Q( \cdot )$ function assigns larger weight to the terms with $(\theta_i - \bar{\theta}) <0$  than to the terms with $(\theta_i - \bar{\theta}) >0$. 

Furthermore, $\alpha_n \approx \beta_n$ for large $n$  if  either $\vert \theta_i - \bar{\theta} \vert \approx 0, \forall i$, or $\vert \theta_i - \bar{\theta} \vert \to \infty, \forall i$. In the first case,  if $\theta_i = \bar{\theta}$,  for $i = 1,\cdots,n$, we get $\beta_n = \alpha_n = \Vert \bst - \bar{\theta}\mathbf{1}  \Vert^2/n = 0$. In the second case, suppose that $n_1$ of the $\theta_i$ values equal $p_1$ and the remaining $(n-n_1)$ values equal $-p_2$ for some $p_1,p_2 >0$. Then, as $p_1, p_2 \to \infty$, it can be verified that $\beta_n \to [\Vert \bst \Vert^2 - n_1 p_1^2 - (n-n_1)p_2^2]/n = 0$. Therefore, the asymptotic normalized risk $R( \boldsymbol{\theta}, \hat{\boldsymbol{\theta}}_{JS_{2}})/n$ converges to $0$ in both cases.
\end{rem}

The proof of Theorem \ref{thm1} further leads to the following corollaries.
\begin{corr}\label{cor_pp_JS}
 The loss function of the positive-part JS-estimator in \eqref{pp_JS} satisfies the following:
 \begin{enumerate}[(1)]
 \item For any $\epsilon > 0$, 
 \begin{equation*}
 \mathbb{P}\left( \left \vert\frac{\norm{\boldsymbol{\theta} - \hat{\boldsymbol{\theta}}_{JS_{+}}}^2}{n} - \frac{\gamma_n\sigma^2}{\gamma_n+\sigma^2}  \right \vert\geq  \epsilon \right) \leq Ke^{-nk\min(\epsilon^2,1)},
 \end{equation*}
 where $  \gamma_n  \vcentcolon = {\left\Vert \boldsymbol{\theta} \right \Vert^2}/n$, and $K$ and $k$ are positive constants.
 \item For a sequence of $\bst$ with increasing dimension $n$, if $\limsup_{n \to \infty} \Vert\bst \Vert^2/n < \infty$, we have $$ \lim_{n\to \infty} \left\vert \frac{1}{n}R( \boldsymbol{\theta}, \hat{\boldsymbol{\theta}}_{JS_+}) - \frac{\gamma_n\sigma^2}{\gamma_n+\sigma^2} \right\vert = 0.$$
 \end{enumerate}
 \end{corr}

Note that the positive-part Lindley's estimator in \eqref{pp_lindley} is essentially a single-cluster estimator which shrinks all the points towards $\bar{y}$. Henceforth, we denote it by $\bsth_{JS_1}$.
\begin{corr}\label{cor_pp_JS_Lindley} 
The loss function of the positive-part Lindley's estimator in \eqref{pp_lindley} satisfies the following:
\begin{enumerate}[(1)]
 \item For any $\epsilon > 0$, 
 \begin{equation*}
 \mathbb{P}\left( \left \vert\frac{\norm{\boldsymbol{\theta} - \hat{\boldsymbol{\theta}}_{JS_1}}^2}{n} - \frac{\rho_n\sigma^2}{\rho_n+\sigma^2}  \right \vert\geq  \epsilon \right) \leq Ke^{-nk\min(\epsilon^2,1)},
 \end{equation*}
 where $K$ and $k$ are positive constants, and 
 \begin{equation}\label{eq_rho_n}
  \rho_n  \vcentcolon = \frac{\left\Vert \boldsymbol{\theta} - \bar{\theta}\mathbf{1}\right \Vert^2}{n}.
 \end{equation}
 \item For a sequence of $\bst$ with increasing dimension $n$, if $\limsup_{n \to \infty} \Vert\bst \Vert^2/n < \infty$, we have
 \begin{equation}\label{eq_corr_pp_lindley}
   \lim_{n\to \infty} \left\vert \frac{1}{n}R\left( \boldsymbol{\theta}, \hat{\boldsymbol{\theta}}_{JS_1}\right) - \frac{\rho_n\sigma^2}{\rho_n+\sigma^2} \right\vert = 0.
 \end{equation}  
 \end{enumerate}
\end{corr}

\begin{rem}
Statement $(2)$ of Corollary \ref{cor_pp_JS}, which  is known in the literature \cite{beran}, implies that $\bsth_{JS_{+}}$ is asymptotically minimax over Euclidean balls. Indeed, if $\Theta_n$ denotes  the set of $\bst$ such that $\gamma_n=\norm{\bst}^2/n \leq c^2$, then Pinsker's theorem  \cite[Ch. 5]{johnstonebook} implies that the minimax risk over $\Theta_n$ is asymptotically (as $n \to \infty$) equal to  $\frac{\sigma^2 c^2}{c^2 + \sigma^2}$. 

Statement $(1)$ of Corollary \ref{cor_pp_JS} and both the statements of Corollary \ref{cor_pp_JS_Lindley} are new, to the best of our knowledge. Comparing Corollaries \ref{cor_pp_JS} and \ref{cor_pp_JS_Lindley}, we observe that $\rho_n \leq \gamma _n$  since  $\norm{\bst -  \bar{\theta} \mbf{1}} \leq \norm{\bst}$ for all $\bst \in \mathbb{R}^n$ with strict inequality whenever $ \bar{\theta} \neq 0$. Therefore  the positive-part Lindley's estimator  asymptotically dominates the positive part JS-estimator. 
\end{rem}

\begin{figure}[t]
    \centering
    \includegraphics[width= 3.25in]{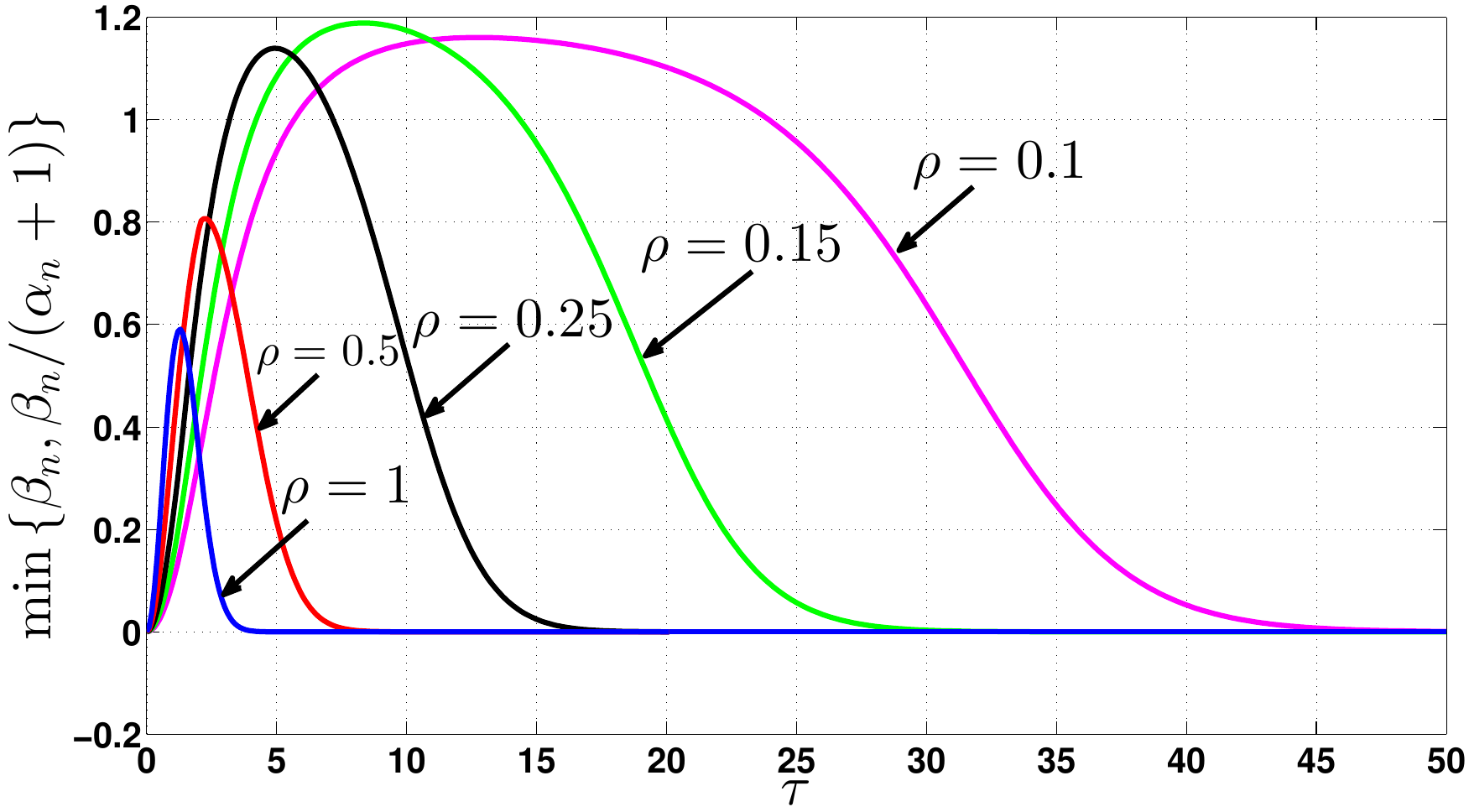}    
    \caption{\small The asymptotic risk  term $\min \{ \beta_n, \beta_n/(\alpha_n+\sigma^2) \}$ for the two-cluster estimator is plotted vs $\tau$ for $n =1000$, $\sigma=1$, and different values of $\rho$. Here, the components of $\bst$ take only two values, $\tau$ and $-\rho \tau$. The number of components taking the value $\tau$ is $\lfloor n\rho/(1+\rho) \rfloor$. }
    \label{fig_beta_n}
    \vspace{-5pt}
\end{figure} 

It is well known that both $\bsth_{JS_{+}}$ and $\bsth_{JS_{1}}$ dominate the ML-estimator \cite{baranchik}. From Corollary \ref{cor_pp_JS}, it is clear that asymptotically, the normalized risk of $\bsth_{JS_{+}}$ is small when $\gamma_n$ is small, i.e.,  when $\bst$ is close to the origin. Similarly, from Corollary \ref{cor_pp_JS_Lindley}, the asymptotic normalized risk of $\bsth_{JS_{1}}$ is small when $\rho_n$ is small, which occurs when  the components of $\bst$ are all very close to the mean $\bar{\theta}$. It is then natural to ask if the two-cluster estimator $\bsth_{JS_{2}}$ dominates $\bsth_{ML}$, and when its asymptotic normalized risk is close to $0$. To answer these questions, we use the following example, shown in Fig. \ref{fig_beta_n}.  Consider $\bst$  whose components take one of  two values, $\tau$ or $-\rho \tau$, such that $\bar{\theta}$ is as close to zero as possible. Hence the number of components taking the value $\tau$ is $\lfloor n\rho/(1+\rho) \rfloor$. Choosing $\sigma =1$, $n=1000$,   the key asymptotic risk term $\min \{ \beta_n, \beta_n/(\alpha_n+\sigma^2) \}$ in Theorem \ref{thm1} is plotted as a function of $\tau$  in Fig. \ref{fig_beta_n} for various values of $\rho$.  

Two important observations can be made from the plots. Firstly, $\min \{ \beta_n, \beta_n/(\alpha_n+\sigma^2) \}$ exceeds $\sigma^2 = 1$ for certain values of $\rho$ and $\tau$. Hence, $\bsth_{JS_{2}}$ does not dominate $\bsth_{ML}$.  Secondly, for any $\rho$, the normalized  risk of $\bsth_{JS_{2}}$ goes to zero for large enough $\tau$.  Note that when $\tau$ is large, both $\gamma_n = \Vert \bst \Vert^2/n$ and $\rho_n = \Vert \bst - \bar{\theta}\mathsf{1}\Vert^2/n$ are  large and hence, the normalized risks of both $\bsth_{JS_{+}}$ and $\bsth_{JS_{1}}$ are close to $1$. So, although $\bsth_{JS_{2}}$ does not dominate $\bsth_{ML}$, $\bsth_{JS_{+}}$ or $\bsth_{JS_{2}}$, there is a range of $\bst$ for which $R( \boldsymbol{\theta}, \hat{\boldsymbol{\theta}}_{JS_2})$ is much lower than both $R( \boldsymbol{\theta}, \hat{\boldsymbol{\theta}}_{JS_{+}})$ and $R( \boldsymbol{\theta}, \hat{\boldsymbol{\theta}}_{JS_1})$. This serves as motivation for designing a hybrid estimator that attempts to pick the better of $\bsth_{JS_1}$ and $\bsth_{JS_2}$ for the $\bst$ in context. This is described in the next section.

In the example of Fig. \ref{fig_beta_n}, it is worth examining why the two-cluster estimator performs poorly for a certain range of $\tau$, while giving significantly risk reduction for large enough $\tau$. First consider  an ideal case, where it is known which components of theta are equal to $\tau$ and which ones are equal to  $-\rho \tau$ (although the values of $\rho, \tau$ may not be known). In this case, we could use a James-Stein estimator $\bsth_{JS_{\mathbb{V}}}$ of the form \eqref{JS_attractor_gen} with the target subspace $\mathbb{V}$ being the two-dimensional subspace with basis vectors 
\[\mathbf{u}_1 \vcentcolon= \begin{bmatrix}
  \mathsf{1}_{\{\theta_1 = \tau \}}\\
  \mathsf{1}_{\{\theta_2 = \tau \}}\\
  \vdots \\
  \mathsf{1}_{\{\theta_n = \tau \}}
  \end{bmatrix}, ~~~ \mathbf{u}_2 \vcentcolon= \begin{bmatrix}
  \mathsf{1}_{\{\theta_1 = -\rho \tau \}}\\
  \mathsf{1}_{\{\theta_2 =  -\rho \tau  \}}\\
  \vdots \\
  \mathsf{1}_{\{\theta_n = -\rho \tau  \}}
  \end{bmatrix}.\] 
 Since $\mathbb{V}$ is a fixed subspace that does not depend on the data, it can be shown that   $\bsth_{JS_{\mathbb{V}}}$ dominates the ML-estimator \cite{george1,george3}.   In the actual problem, we do not have access to the ideal basis vectors $\mathbf{u}_1, \mathbf{u}_2$, so  we cannot use $\bsth_{JS_{\mathbb{V}}}$.  The two-cluster estimator $\bsth_{JS_{2}}$ attempts to approximate 
$\bsth_{JS_{\mathbb{V}}}$ by choosing the target subspace from the data. As shown in \eqref{two_partition_attractor}, this is done using the basis vectors:
\[ \widehat{\mathbf{u}}_1 \vcentcolon= \begin{bmatrix}
  \mathsf{1}_{\{y_1 \geq \bar{y} \}}\\
  \mathsf{1}_{\{ y_2 \geq  \bar{y} \}}\\
  \vdots \\
  \mathsf{1}_{\{ y_n \geq \bar{y} \}}
  \end{bmatrix}, ~~~ \widehat{\mathbf{u}}_2 \vcentcolon= \begin{bmatrix}
  \mathsf{1}_{\{ y_1 < \bar{y}   \}}\\
  \mathsf{1}_{\{ y_2 < \bar{y}  \}}\\
  \vdots \\
  \mathsf{1}_{\{ y_n < \bar{y} \}}
  \end{bmatrix}.
  \] 
Since  $\bar{y}$ is a good approximation for $\bar{\theta} =0$, when the separation between $\theta_i$ and $\bar{\theta}$ is large enough, the noise term $w_i$ is unlikely to pull $y_i = \theta_i +w_i$ into the wrong region; hence, the estimated basis vectors 
$\widehat{\mathbf{u}}_1, \widehat{\mathbf{u}}_2$ will be close to the ideal ones $\mathbf{u}_1, \mathbf{u}_2$.  Indeed, Fig. \ref{fig_beta_n} indicates that when the minimum separation between $\theta_i$ and $\bar{\theta}$ (here, equal to $\rho\tau$)  is at least $4.5\sigma$, then $\widehat{\mathbf{u}}_1, \widehat{\mathbf{u}}_2$ approximate the ideal basis vectors very well, and the normalized risk is close to $0$. On the other hand, the  approximation to the ideal basis vectors turns out to be poor when the components of $\bst$ are neither too close to nor too far from $\bar{\theta}$, as evident from Remark \ref{note_beta_n}. 

\section{Hybrid James-Stein estimator with up to two clusters} \label{sec_2_hybrid}

Depending on the underlying $\boldsymbol{\theta}$, either  the positive-part  Lindley estimator  $\hat{\boldsymbol{\theta}}_{JS_{1}}$ or the two-cluster estimator $\hat{\boldsymbol{\theta}}_{JS_{2}}$ could have a smaller loss (cf. Theorem \ref{thm1} and Corollary \ref{cor_pp_JS_Lindley}). So we would like an estimator that selects the better among $\hat{\boldsymbol{\theta}}_{JS_{1}}$ and $\hat{\boldsymbol{\theta}}_{JS_{2}}$ for the $\boldsymbol{\theta}$ in context.  To this end, we estimate the loss  of $\hat{\boldsymbol{\theta}}_{JS_{1}}$ and $\hat{\boldsymbol{\theta}}_{JS_{2}}$ based on $\mbf{y}$. Based on these loss  estimates, denoted by $\hat{L}(\boldsymbol{\theta},  \hat{\boldsymbol{\theta}}_{JS_{1}})$ and 
$\hat{L}(\boldsymbol{\theta},  \hat{\boldsymbol{\theta}}_{JS_{2}})$ respectively,
 we define a hybrid estimator as
\begin{equation}\label{comb_estimator}
  \hat{\boldsymbol{\theta}}_{JS_{H}} = 
\gamma_\by \hat{\boldsymbol{\theta}}_{JS_{1}} + (1-\gamma_\by) 
 \hat{\boldsymbol{\theta}}_{JS_{2}},
\end{equation}
where $ \hat{\boldsymbol{\theta}}_{JS_{1}}$ and $ \hat{\boldsymbol{\theta}}_{JS_{2}}$ are respectively given by \eqref{pp_lindley} and \eqref{eq_two_part_estimator}, and $\gamma_\by$ is  given by
\begin{equation}\label{eq_gamma}
 \gamma_\by = \left\{ \begin{array}{ccc}
                    1 & \textrm{if } & \frac{1}{n}\hat{L}(\boldsymbol{\theta},  \hat{\boldsymbol{\theta}}_{JS_{1}}) \leq 
                    \frac{1}{n}\hat{L}(\boldsymbol{\theta},  \hat{\boldsymbol{\theta}}_{JS_{2}}), \\
                    0 & \textrm{otherwise.}&\\
                  \end{array} \right.
\end{equation}

The loss function estimates $\hat{L}(\boldsymbol{\theta},  \hat{\boldsymbol{\theta}}_{JS_{1}})$ and $\hat{L}(\boldsymbol{\theta},  \hat{\boldsymbol{\theta}}_{JS_{2}})$ are obtained as follows. Based on Corollary \ref{cor_pp_JS_Lindley},  the loss function of $\hat{\boldsymbol{\theta}}_{JS_{1}}$ can be estimated via an estimate of $\rho_n \sigma^2/(\rho_n + \sigma^2)$, where $\rho_n$ is given by \eqref{eq_rho_n}. It  is straightforward to check, along the lines of the proof of Theorem \ref{thm1}, that 
\begin{align}
  g\left(\frac{\left\Vert \mathbf{y} - \bar{y}\mathbf{1} \right\Vert^2}{n}\right)  \doteq g\left(\rho_n + \sigma^2\right) = \rho_n + \sigma^2.
  \label{eq:gyybar_conv}
\end{align}
 Therefore, an estimate of the normalized loss $L(\boldsymbol{\theta},  \hat{\boldsymbol{\theta}}_{JS_{1}})/n$ is
\begin{equation}\label{eq_est_1_partition_risk}
 \frac{1}{n}\hat{L}(\boldsymbol{\theta},  \hat{\boldsymbol{\theta}}_{JS_{1}}) = \sigma^2\left(1-\frac{ \sigma^2 }{g\left({\left\Vert \mathbf{y} - \bar{y}\mathbf{1} \right\Vert^2}/{n}\right) }\right). 
\end{equation}

 The loss function of the two-cluster estimator  $\hat{\boldsymbol{\theta}}_{JS_{2}}$ can be estimated using Theorem  \ref{thm1}, by estimating  $\beta_n$ and $\alpha_n$  defined in \eqref{eq_alpha2} and \eqref{eq_beta2}, respectively. From Lemma \ref{lem3} in Section \ref{subsec:thm1_proof}, we have 
\begin{equation}\label{eq_alpha_n_est}
 \frac{1}{n}\left\Vert \mathbf{y} - \boldsymbol{\nu}_{2} \right\Vert^2 \doteq \alpha_n + \sigma^2 + \kappa_n\delta + o(\delta).
\end{equation}
Further, using the concentration inequalities in Lemmas \ref{prop_y1_y2}  and \ref{prop_bias} in Section \ref{sec_2_attractor}, we can deduce that 
\begin{align}\nonumber
 &\frac{1}{n}\left\Vert \mathbf{y} - \boldsymbol{\nu}_{2} \right\Vert^2 - \sigma^2 + \frac{\sigma^2}{n\delta}\left(\sum_{i=0}^{n} \mathbf{1}_{\left\{\left\vert y_i - \bar{y}\right\vert \leq \delta \right\}}\right)(a_1 - a_2) \\ \label{eq_beta_n_est}
 & \doteq \beta_n + \kappa_n\delta + o(\delta),
 \end{align}
where $a_1, a_2$ are defined in \eqref{eq_attractor1}.  We now use \eqref{eq_alpha_n_est} and \eqref{eq_beta_n_est} to estimate 
the concentrating value in \eqref{eq1_thm1_statement}, noting that 
\[ \min\left( \beta_n, \frac{\beta_n\sigma^2}{\alpha_n+\sigma^2} \right) =   \frac{\beta_n\sigma^2}{g(\alpha_n+\sigma^2)},\]
where $g(x)=\max(x, \sigma^2)$. This yields the following estimate of $L(\boldsymbol{\theta},  \hat{\boldsymbol{\theta}}_{JS_{2}})/n$:
\begin{align} \nonumber
 & \frac{1}{n}\hat{L}(\boldsymbol{\theta},  \hat{\boldsymbol{\theta}}_{JS_{2}}) = \\ \label{eq_est_2_partition_risk}
 & \frac{\sigma^2\left( \frac{1}{n}\left\Vert \mathbf{y} - \boldsymbol{\nu}_{2} \right\Vert^2 - \sigma^2 + \frac{\sigma^2}{n\delta}(a_1 - a_2) \sum_{i=0}^{n} \mathbf{1}_{\left\{\left\vert y_i - \bar{y}\right\vert \leq \delta \right\}}\right)}{ g(\norm{\mbf{y} - \boldsymbol{\nu}_2}^2/n)} . 
\end{align}
The loss function estimates in \eqref{eq_est_1_partition_risk} and \eqref{eq_est_2_partition_risk} complete the specification of the hybrid estimator in \eqref{comb_estimator} and \eqref{eq_gamma}.  The following theorem characterizes the loss function of the hybrid estimator, by showing that the loss estimates  in \eqref{eq_est_1_partition_risk} and \eqref{eq_est_2_partition_risk} concentrate around the values specified in Corollary \ref{cor_pp_JS_Lindley} and  Theorem \ref{thm1}, respectively.  
\begin{thm}\label{thm2}
The loss function of the hybrid JS-estimator in \eqref{comb_estimator} satisfies the following: 
 \begin{enumerate}[(1)]
  \item For any $\epsilon > 0$, 
 \begin{align*}
 & \mathbb{P}\left( \frac{\norm{\boldsymbol{\theta} - \hat{\boldsymbol{\theta}}_{JS_{H}}}^2}{n} - \min\left(  \frac{\norm{\boldsymbol{\theta} - \hat{\boldsymbol{\theta}}_{JS_{1}}}^2}{n}, \frac{\norm{\boldsymbol{\theta} - \hat{\boldsymbol{\theta}}_{JS_{2}}}^2}{n}\right) \right. \\
 & \hspace{0.2in} \geq  \epsilon \Bigg) \leq Ke^{-\frac{nk \min(\epsilon^2,1)}{\max(\Vert \bst\Vert^2/n,1)}},
 \end{align*}
 where $K$ and $k$ are positive constants.
 \item For a sequence of $\bst$ with increasing dimension $n$, if $\limsup_{n \to \infty} \Vert\bst \Vert^2/n < \infty$, we have
 \begin{align*}
   &\limsup_{n\to \infty} \frac{1}{n} \left( R\left( \boldsymbol{\theta}, \hat{\boldsymbol{\theta}}_{JS_{H}}\right) \right.\\ &\hspace{0.7in}-\min\left[R\left( \boldsymbol{\theta}, \hat{\boldsymbol{\theta}}_{JS_{1}}\right), R\left( \boldsymbol{\theta}, \hat{\boldsymbol{\theta}}_{JS_{2}}\right) \right] \bigg ) \leq  0.
 \end{align*}
\end{enumerate}
\end{thm}
The proof of the theorem in given in Section \ref{subsec:thm2_proof}. The theorem implies that the hybrid estimator chooses the better of the $\bsth_{JS_1}$ and $\bsth_{JS_2}$ with high probability, with the probability of choosing the worse estimator decreasing exponentially in $n$. It also implies that asymptotically, $ \hat{\boldsymbol{\theta}}_{JS_{H}}$ dominates both $ \hat{\boldsymbol{\theta}}_{JS_{1}}$ and $ \hat{\boldsymbol{\theta}}_{JS_{2}}$, and hence, $ \hat{\boldsymbol{\theta}}_{ML}$ as well. 

\begin{rem}
Instead of picking one among the two (or several) candidate estimators, one could consider a hybrid estimator which is a weighted  combination of the candidate estimators.  Indeed, George \cite{george1} and Leung and Barron \cite{LeungBarron06,LeungThesis} have proposed combining the estimators using exponential mixture weights based on Stein's unbiased risk estimates (SURE) \cite{stein2}. Due to the presence of indicator functions in the definition of the attracting vector, it is challenging to obtain a SURE for $\bsth_{JS_{2}}$. We therefore  use loss estimates to choose the better estimator. Furthermore, instead of choosing one estimator based on the loss estimate,  if we were to follow the approach in \cite{LeungBarron06} and employ a combination of the estimators using exponential mixture weights based on the un-normalized loss estimates, then the weight assigned to the estimator with the smallest loss estimate is exponentially larger (in $n$) than the other. Therefore, when the dimension is high, this is effectively equivalent to picking the estimator with the smallest loss estimate.
\end{rem}

\section{General multiple-cluster James-Stein estimator}\label{sec_multi_attractor}

In this section, we generalize the two-cluster estimator of Section \ref{sec_2_attractor} to an $L$-cluster estimator defined by an arbitrary partition of the real line.  The partition is defined by $L-1$ functions $s_j \vcentcolon \mathbb{R}^{n} \to \mathbb{R}$, such that 
\begin{equation}\label{eq_multi_part_points}
 s_j(\mathbf{y})\vcentcolon= s_j\left(y_1,\cdots,y_n \right) 	\doteq \mu_j, \quad  \forall j =1,\cdots,(L-1),
\end{equation}
 with constants $\mu_1 > \mu_2 > \cdots > \mu_{L-1}$. In words, the partition can be defined via any $L-1$ functions of $\mbf{y}$,  each of which concentrates around a deterministic value as $n$ increases. In the two-cluster estimator, we only have one function $s_1(\mbf{y}) = \bar{y}$, which concentrates around $\bar{\theta}$. The points in \eqref{eq_multi_part_points} partition the real line as
 \begin{equation*}
 \mathbb{R} =  \left(-\infty,s_{L-1}(\mathbf{y}) \right] \cup \left(s_{L-1}(\mathbf{y}), s_{L-2}(\mathbf{y})\right]  \cup \cdots \cup  \left(s_1(\mathbf{y}), \infty\right).
\end{equation*}

 The clusters are defined as $\mc{C}_j = \{ y_i, \ 1 \leq i \leq n \mid y_i \in (s_{j}(\mathbf{y}), s_{j-1}(\mathbf{y})] \}$, for $1\leq j \leq L$, with $s_{0}(\mathbf{y}) = \infty$ and $s_{L}(\mathbf{y}) = -\infty$. In Section \ref{subsec:obt_clusters}, we discuss one choice of partitioning points to define the $L$ clusters, but here we first construct and analyse an estimator based on a general partition satisfying   \eqref{eq_multi_part_points}.

 The points  in $\mc{C}_j$ are all shrunk towards the same point $a_j$, defined in \eqref{eq:a1aL} later in this section. The attracting vector  is 
 \begin{equation} \label{eq_mult_part_attractor}
\begin{split}
 \boldsymbol{\nu}_{L} & \vcentcolon= \sum_{j=1}^{L}a_j\left[ \begin{array}{c}
  \mathsf{1}_{\{y_1 \in \mathcal{C}_j \}}\\
  \vdots \\
  \mathsf{1}_{\{y_n \in \mathcal{C}_j\}}
  \end{array}
 \right],
\end{split}
\end{equation}
and the proposed $L$-cluster JS-estimator is 
 \begin{equation} \label{eq_mult_part_estimator}
   \hat{\boldsymbol{\theta}}_{JS_{L}} = \boldsymbol{\nu}_{L}   + \left[1 - \frac{\sigma^2}
   {g\left({ \| \mathbf{y} - \boldsymbol{\nu}_{2} \|^2}/{n}\right)} \right] \left( \mathbf{y} - \boldsymbol{\nu}_{L}  \right ),
\end{equation}
where $g(x)=\max(\sigma^2,x)$.

 The attracting vector $\boldsymbol{\nu}_{L}$ lies in an $L$-dimensional subspace defined by the $L$ orthogonal vectors 
$[\mathsf{1}_{\{y_1 \in \mc{C}_1 \}},  \cdots,  \mathsf{1}_{\{y_n \in \mc{C}_1 \}} ]^T, \ldots, $ 
 $[\mathsf{1}_{\{y_1  \in \mc{C}_L \}},  \cdots,  \mathsf{1}_{\{y_n  \leq \mc{C}_L\}} ]^T$.  The desired values for $a_1, \ldots,a_L$ in \eqref{eq_mult_part_attractor} are such that the attracting vector $ \boldsymbol{\nu}_{L}$ is the projection of $\bst$ onto the $L$-dimensional subspace. Computing this projection, we find  the desired values to be the means of the $\theta_i$'s in each cluster:
 \be
 a^{des}_1 =  \frac{\sum_{i=1}^n\theta_i \mathsf{1}_{\{y_i \in  \mc{C}_1 \}}}{\sum_{i=1}^n\mathsf{1}_{\{y_i \in \mc{C}_1 \}}},  \ldots
 , a^{des}_L= \frac{\sum_{i=1}^n\theta_i \mathsf{1}_{\{y_i \in  \mc{C}_L \}}}{\sum_{i=1}^n\mathsf{1}_{\{y_i \in \mc{C}_L \}}}.  \label{eq:a1aL_des}
 \ee
 As the $\theta_i$'s are unavailable, we set $a_1,\ldots, a_L$ to be  approximations  of $a^{des}_1, \ldots,  a^{des}_L$,  obtained using concentration results similar to Lemmas \ref{prop_y1_y2} and \ref{prop_bias} in Section \ref{sec_2_attractor}. The attractors are given by
 \begin{equation}
\begin{split}
 {a}_j &= \frac{\sum_{i=1}^n y_i\mathsf{1}_{\left\{ y_i \in \mathcal{C}_j\right\}} }{\sum_{i=1}^n \mathsf{1}_{\left\{y_i \in \mathcal{C}_j\right\}}} \\
 &\quad-\frac{\frac{\sigma^2}{2\delta}\sum_{i=0}^{n} \left[\mathsf{1}_{\left\{\vert y_i - s_j(\mathbf{y}) \vert \leq \delta \right\}} - \mathsf{1}_{\left\{\vert y_i - s_{j-1}(\mathbf{y}) \vert \leq \delta \right\}}  \right]}{\sum_{i=1}^n \mathsf{1}_{\left\{y_i \in \mathcal{C}_j\right\}}}
\end{split}
\label{eq:a1aL}
\end{equation} 
for $1\leq j \leq L$. With $\delta >0$ chosen to be a small positive number as before, this completes the specification of the attracting vector in \eqref{eq_mult_part_attractor}, and hence the $L$-cluster JS-estimator in \eqref{eq_mult_part_estimator}.
\begin{thm}\label{thm3}
The loss function of the $L$-cluster JS-estimator in \eqref{eq_mult_part_estimator}
satisfies the following:
 \begin{enumerate}[(1)] 
 \item For any $\epsilon > 0$, 
 \begin{align*}
 &\mathbb{P}\left( \left \vert\frac{1}{n}\norm{\boldsymbol{\theta} - \hat{\boldsymbol{\theta}}_{JS_{L}}}^2  - \min\left(\beta_{n,L},\frac{\beta_{n,L}\sigma^2}{\alpha_{n,L}+\sigma^2}\right)\right. \right.\\
 &\hspace{0.2in} -  \kappa_n\delta + o(\delta) \bigg\vert\geq  \epsilon \Bigg) \leq Ke^{-\frac{nk \min(\epsilon^2,1)}{\max(\Vert \bst\Vert^2/n,1)}},
 \end{align*}
 where $K$ and $k$ are positive constants, and 
 \begin{align} 
 \nonumber 
\alpha_{n,L}   &\vcentcolon= \frac{\Vert \boldsymbol{\theta} \Vert^2}{n}  \\ \nonumber
 &- \sum_{j=0}^L\frac{ c_j^2}{n}\sum_{i=1}^n\left[Q\left(\frac{\mu_{j}-\theta_i}{\sigma}\right) -Q\left(\frac{\mu_{j-1}-\theta_i}{\sigma}\right)\right]  \\ 
 \label{eq_m1_thm3}
&- \frac{2\sigma}{n\sqrt{2\pi}} \sum_{j=1}^Lc_j\sum_{i=1}^n\left[e^{-\frac{\left(\mu_{j}-\theta_i\right)^2}{2\sigma^2}}-e^{-\frac{\left(\mu_{j-1}-\theta_i\right)^2}{2\sigma^2}}\right],\\
 \nonumber
\beta_{n,L} & \vcentcolon= \frac{\Vert \boldsymbol{\theta} \Vert^2}{n} \\ \label{eq_m2_thm3}
&- \sum_{j=0}^L\frac{c_j^2}{n}\sum_{i=1}^n\left[Q\left(\frac{\mu_{j}-\theta_i}{\sigma}\right) -Q\left(\frac{\mu_{j-1}-\theta_i}{\sigma}\right)\right],
\end{align}
with 
\begin{equation}\label{eq_cj}
 c_j \vcentcolon= \frac{\sum_{i=1}^n \theta_i \left[Q\left(\frac{\mu_j-\theta_i}{\sigma}\right) -Q\left(\frac{\mu_{j-1}-\theta_i}{\sigma}\right)\right] }{\sum_{i=1}^n\left[Q\left(\frac{\mu_j-\theta_i}{\sigma}\right) -Q\left(\frac{\mu_{j-1}-\theta_i}{\sigma}\right)\right]}, 
\end{equation}
for $ 1\leq j \leq L$.
 \item For a sequence of $\bst$ with increasing dimension $n$, if $\limsup_{n \to \infty} \Vert\bst \Vert^2/n < \infty$, we have
 \begin{align*}
   &\lim_{n\to \infty} \left\vert\frac{1}{n}R\left( \boldsymbol{\theta}, \hat{\boldsymbol{\theta}}_{JS_L}\right)- \min\left(\beta_{n,L},\frac{\sigma^2\beta_{n,L}}{\alpha_{n,L}+\sigma^2}\right) \right. \\
   & \hspace{0.4in} - \kappa_n\delta + o(\delta)\bigg \vert = 0. 
 \end{align*} 
 \end{enumerate} 
\end{thm}
\noindent The proof is similar to that of Theorem \ref{thm1}, and we provide its sketch in Section \ref{subsec:thm3_proof}.
 
 \begin{figure}[t]
    \centering
    \begin{subfloat}[\label{fig_beta_n_4}]{
       \includegraphics[width= 3.25in]{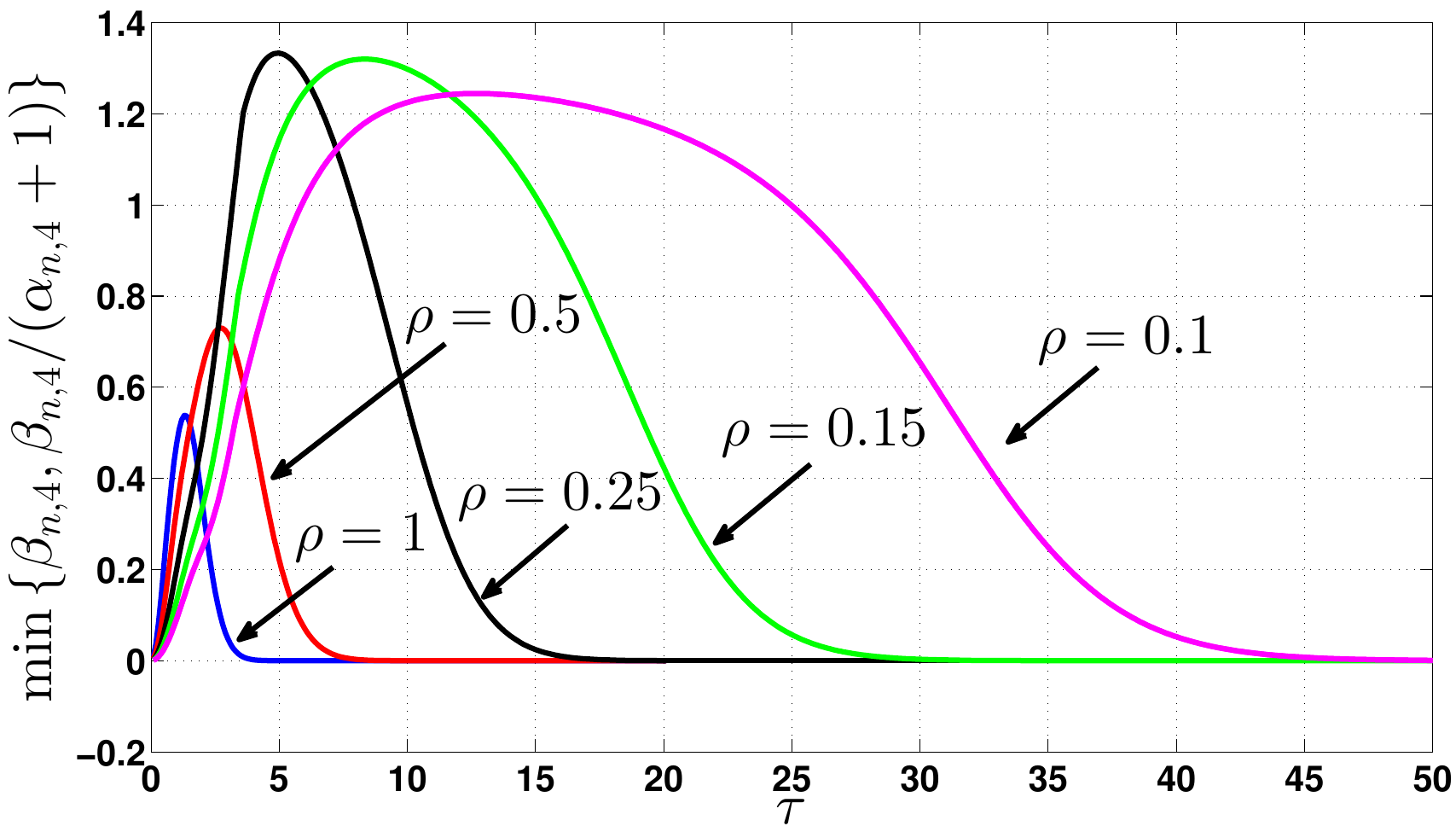}   
       }
    \end{subfloat}
    \begin{subfloat}[\label{fig_beta_n_4_1}]{
       \includegraphics[width= 3.25in]{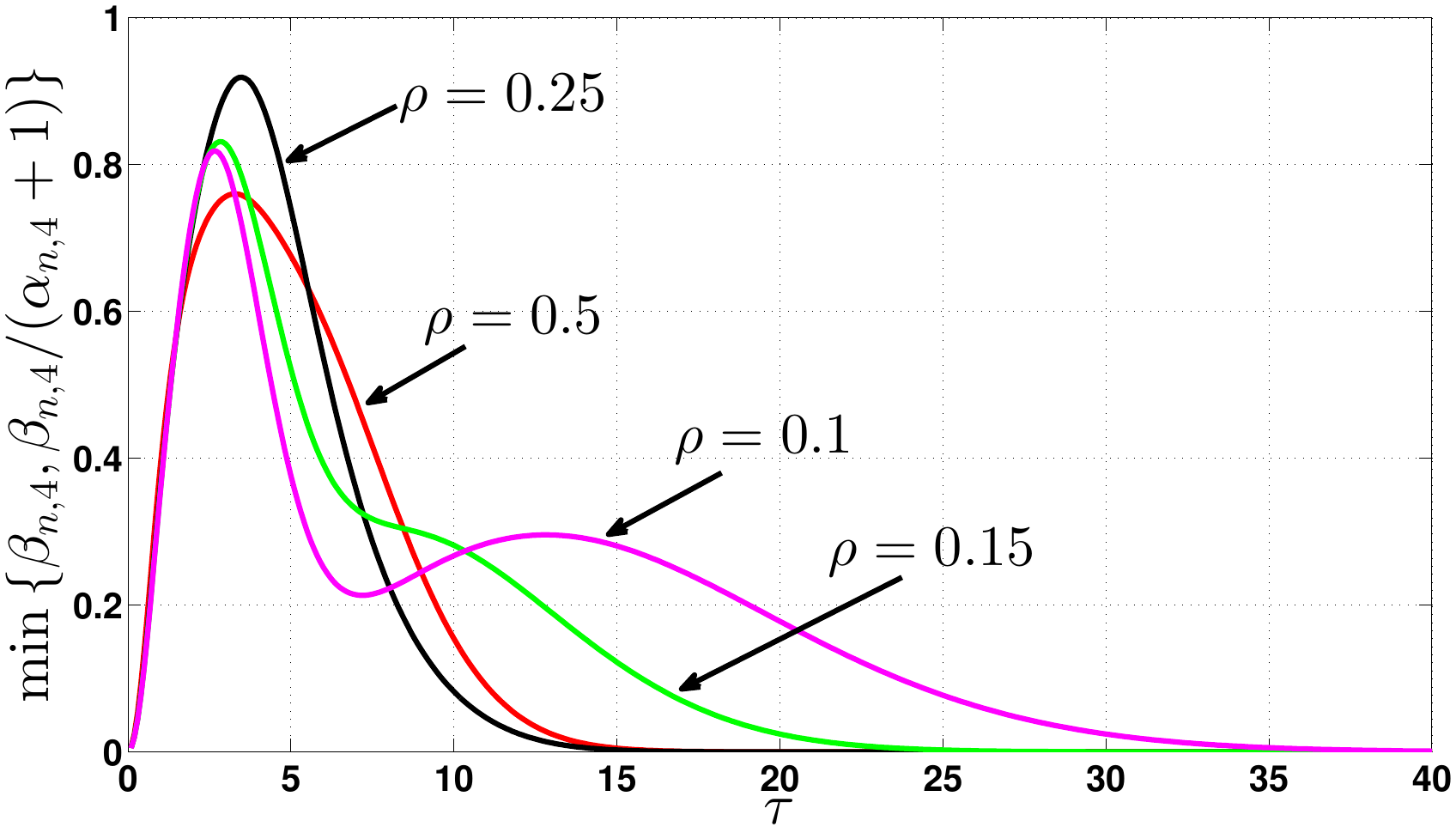}
       }
    \end{subfloat}    
    \caption{\small The asymptotic risk term  $\min\{ \beta_n,  \beta_{n,4}/(\alpha_{n,4}+1) \}$ for the four-cluster estimator is plotted vs $\tau$, for $n =1000, \sigma=1$, and different values of $\rho$. In (a), the components of $\bst$ take only two values $\tau$ and $-\rho \tau$, with $\lfloor 1000\rho/(1+\rho) \rfloor$ components taking the value $\tau$. In (b), they take values from the set $\{\tau, \rho \tau, -\rho \tau, -\tau\}$ with equal probability.}
    \label{fig:fig_beta_n_4_main}
\end{figure}

To get intuition on how the asymptotic normalized risk $R(\boldsymbol{\theta},\hat{\boldsymbol{\theta}}_{JS_L})/n$ depends on $L$, consider the four-cluster estimator with $L =4$. For the same setup as in Fig. \ref{fig_beta_n}, i.e., the components of $\bst$ take one of two values: $\tau$ or $-\rho \tau$,  Fig. \ref{fig_beta_n_4}  plots the asymptotic risk term $\min\{ \beta_n, \beta_{n,4}/(\alpha_{n,4}+1) \}$ versus $\tau$ for the four-cluster estimator. Comparing  Fig. \ref{fig_beta_n_4} and Fig. \ref{fig_beta_n}, we observe that the four-cluster estimator's risk $\min\{\beta_n, \beta_{n,4}/(\alpha_{n,4}+1)\}$ behaves similarly to the two-cluster estimator's risk  $\min\{ \beta_n, \beta_{n,2}/(\alpha_{n,2}+1) \}$ with the notable difference being the magnitude. For $\rho =0.5,1$, the peak value of $\beta_{n,4}/(\alpha_{n,4}+1)$ is smaller than that of  $\beta_{n,2}/(\alpha_{n,2}+1)$. However, for the smaller values of $\rho$, the reverse is true. This means that $\bsth_{JS_4}$ can be better than $\bsth_{JS_2}$,  even in certain scenarios where the $\theta_i$ take only two values. In the two-value example,  $\bsth_{JS_4}$ is typically better when  two of the four attractors of  $\bsth_{JS_4}$ are closer to the $\theta_i$ values, while the two attracting points of $\bsth_{JS_2}$ are closer to $\bar{\theta}$ than the respective $\theta_i$ values.

Next consider an example where $\theta_i$ take values from $\{\tau, \rho \tau, -\rho \tau, -\tau\}$ with equal probability. This is the scenario favorable to $\bsth_{JS_4}$. Figure \ref{fig_beta_n_4_1} shows the plot of $\min\{ \beta_n, \beta_{n,4}/(\alpha_{n,4}+1) \}$ as a function of $\tau$ for different values of $\rho$. Once again, it is clear that when the separation between the points is large enough, the asymptotic normalized risk approaches $0$.

\subsection{$L$-hybrid James-Stein estimator}
 
 Suppose that we have estimators $\hat{\boldsymbol{\theta}}_{JS_1}, \ldots, \hat{\boldsymbol{\theta}}_{JS_L}$, where 
 $\hat{\boldsymbol{\theta}}_{JS_\ell}$ is an $\ell$-cluster JS-estimator  constructed as described above, for $\ell=1, \ldots, L$. (Recall that  
 $\ell=1$ corresponds to Lindley's positive-part estimator in \eqref{pp_lindley}.) Depending on $\bst$,  any one of these $L$ estimators could achieve the smallest loss.  We  would like to design a hybrid estimator that picks the best of these $L$ estimators for the $\bst$ in context. As in Section \ref{sec_2_hybrid}, we construct loss estimates  for each of the $L$ estimators, and define a hybrid estimator as 
 \begin{equation}\label{comb_L_estimator}
  \hat{\boldsymbol{\theta}}_{JS_{H,L}} = 
\sum_{\ell=1}^L\gamma_\ell \, \hat{\boldsymbol{\theta}}_{JS_\ell} 
\end{equation}
where 
\begin{equation*}
 \gamma_\ell = \left\{ \begin{array}{ll}
                    1 & \textrm{ if }  \frac{1}{n}\hat{L}(\boldsymbol{\theta},  \hat{\boldsymbol{\theta}}_{JS_{\ell}}) = 
                    \min_{1 \leq k \leq L} \ \frac{1}{n}\hat{L}(\boldsymbol{\theta},  \hat{\boldsymbol{\theta}}_{JS_k} ), \\
                    0  & \textrm{ otherwise}\\
                  \end{array} \right.
\end{equation*}
 with $\hat{L}(\boldsymbol{\theta},  \hat{\boldsymbol{\theta}}_{JS_{\ell}})$ denoting the loss function estimate of 
 $\hat{\boldsymbol{\theta}}_{JS_{\ell}}$. 

For $\ell \geq 2$, we estimate the loss of  $\hat{\boldsymbol{\theta}}_{JS_{\ell}}$ using Theorem  \ref{thm3}, by estimating  $\beta_{n,\ell}$ and $\alpha_{n,\ell}$ which are defined in \eqref{eq_m1_thm3} and \eqref{eq_m2_thm3}, respectively. 
From \eqref{eq_thm3} in Section \ref{subsec:thm3_proof}, we obtain 
\begin{equation}\label{eq_alpha_nL_est}
 \frac{1}{n}\left\Vert \mathbf{y} - \boldsymbol{\nu}_{\ell} \right\Vert^2 \doteq \alpha_{n,\ell} + \sigma^2 + \kappa_n\delta + o(\delta).\end{equation}
Using concentration inequalities similar to those in Lemmas \ref{prop_y1_y2}  and \ref{prop_bias} in Section \ref{sec_2_attractor}, we deduce that 
\begin{align} \nonumber
&\frac{1}{n}\left\Vert \mathbf{y} - \boldsymbol{\nu}_{\ell} \right\Vert^2 - \sigma^2 + \frac{\sigma^2}{n\delta}\left(\sum_{j=1}^\ell a_j\sum_{i=0}^{n}   \big[\mathsf{1}_{\left\{\left\vert y_i - s_j(\mathbf{y})\right\vert \leq \delta \right\}} \right.\\ \label{eq_beta_nL_est}
&-  \mathsf{1}_{\left\{\left\vert y_i - s_{j-1}(\mathbf{y})\right\vert \leq \delta \right\}}\big]\Bigg) \doteq \beta_{n,\ell} + \kappa_n\delta +  o(\delta),
 \end{align}
where $a_1, \ldots,a_\ell$ are as defined in \eqref{eq:a1aL}.  We now use \eqref{eq_alpha_nL_est} and \eqref{eq_beta_nL_est} to estimate 
the concentrating value in Theorem \ref{thm3}, and thus obtain the following estimate of 
$L(\boldsymbol{\theta},  \hat{\boldsymbol{\theta}}_{JS_{\ell}})/n$:
\begin{equation}
\label{eq_est_L_partition_risk}
\begin{split}
&  \frac{1}{n}\hat{L}(\boldsymbol{\theta},  \hat{\boldsymbol{\theta}}_{JS_{\ell}})   = 
 \frac{\sigma^2}{g( \| \mbf{y} - \boldsymbol{\nu}_l \|^2/n)}  
 \Bigg( \frac{1}{n}  \| \mathbf{y} - \boldsymbol{\nu}_{\ell} \|^2 - \sigma^2  \\
&  \left. \ \  + \frac{\sigma^2}{n\delta}\sum_{j=1}^\ell a_j\sum_{i=0}^{n}  \left[ \mathsf{1}_{\left\{\left\vert y_i - s_j(\mathbf{y})\right\vert \leq \delta \right\}} -  \mathsf{1}_{\left\{\left\vert y_i - s_{j-1}(\mathbf{y})\right\vert \leq \delta \right\}}\right]  \right).  
 \end{split}
 \end{equation}

The loss function estimator in \eqref{eq_est_L_partition_risk} for $2 \leq \ell \leq L$,  together with the loss function estimator in \eqref{eq_est_1_partition_risk} for $\ell=1$, completes the specification of the $L$-hybrid estimator in \eqref{comb_L_estimator}.   Using steps similar to those in Theorem \ref{thm2}, we can show that 
\begin{equation}\label{eq_est_L_partition_risk1}
 \frac{1}{n}\hat{L}\left(\boldsymbol{\theta},  \hat{\boldsymbol{\theta}}_{JS_{\ell}}\right) \doteq \min\left(\beta_{n,\ell}, \frac{\sigma^2\beta_{n,\ell}}{\alpha_{n,\ell} + \sigma^2}\right) + \kappa_n\delta +  o(\delta),
\end{equation}
for $ 2 \leq \ell \leq L$.
\begin{thm}\label{thm_hybrid_partition}
The loss function of the $L$-hybrid JS-estimator in  \eqref{comb_L_estimator} satisfies the following: 
 \begin{enumerate}[(1)]
 \item For any $\epsilon > 0$, 
 \begin{align*}
 & \mathbb{P}\left( \frac{\norm{\boldsymbol{\theta} - \hat{\boldsymbol{\theta}}_{JS_{H,L}}}^2}{n} - \min_{1\leq \ell \leq L}\left(\frac{\norm{\boldsymbol{\theta} - \hat{\boldsymbol{\theta}}_{JS_{\ell}}}^2}{n}\right)\geq  \epsilon \right) \\
 &\quad \leq Ke^{-\frac{nk \min(\epsilon^2,1)}{\max(\Vert \bst\Vert^2/n,1)}},
 \end{align*}
 where $K$ and $k$ are positive constants.
 \item For a sequence of $\bst$ with increasing dimension $n$, if $\limsup_{n \to \infty} \Vert\bst \Vert^2/n < \infty$, we have
 \[   \limsup_{n\to \infty} \frac{1}{n} \left[ R\left( \boldsymbol{\theta}, \hat{\boldsymbol{\theta}}_{JS_{H,L}}\right)- \min_{1 \leq \ell \leq L}R\left( \boldsymbol{\theta}, \hat{\boldsymbol{\theta}}_{JS_{\ell}}\right) \right ] \leq 0.\]
 \end{enumerate}
 \end{thm}
  The proof of the theorem is omitted as it is along the same lines as the proof of Theorem \ref{thm3}. Thus with high probability, the $L$-hybrid estimator chooses the best  of the $\bsth_{JS_1}, \ldots, \bsth_{JS_L}$,  with the probability of choosing a worse estimator decreasing exponentially in $n$.

  \subsection{Obtaining the clusters} \label{subsec:obt_clusters}

 In this subsection, we present a simple method to obtain the $(L-1)$ partitioning points $s_j(\mathbf{y})$, $1\leq j \leq (L-1)$, for an $L$-cluster JS-estimator when $L = 2^a$ for an integer $a > 1$. We do this recursively, assuming that we already have a $2^{a-1}$-cluster estimator with its associated  partitioning points $s^\prime_j(\mathbf{y})$, $j=1,\cdots,2^{a-1}-1$.
 This means that for the $2^{a-1}$-cluster estimator, the real line is partitioned as
\begin{align*}
 \mathbb{R} = & \left(-\infty,s^\prime_{2^{a-1}-1}(\mathbf{y}) \right] \cup \left(s^\prime_{2^{a-1}-1}(\mathbf{y}), s^\prime_{2^{a-1}-2}(\mathbf{y})\right]  \cup \cdots \\
 & \quad \cup  \left(s^\prime_1(\mathbf{y}), \infty\right).
\end{align*}
Recall that Section \ref{sec_2_attractor} considered the case of $a=1$, with the single partitioning point being $\bar{y}$. 

The  new partitioning points $s_k(\mathbf{y})$, $k=1,\cdots,(2^a-1)$, are obtained as follows. For $j =1 ,\cdots, (2^{a-1}-1)$, define
\begin{align*}
 s_{2j}(\mathbf{y}) & = s^\prime_{j}(\mathbf{y}),   
 s_{2j-1}(\mathbf{y}) = \frac{\sum_{i=1}^n y_i \mathsf{1}_{\left\{s^\prime_{j}(\mathbf{y}) < y_i \leq s^\prime_{j-1}(\mathbf{y})\right \}}}{\sum_{i=1}^n\mathsf{1}_{\left\{s^\prime_{j}(\mathbf{y}) < y_i \leq s^\prime_{j-1}(\mathbf{y})\right \}}}, \\
 s_{2^{a}-1}(\mathbf{y}) & = \frac{\sum_{i=1}^n y_i \mathsf{1}_{\left\{-\infty  < y_i \leq s^\prime_{2^{a-1}-1}(\mathbf{y})\right \}}}{\sum_{i=1}^n\mathsf{1}_{\left\{-\infty < y_i \leq s^\prime_{2^{a-1}-1}(\mathbf{y})\right \}}}
\end{align*}
where $s^\prime_{0}(\mbf{y}) = \infty$. Hence, the partition for the $L$-cluster estimator is 
\begin{align*}
 \mathbb{R} = & \left(-\infty,s_{2^{a}-1}(\mathbf{y}) \right] \cup \left(s_{2^{a}-1}(\mathbf{y}), s_{2^{a}-2}(\mathbf{y})\right]  \cup \cdots \\
 &\cup  \left(s_1(\mathbf{y}), \infty\right).
\end{align*}
We use such a partition to construct a $4$-cluster estimator for our simulations in the next section.


\section{Simulation Results} \label{sec:simulations}

 \begin{figure}[p]
    \centering
    \begin{subfloat}[\label{subfig_n_10_0_5_50}]{
       \includegraphics[width= 3.25in]{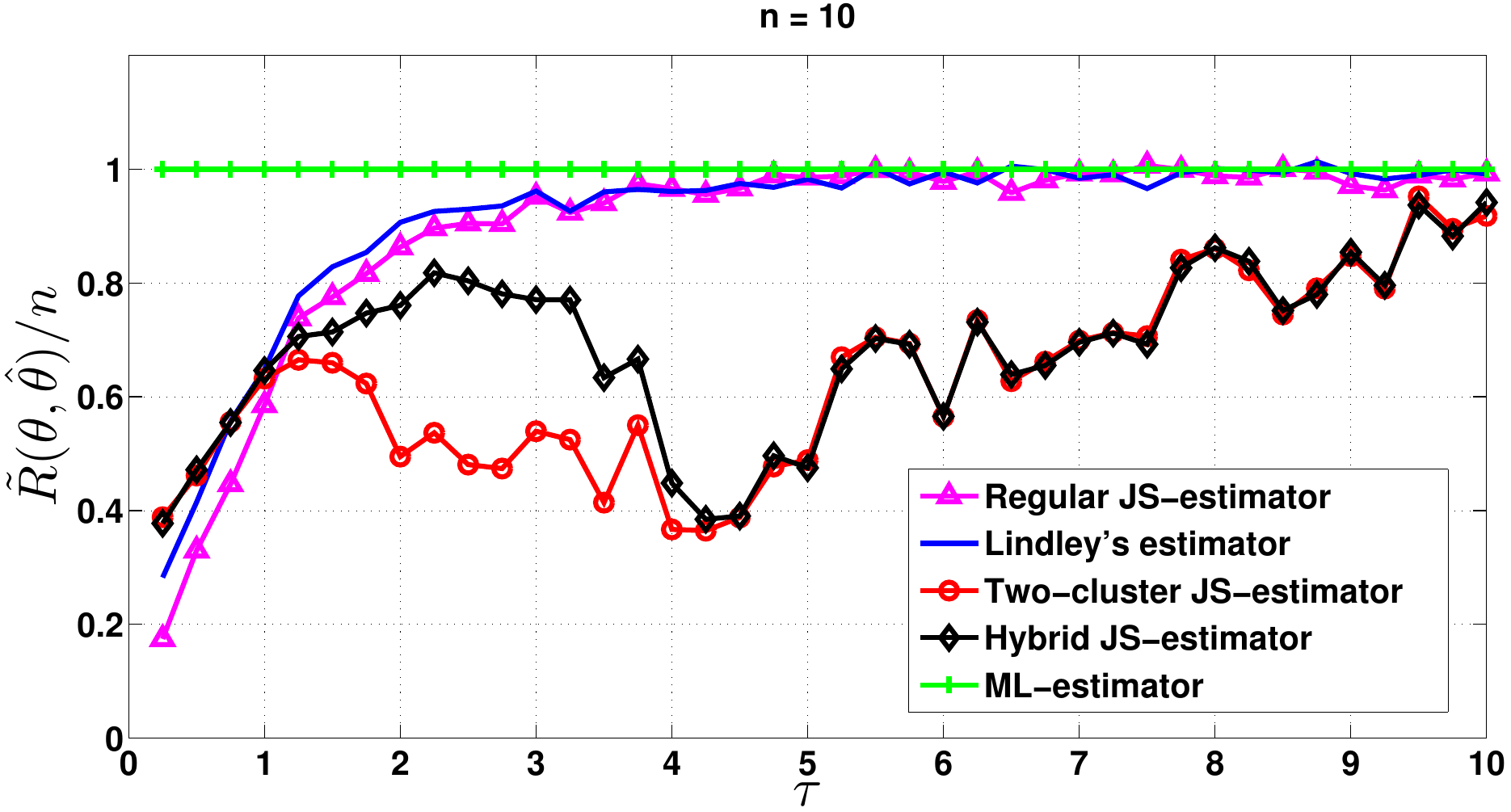}   
       }
    \end{subfloat}
  \quad
    \begin{subfloat}[\label{subfig_n_50_0_5_50}]{
       \includegraphics[width= 3.25in]{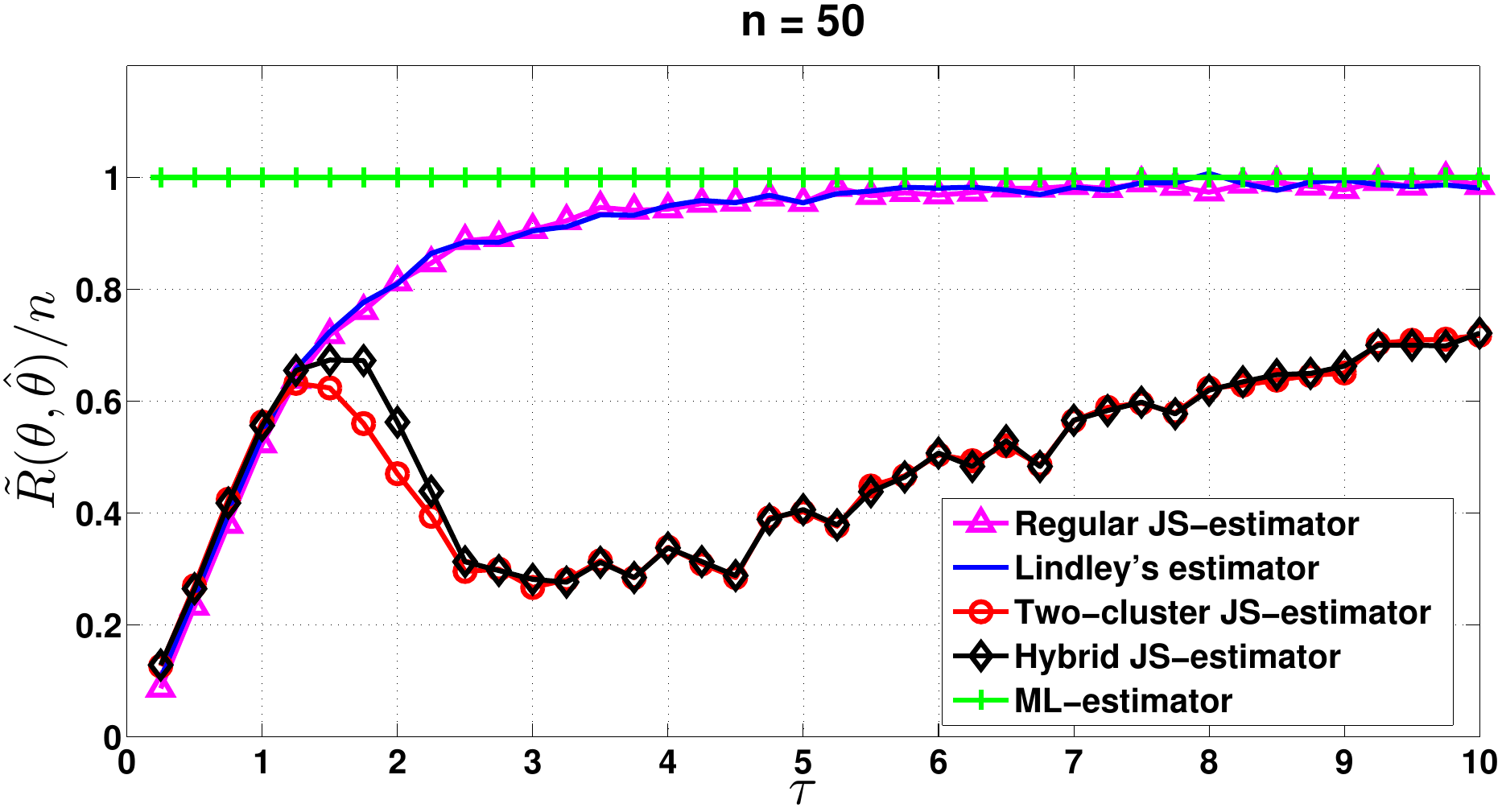}
       }
    \end{subfloat}
    \quad
    \begin{subfloat}[\label{subfig_n_100_0_5_30}]{
       \includegraphics[width= 3.25in]{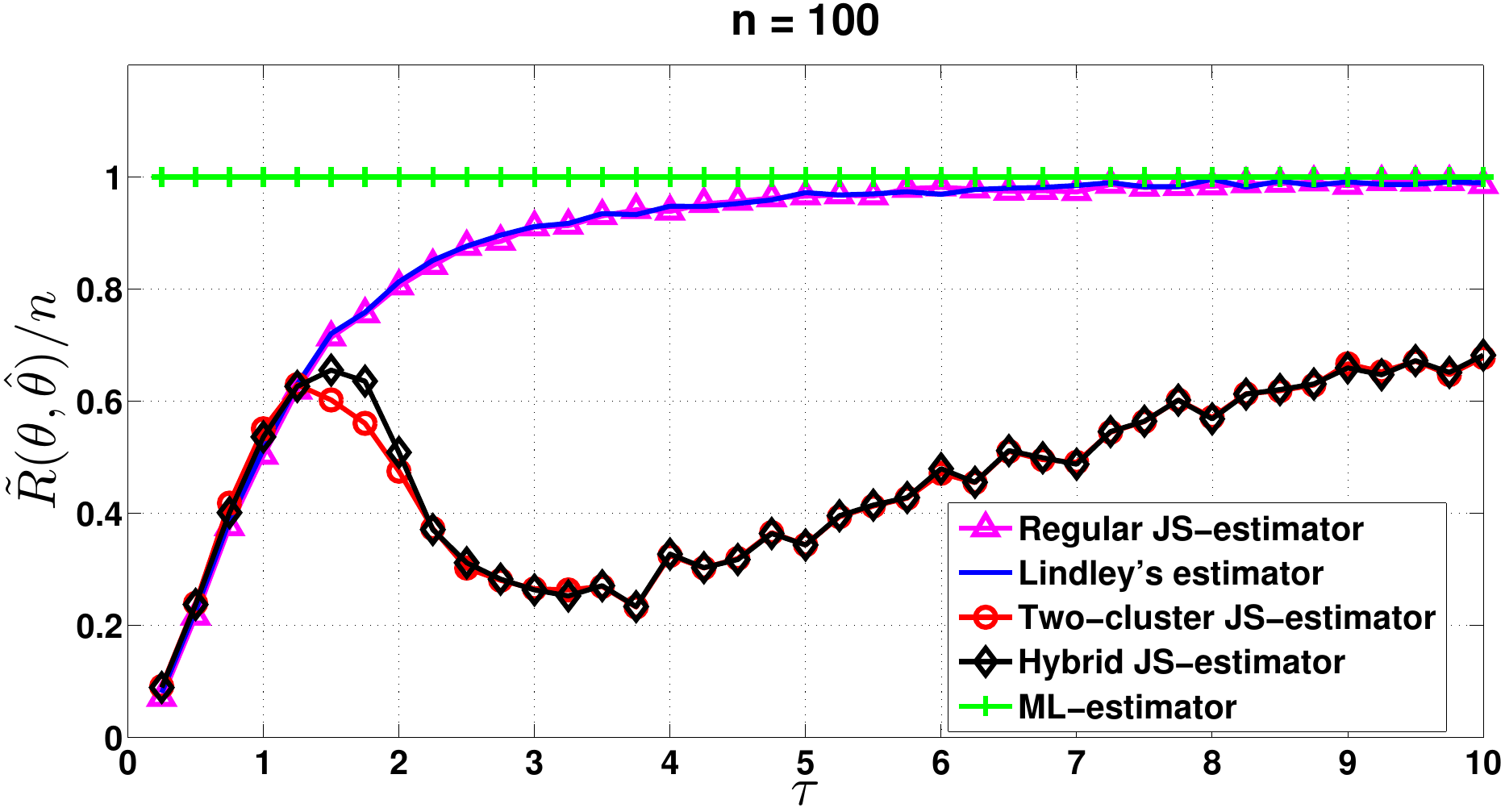}
       }
    \end{subfloat}
    \quad
    \begin{subfloat}[\label{subfig_n_100_0_5_50}]{
       \includegraphics[width= 3.25in]{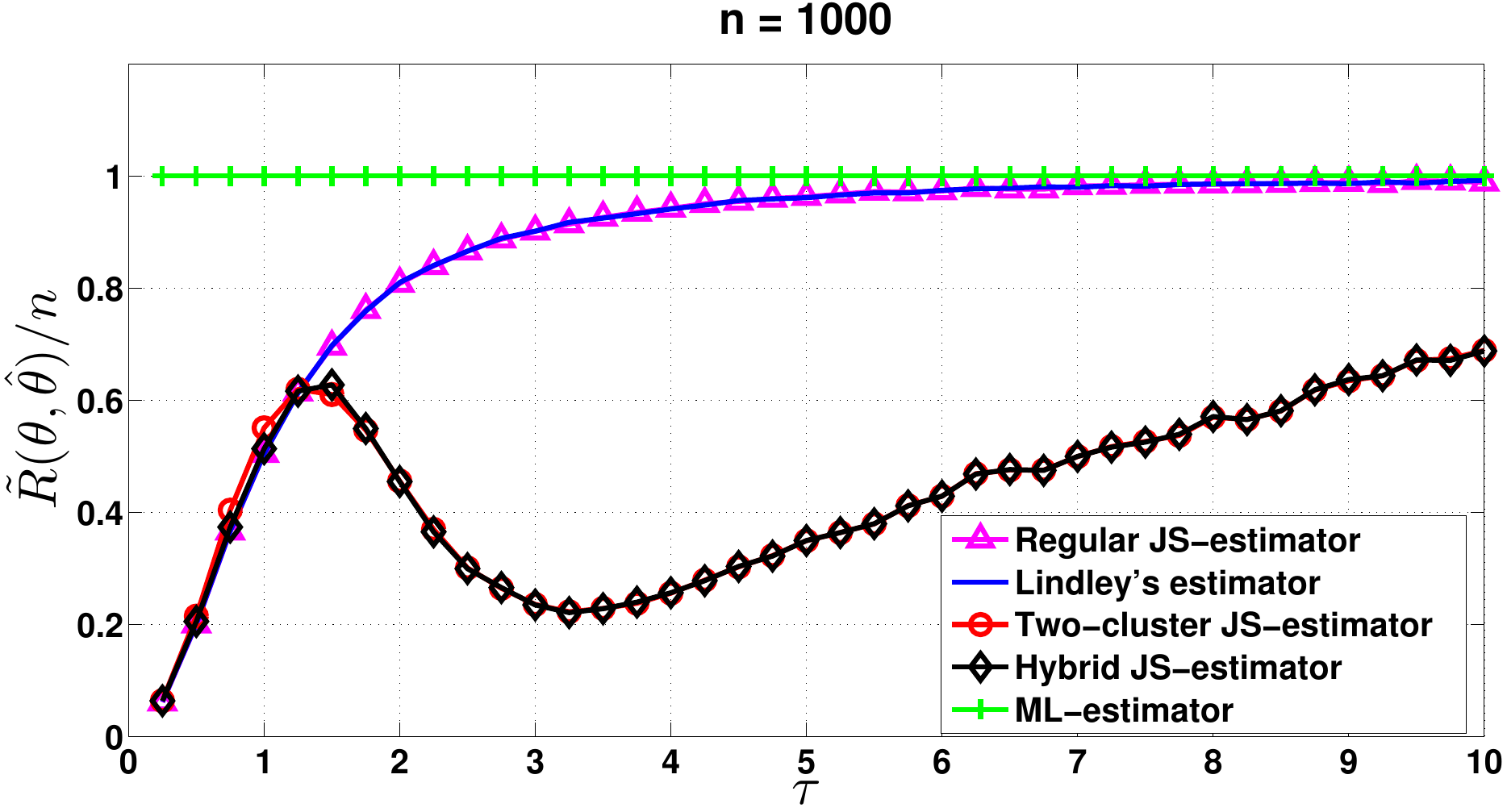}
       }
    \end{subfloat} 
    \caption{\small Average normalized loss of various estimators for different values of $n$. The $\{\theta_i\}_{i=1}^n$ are placed in two clusters, one centred at $\tau$ and another at $-\tau$. Each cluster  has width $0.5\tau$ and $n/2$ points. }
    \label{fig_2_hybrid_1}
\end{figure}
\begin{figure}[p]
    \centering
    \begin{subfloat}[\label{subfig_n_1000_0_25_30}]{
       \includegraphics[width= 3.25in]{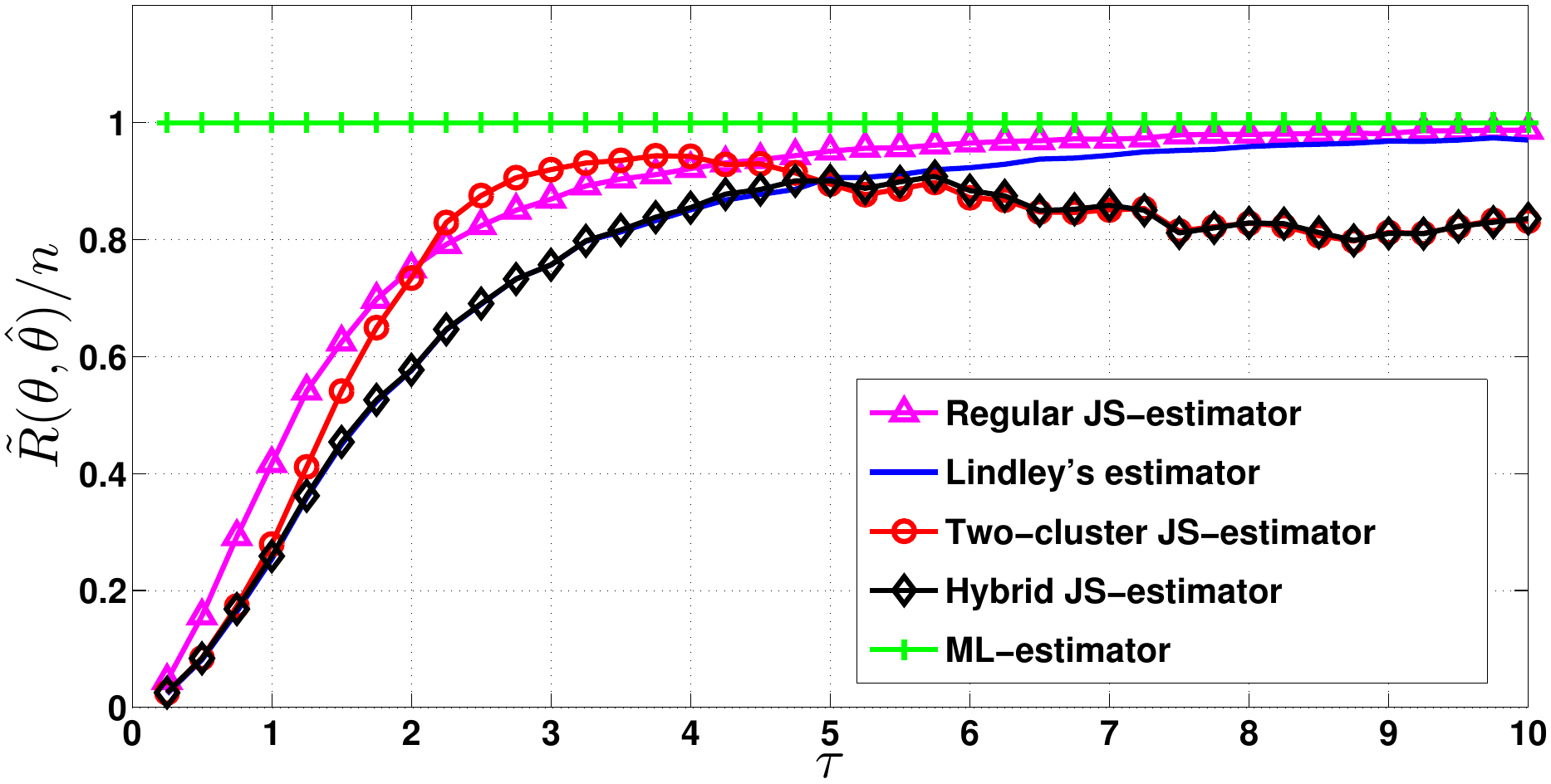}
       }
    \end{subfloat}
    \quad
    \begin{subfloat}[\label{subfig_n_1000_0_5_50}]{
       \includegraphics[width= 3.25in]
       {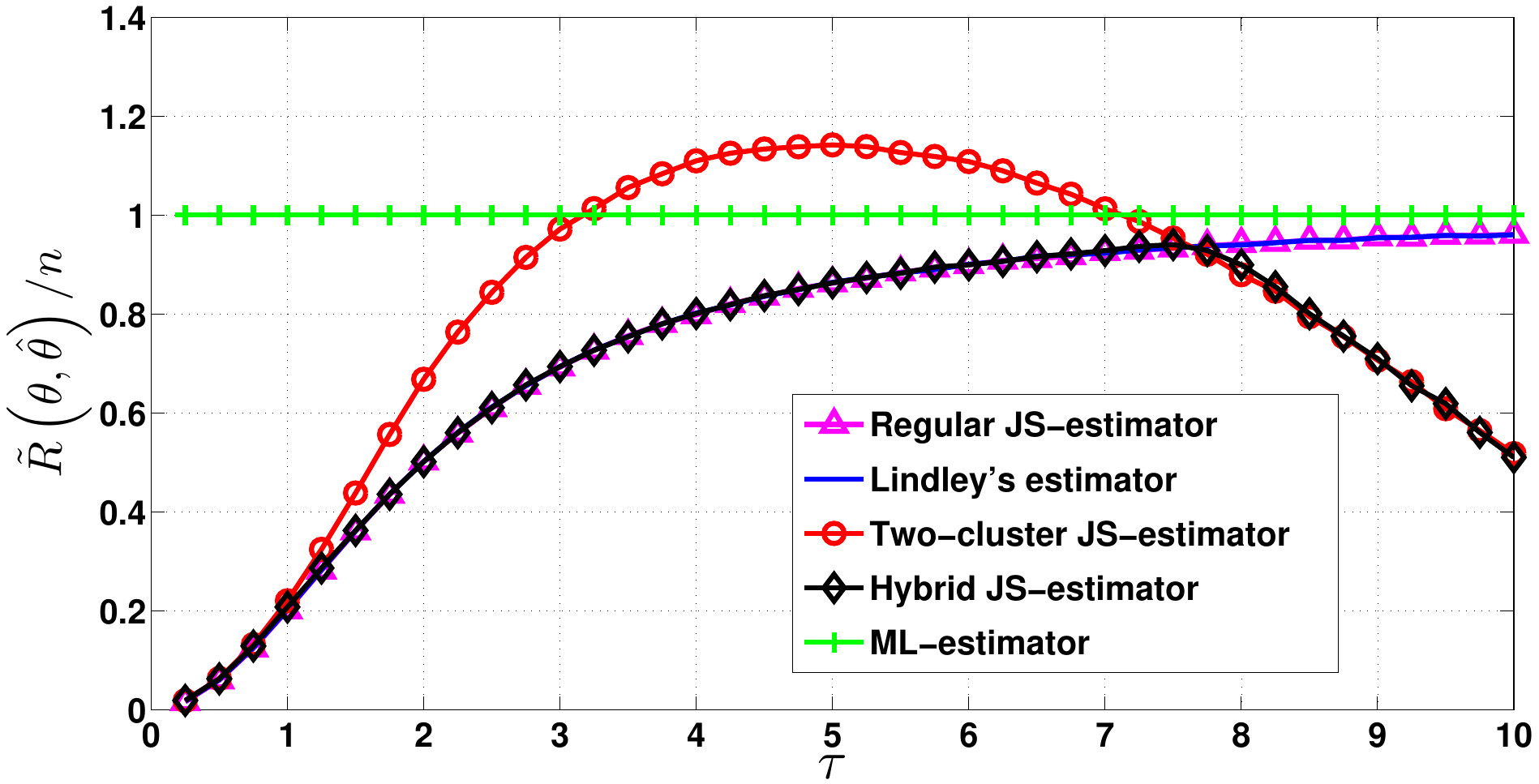}       
       }
    \end{subfloat}
    \quad
    \begin{subfloat}[\label{subfig_n_1000_0_125_30}]{
       \includegraphics[width= 3.25in]{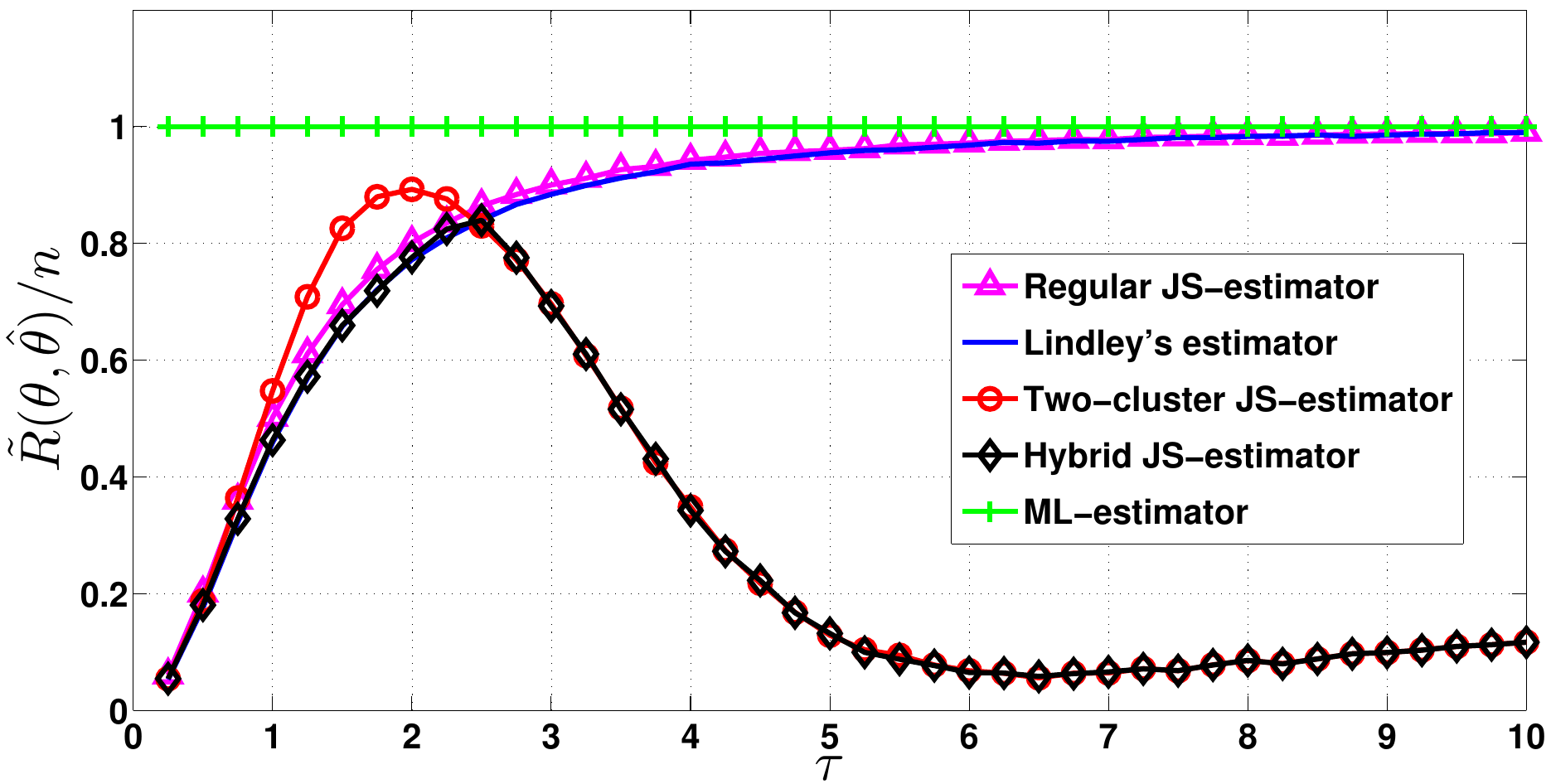}
       }
    \end{subfloat}
    \quad
    \begin{subfloat}[\label{subfig_n_1000_uniform}]{
       \includegraphics[width= 3.25in]{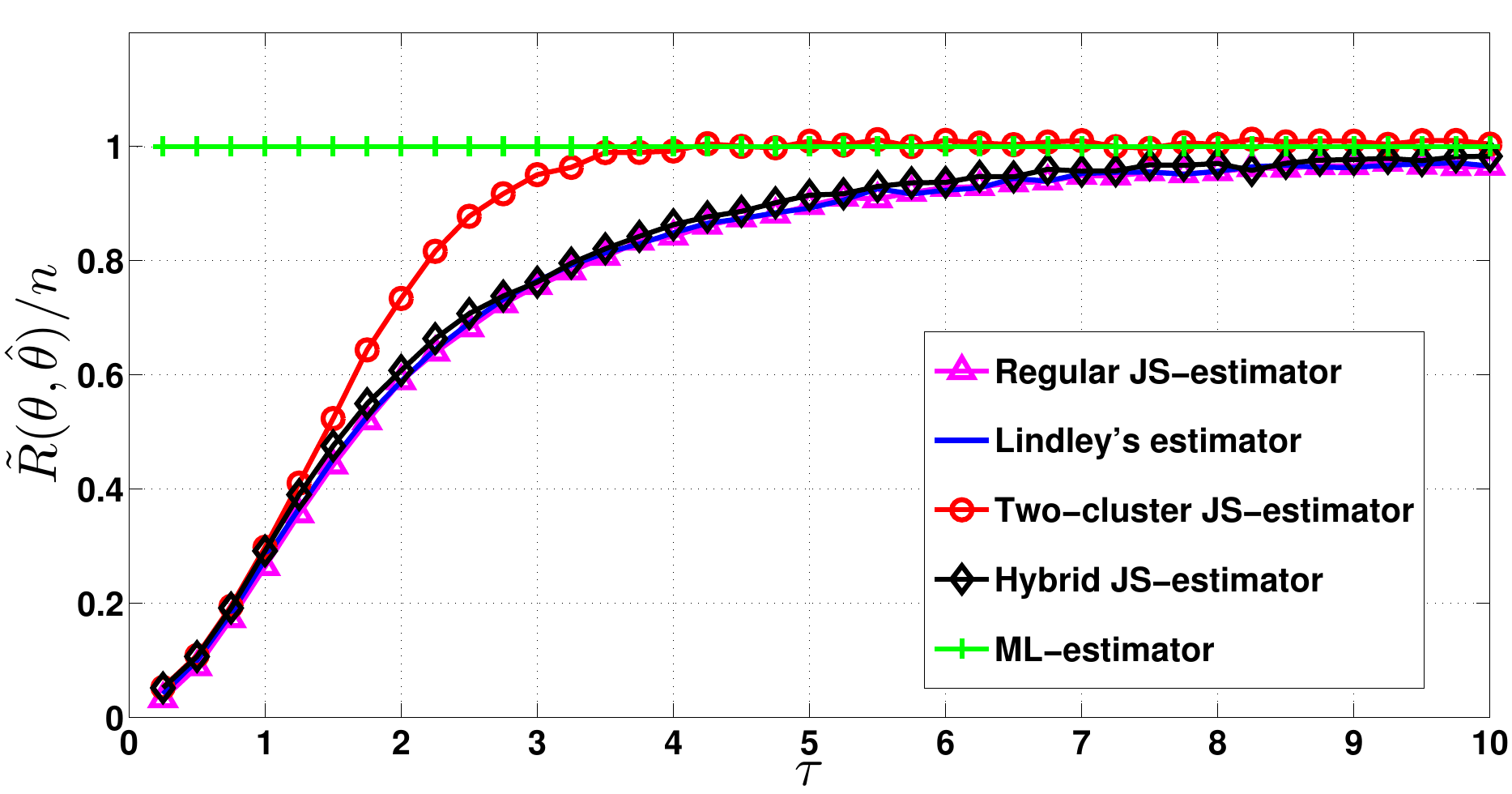}
       }
    \end{subfloat}    
    \caption{\small Average normalized loss of various estimators for different arrangements of the samples $\{\theta_i\}_{i=1}^n$, with $n=1000$. In  (a), there are 2 clusters each of width $0.5\tau$, one around $0.25\tau$ containing 300 points, and the other around $-\tau$ and containing 700 points. In (b), $\bst$ consists of $200$ components taking the value $\tau$ and the remaining $800$ taking the value $-0.25\tau$. In  (c), there are two clusters of  width  $0.125\tau$,  one  around $\tau$ containing $300$ points and another around $-\tau$ containing $700$ points. In (d), $\{\theta_i\}_{i=1}^n$ are arranged uniformly from $-\tau$ to $\tau$.}
    \label{fig_2_hybrid_2}
\end{figure}
\begin{figure}[p]
    \centering
    \begin{subfloat}[\label{subfig_risk_est_2_0_5_50}]{
       \includegraphics[width= 3.25in]{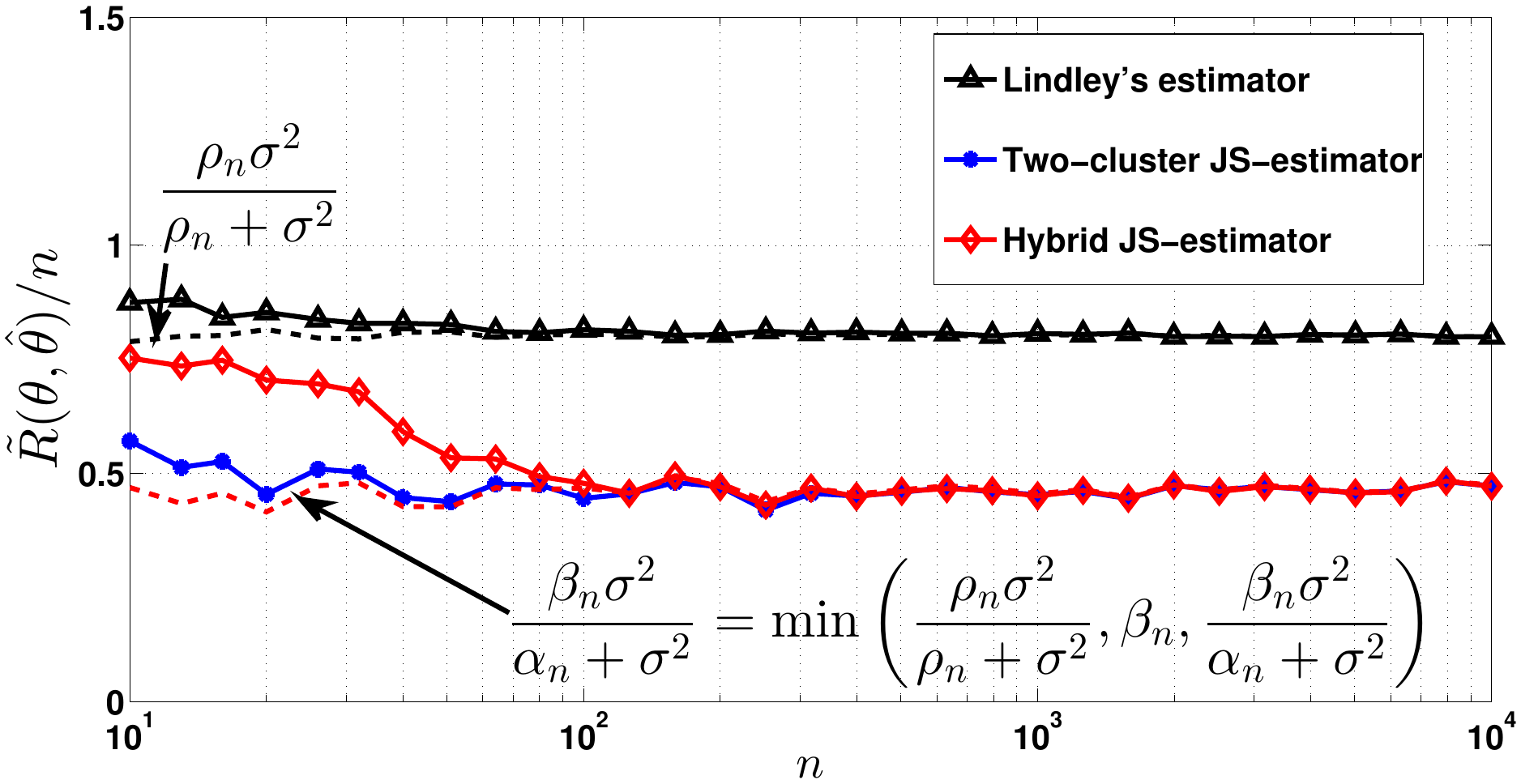}
       }
    \end{subfloat}
    \quad
    \begin{subfloat}[\label{subfig_risk_est_5_0_5_50}]{
       \includegraphics[width= 3.25in]{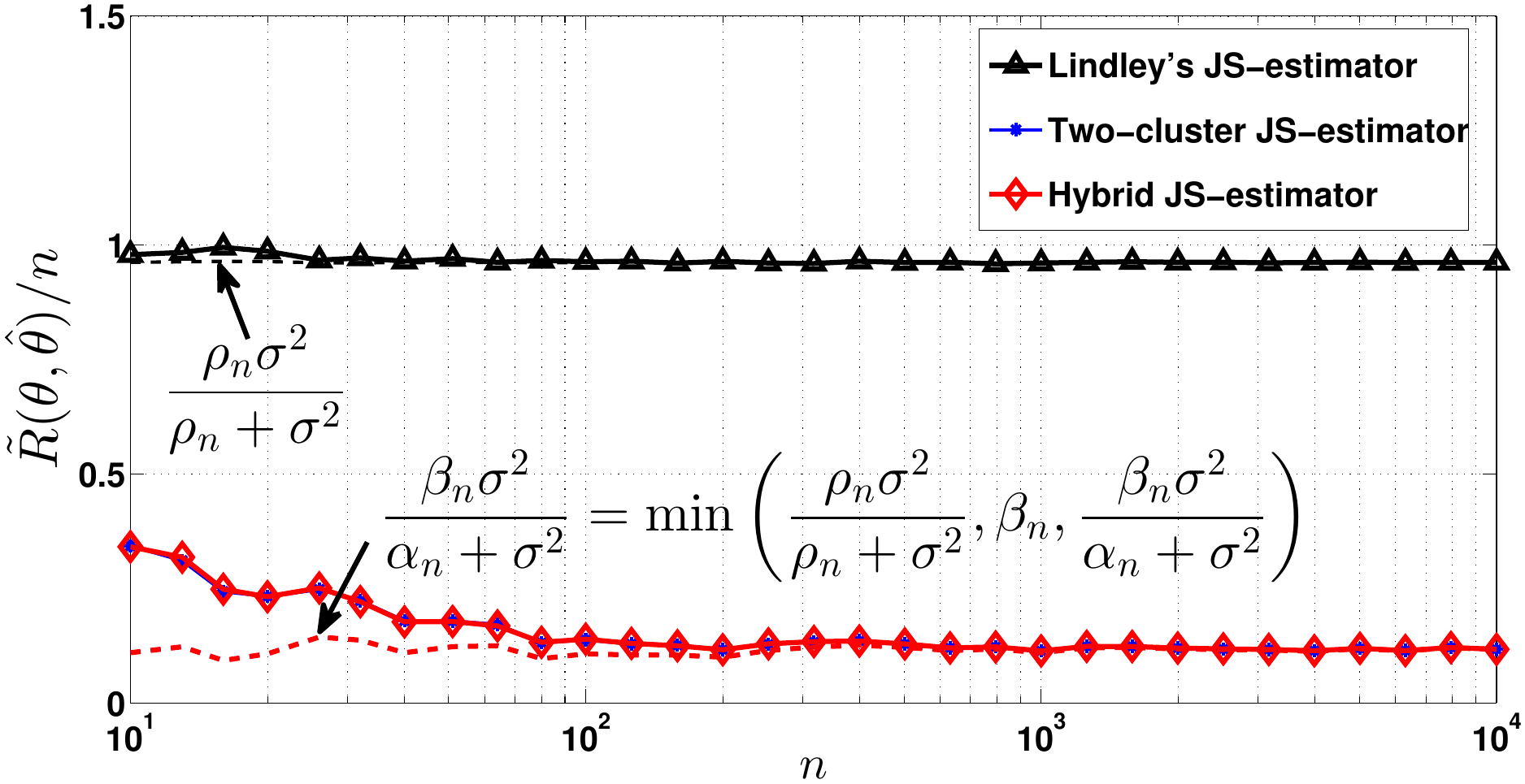}
       }
    \end{subfloat}
    \quad
    \begin{subfloat}[\label{subfig_risk_est_0_5_0_5_50}]{
       \includegraphics[width= 3.25in]{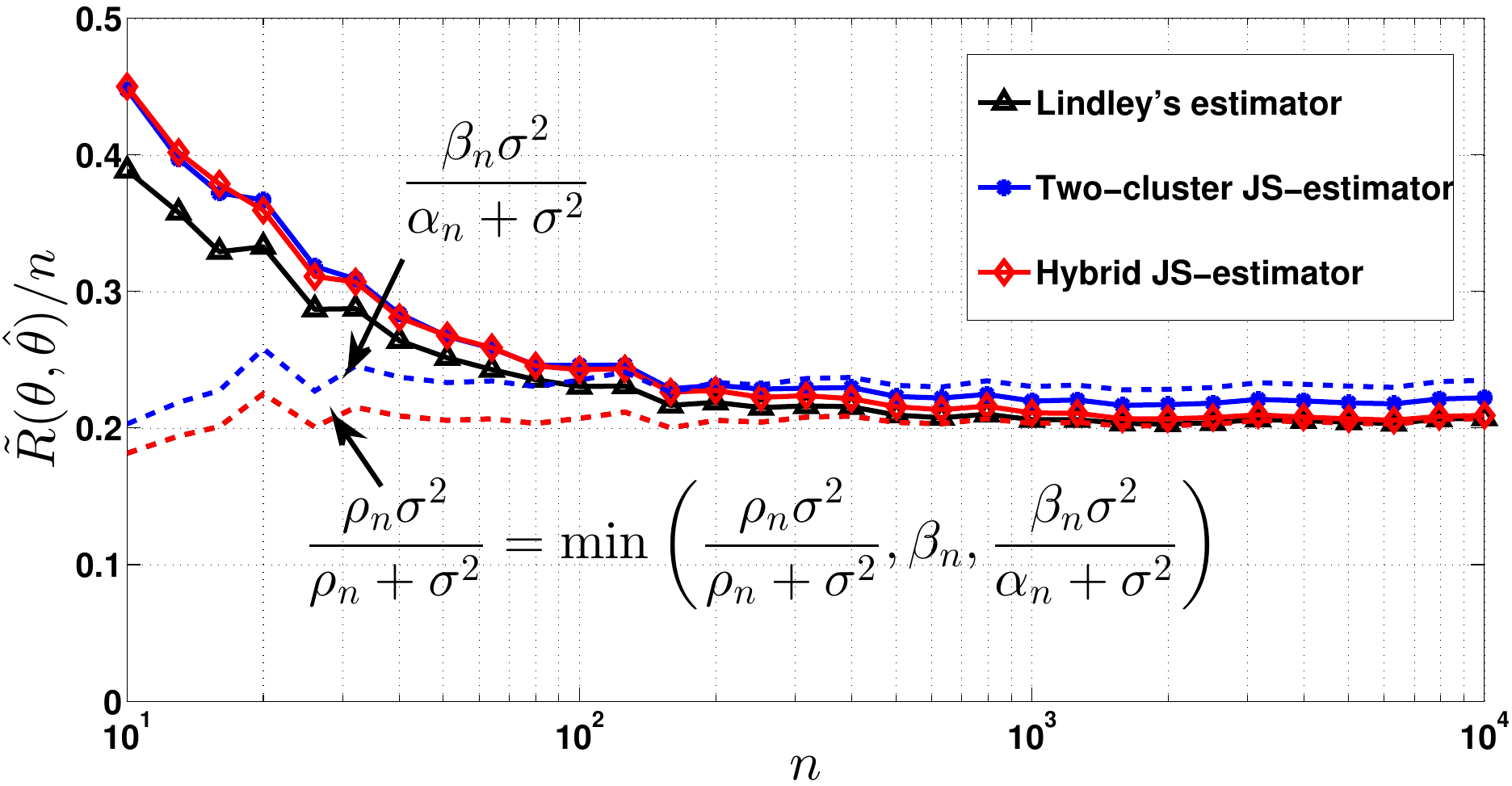}
       }
    \end{subfloat}
    \quad
    \begin{subfloat}[\label{subfig_risk_est_2_uniform}]{
       \includegraphics[width= 3.25in]{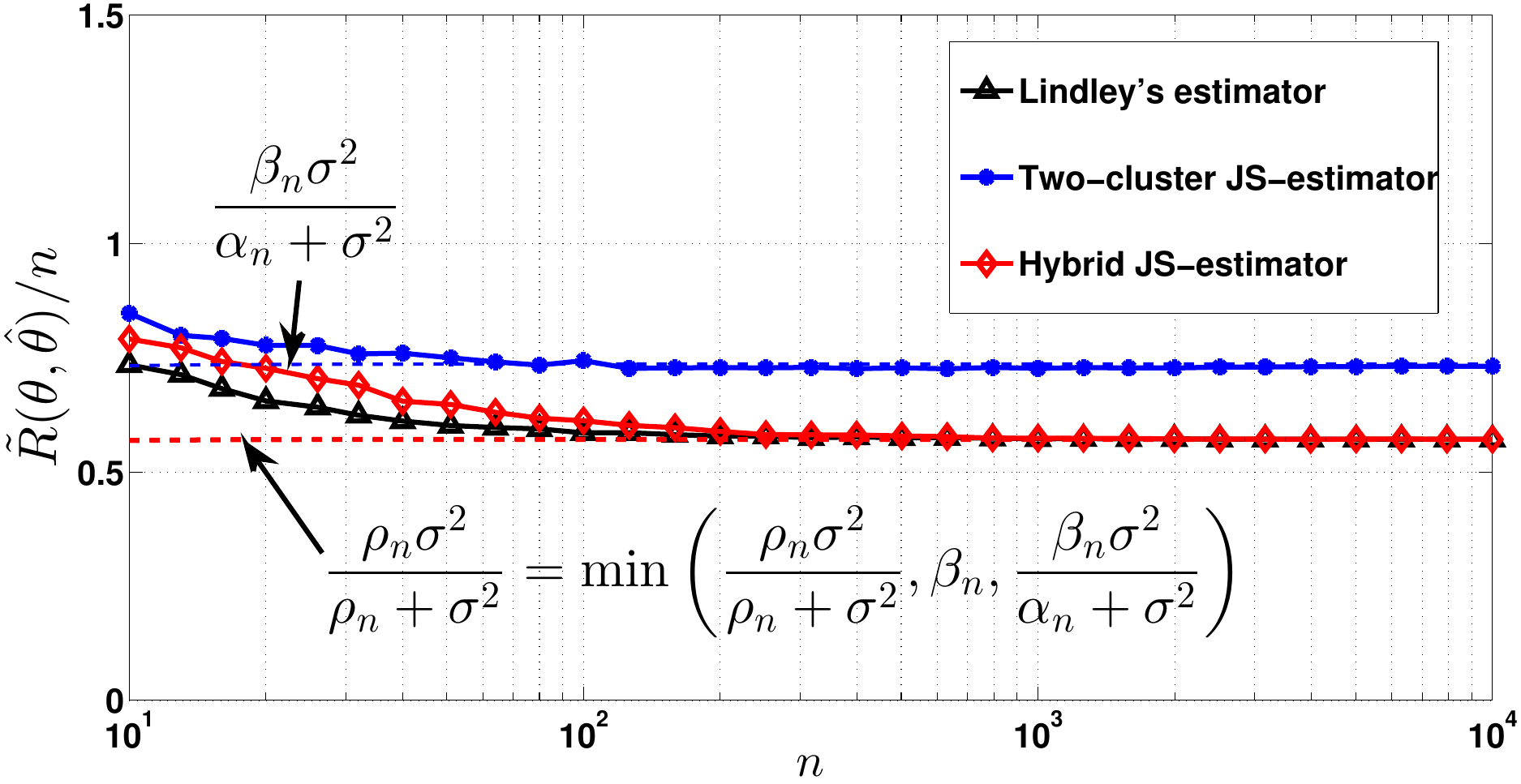}
       }
    \end{subfloat}  
    \caption{\small Average normalized loss of various estimators versus $\frac{\rho_n}{\rho_n + \sigma^2}$, $\frac{\beta_n}{\alpha_n + \sigma^2}$, $\min(\frac{\rho_n}{\rho_n + \sigma^2}, \beta_n, \frac{\beta_n}{\alpha_n + \sigma^2})$ as a function of $n$ for different arrangements of $\{\theta_i\}_{i=1}^n$. In (a), the $\{\theta_i\}_{i=1}^n$ are placed in two clusters of width $1$, one around $2$ and the other around $-2$, each containing an equal number of points. In (b), $\{\theta_i\}_{i=1}^n$ are placed in two clusters of width $1.25$, one around $5$ and the other around $-5$, each containing an equal number of points. In (c), $\{\theta_i\}_{i=1}^n$ are placed in two clusters of width $0.25$, one around $0.5$ and the other around $-0.5$, each containing an equal number of points. In (d), $\{\theta_i\}_{i=1}^n$ are placed uniformly between $-2$ and $2$. }
    \label{fig_risk_est}
\end{figure}

\begin{figure}[t]
    \centering
    \begin{subfloat}[\label{subfig4_n_1000_0_5}]{
       \includegraphics[width= 3.25in]{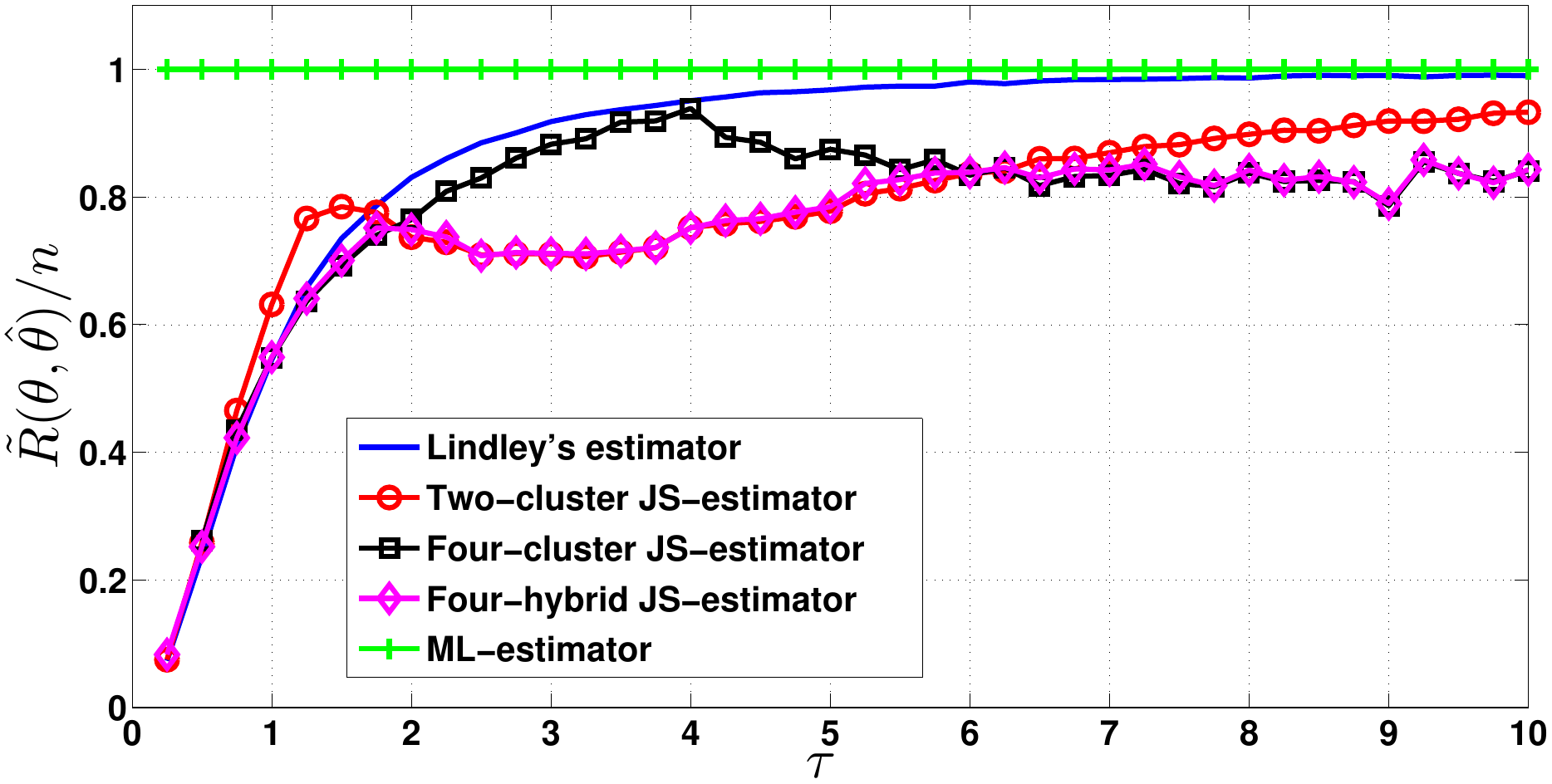}
       }
    \end{subfloat}
    \quad
    \begin{subfloat}[\label{subfig4_n_1000_0_25}]{
       \includegraphics[width= 3.25in]{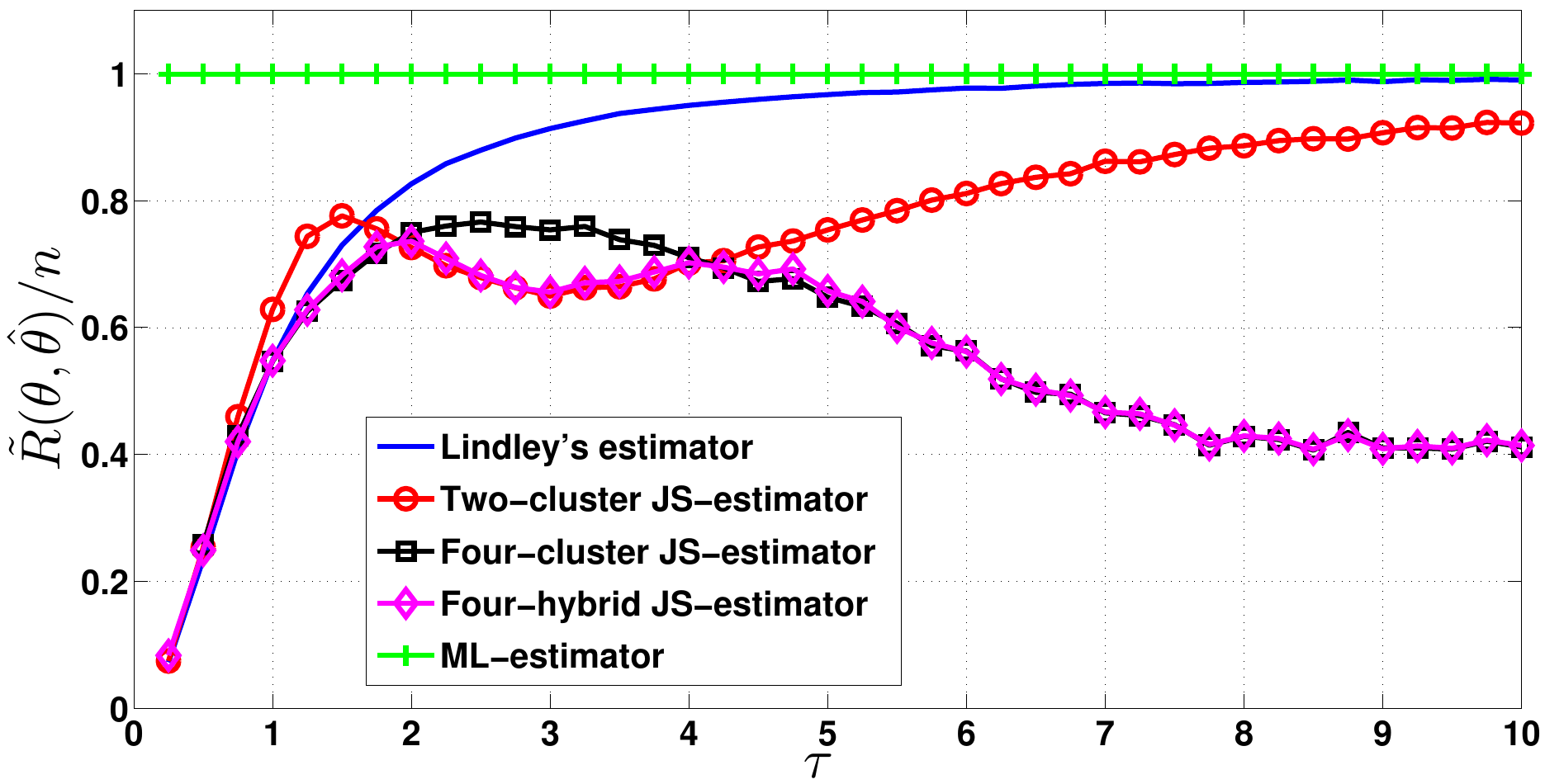}
       }
    \end{subfloat} 
    \caption{\small Average normalized loss of various estimators for $n = 1000$ and for different arrangements of the samples $\{\theta_i\}_{i=1}^n$. In (a), $\{\theta_i\}_{i=1}^n$ are placed in four equal-sized clusters of width $0.5\tau$ and in (b), the clusters are of width $0.25\tau$. In both cases, the clusters are centred at $1.5\tau$, $0.9\tau$, $-0.5\tau$ and $-1.25\tau$. }
    \label{fig_4_hybrid_2}
    \vspace{-5pt}
\end{figure}

In this section, we present simulation plots that compare the average normalized loss of the proposed estimators with that of the regular JS-estimator and Lindley's estimator, for various choices of $\bst$.  In each plot, the normalized loss, labelled 
$\frac{1}{n}\tilde{R}( \boldsymbol{\theta}, \hat{\boldsymbol{\theta}})$ on the $Y$-axis,  is computed by averaging over $1000$ realizations of $\mathbf{w}$.
 We use $\mathbf{w} \sim \mathcal{N}(0,\mathbf{I})$, i.e., the noise variance $\sigma^2=1$.  Both the regular JS-estimator $\hat{\boldsymbol{\theta}}_{JS}$ and Lindley's estimator $\hat{\boldsymbol{\theta}}_{JS_1}$ used are the positive-part versions, respectively given by \eqref{pp_JS} and \eqref{pp_lindley}. We choose $\delta = {5}/{\sqrt{n}}$ for our proposed estimators.
 
In Figs. \ref{fig_2_hybrid_1}--\ref{fig_4_hybrid_2}, we consider three different structures for $\bst$, representing varying degrees of clustering. In the first structure, the components $\{ \theta_i \}_{i=1}^n$
are arranged in two clusters.   In the second structure for $\bst$,  $\{ \theta_i \}_{i=1}^n$ are uniformly distributed within an interval whose length is varied.  In the third structure, $\{\theta_i \}_{i=1}^n$ are arranged in four clusters.  In both clustered structures, the locations and the widths of the clusters as well as the number of points within each cluster are varied; the locations of the points within each cluster are chosen uniformly at random.  The captions of the figures explain the details of each structure.

In Fig. \ref{fig_2_hybrid_1},  $\{\theta_i\}_{i=1}^n$ are arranged in two clusters, one centred at $-\tau$ and the other at $\tau$.  The plots show the average normalized loss as a function of $\tau$ for different values of $n$,  for four estimators: $\hat{\boldsymbol{\theta}}_{JS}$, $\hat{\boldsymbol{\theta}}_{JS_1}$, the two-attractor JS-estimator $\hat{\boldsymbol{\theta}}_{JS_2}$ given by \eqref{eq_two_part_estimator}, and the hybrid JS-estimator $\hat{\boldsymbol{\theta}}_{JS_{H}}$ given by \eqref{comb_estimator}. We observe that as $n$ increases, the average loss of 
$\hat{\boldsymbol{\theta}}_{JS_{H}}$ gets closer to the minimum of that of $\hat{\boldsymbol{\theta}}_{JS_{1}}$ and $\hat{\boldsymbol{\theta}}_{JS_{2}}$; 

Fig. \ref{fig_2_hybrid_2} shows the the average normalized loss for different arrangements of $\{\theta_i\}$, with $n$ fixed at $1000$. The plots illustrate a few cases where $\hat{\boldsymbol{\theta}}_{JS_{2}}$ has significantly lower risk than $\hat{\boldsymbol{\theta}}_{JS_{1}}$, and also the strength of $\hat{\boldsymbol{\theta}}_{JS_{H}}$ when $n$ is large. 

Fig. \ref{fig_risk_est} compares the average normalized losses of $\hat{\boldsymbol{\theta}}_{JS_1}$, $\hat{\boldsymbol{\theta}}_{JS_2}$, and $\hat{\boldsymbol{\theta}}_{JS_{H}}$ with their asymptotic risk values, obtained in Corollary \ref{cor_pp_JS_Lindley}, Theorem \ref{thm1}, and Theorem \ref{thm2}, respectively. Each subfigure considers a different arrangement of $\{\theta_i\}_{i=1}^n$, and shows how  the average losses converge to their respective theoretical values with growing $n$.

Fig. \ref{fig_4_hybrid_2}  demonstrates the effect of choosing four attractors when $\{ \theta_i \}_{i=1}^n$ form four clusters. The four-hybrid estimator $\hat{\boldsymbol{\theta}}_{JS_{H,4}}$ attempts to choose the best among $\hat{\boldsymbol{\theta}}_{JS_1}$, $\hat{\boldsymbol{\theta}}_{JS_2}$ and $\hat{\boldsymbol{\theta}}_{JS_4}$ based on the data $\mbf{y}$. It is clear that depending on the values of $\{\theta_i\}$, $\hat{\boldsymbol{\theta}}_{JS_{H,4}}$ reliably tracks the best of these. and can have significantly lower loss than both $\hat{\boldsymbol{\theta}}_{JS_1}$ and $\hat{\boldsymbol{\theta}}_{JS_{2}}$, especially for large values of $n$.


\section{Proofs} \label{sec:proofs}

\subsection{Mathematical preliminaries} \label{subsec:prelim}

Here we list some concentration results that are used in the proofs of the theorems. 

\begin{lemma}\label{prop_as}
 Let $\{X_n(\bst),\bst \in \mathbb{R}^n\}_{n=1}^\infty$ be a sequence of random variables such that $X_n(\bst) \doteq 0$, i.e., for any $\epsilon > 0$,
 \begin{equation*}
  \mathbb{P}(\vert X_n(\bst) \vert \geq \epsilon ) \leq Ke^{-\frac{nk\min(\epsilon^2,1)}{\max(\Vert \bst\Vert^2/n,1)}},
 \end{equation*}
where $K$ and $k$ are positive constants. If $C \vcentcolon= \limsup_{n \to \infty}\Vert \bst \Vert^2/n < \infty$, then $X_n \overset{a.s.}{\longrightarrow} 0$.
\end{lemma}
\begin{proof}
For any $\tau > 0$, there exists a positive integer $M$ such that $\forall n \geq M$, $\Vert \bst\Vert^2/n < C + \tau$. Hence, we have, for any $\epsilon >0$, and for some $\tau > 1$,
\begin{align*}
 \sum_{n=1}^{\infty}\mathbb{P}(\vert X_n(\bst) \vert \geq \epsilon ) & \leq \sum_{n=1}^{\infty} Ke^{-\frac{nk\min(\epsilon^2,1)}{\max(\Vert \bst\Vert^2/n,1)}} \\
 &\leq C_0 + \sum_{n=M}^{\infty} Ke^{-\frac{nk\min(\epsilon^2,1)}{C + \tau}} < \infty.
\end{align*}
 Therefore, we can use the Borel-Cantelli lemma to conclude that $X_n \overset{a.s.}{\longrightarrow} 0$.  
\end{proof}

\begin{lemma}\label{lem_4}
For  sequences of random variables $\{X_n\}_{n=1}^\infty$, $\{Y_n\}_{n=1}^\infty$ such that $X_n \doteq 0 $, $Y_n  \doteq 0$, it follows that $ X_n + Y_n \doteq 0$.  
\end{lemma}
 \begin{proof}
For $\epsilon >0$, if  $\mathbb{P}(\vert X_n \vert \geq \epsilon) \leq K_1e^{-\frac{nk_1\min(\epsilon^2,1)}{\max(\Vert \bst\Vert^2/n,1)}}$ and $\mathbb{P}(\vert Y_n \vert \geq \epsilon) \leq K_2e^{-\frac{nk_2\min(\epsilon^2,1)}{\max(\Vert \bst\Vert^2/n,1)}}$ for positive constants $K_1$, $K_2$, $k_1$ and $k_2$, then by the triangle inequality,
\begin{align*}
 \mathbb{P}(\vert X_n + Y_n \vert \geq \epsilon) & \leq \mathbb{P}\left(\vert X_n \vert \geq \frac{\epsilon}{2}\right) + \mathbb{P}\left(\vert  Y_n \vert \geq \frac{\epsilon}{2}\right) \\
 & \leq Ke^{-\frac{nk\min(\epsilon^2,1)}{\max(\Vert \bst\Vert^2/n,1)}}
\end{align*}
where $K=K_1+K_2$ and $k = \min\left(\frac{k_1}{4}.\frac{k_2}{4} \right)$.
 \end{proof}

 \begin{lemma}\label{lem1_app2}
 Let $X$ and $Y$ be random variables such that for any $\epsilon >0$, 
 \begin{equation*}
 \begin{split}
 \mathbb{P}\left( \left\vert X - a_X \right\vert \geq \epsilon \right) \leq K_1e^{-nk_1 \min\left(\epsilon^2,1\right)},\\
 \mathbb{P}\left( \left\vert Y - a_Y \right\vert  \geq \epsilon \right) \leq K_2e^{-nk_2 \min\left(\epsilon^2,1\right)}
 \end{split}
 \end{equation*}
 where $k_1$, $k_2$ are positive constants, and $K_1$, $K_2$ are positive integer constants. Then,
 \begin{equation*}
  \mathbb{P}\left( \left\vert XY - a_Xa_Y \right\vert \geq \epsilon \right) \leq Ke^{-nk\min\left(\epsilon^2,1\right)} 
  \end{equation*}
 where $K = 2(K_1 + K_2)$, and $k$ is a positive constant depending on $k_1$ and $k_2$.
\end{lemma}
\begin{proof}
We have
 \begin{align*}
&\mathbb{P}\left( \left\vert XY - a_Xa_Y \right\vert \geq \epsilon \right) \\
&= \mathbb{P}\left( \left\vert (X-a_X)(Y-a_Y) + Xa_Y +Ya_X- 2a_Xa_Y \right\vert \geq \epsilon \right)\\
& \leq \mathbb{P}\left( \left\vert (X-a_X)(Y-a_Y) \right\vert \geq \frac{\epsilon}{2} \right) \\
& \quad + \mathbb{P}\left(\left\vert Xa_Y +Ya_X- 2a_Xa_Y \right\vert \geq \frac{\epsilon}{2} \right)\\
& \leq \mathbb{P}\left( \left\vert (X-a_X)(Y-a_Y) \right\vert \geq \frac{\epsilon}{2} \right)  \\
& \quad + \mathbb{P}\left(\left\vert (X-a_X)a_Y \right\vert \geq \frac{\epsilon}{4}\right) + \mathbb{P}\left(\left\vert (Y-a_Y)a_X \right\vert \geq \frac{\epsilon}{4}\right)\\
&  \leq \mathbb{P}\left( \left\vert (X-a_X)\right\vert \geq \sqrt{\frac{\epsilon}{2}} \right) + \mathbb{P}\left( \left\vert (Y-a_Y) \right\vert \geq \sqrt{\frac{\epsilon}{2}} \right) \\
& \quad + \mathbb{P}\left(\left\vert (X-a_X)a_Y \right\vert \geq \frac{\epsilon}{4}\right) + \mathbb{P}\left(\left\vert (Y-a_Y)a_X \right\vert \geq \frac{\epsilon}{4}\right)\\
& \leq K_1e^{-nk_1^\prime \min\left(\epsilon,1\right)} + K_2e^{-nk_2^\prime \min\left(\epsilon,1\right)} + K_1e^{-nk_1^{\prime \prime} \min\left(\epsilon^2,1\right)}\\
& \quad + K_2e^{-nk_2^{\prime \prime} \min\left(\epsilon^2,1\right)} \leq Ke^{-nk \min\left(\epsilon^2,1\right)}
\end{align*}
where $k_1^\prime = \frac{k_1}{2}$, $k_2^\prime = \frac{k_2}{2}$, $k_1^{\prime \prime}= \frac{k_1}{16(a_Y)^2}$, $k_2^{\prime\prime} = \frac{k_2}{16(a_Y)^2}$, and $k = \frac{\min(k_1,k_2)}{16(a_Y)^2}$.
\end{proof}

\begin{lemma}\label{lem1_app}
 Let $Y$ be a non-negative random variable such that there exists $a_Y > 0$ such that for any $\epsilon >0$, 
 \begin{align*}
 \mathbb{P}\left( Y - a_Y  \leq -\epsilon \right) & \leq  K_1e^{-nk_1 \min\left(\epsilon^2,1\right)},\\
 \mathbb{P}\left( Y - a_Y  \geq \epsilon \right) & \leq  K_2e^{-nk_2 \min\left(\epsilon^2,1\right)} 
 \end{align*}
 where $k_1,k_2$ are positive constants, and $K_1,K_2$ are positive integer constants. Then, for any $\epsilon > 0$,
 \begin{equation*}
  \mathbb{P}\left( \left\vert \frac{1}{Y} - \frac{1}{a_Y} \right\vert \geq \epsilon \right) \leq Ke^{-nk \min\left(\epsilon^2,1\right)} 
  \end{equation*}
  where $K = K_1+K_2$, and $k$ is a positive constant.
\end{lemma}
\begin{proof}
 We have 
\begin{align} \nonumber
&  \mathbb{P}\left(\frac{1}{Y} - \frac{1}{a_Y} \geq \epsilon \right)  = \mathbb{P}\left(\frac{1}{Y} \geq \frac{1}{a_Y} + \epsilon \right) \\ \nonumber
& = \mathbb{P}\left(Y \leq \frac{1}{\epsilon + {1}/{a_Y}}  \right) =  \mathbb{P}\left(Y - a_Y\leq \frac{1}{\epsilon + {1}/{a_Y}} - a_Y \right)\\ \nonumber
&=\mathbb{P}\left(Y - a_Y \leq -a_Y\left(\frac{\epsilon a_Y}{1 + \epsilon a_Y}\right)\right)\\ \label{lem1_app_eq1}
 & \leq K_1 e^{-nk_1\min\left(\left(\frac{\epsilon(a_Y)^2}{1 + \epsilon a_Y}\right)^2,1\right)} \leq K_1 e^{-nk_1^\prime\min\left(\epsilon^2,1\right)}
\end{align}
where $k_1^\prime = k_1\min\left(\left(\frac{(a_Y )^2}{1 + a_Y}\right)^2,1\right)$. Similarly
\begin{align*} 
 \mathbb{P}\left(\frac{1}{Y} - \frac{1}{a_Y} \leq -\epsilon \right) & = \mathbb{P}\left(\frac{1}{Y} \leq \frac{1}{a_Y} - \epsilon \right).
 \end{align*}
 Note that when $ \epsilon > \frac{1}{a_Y}$, $\mathbb{P}\left(\frac{1}{Y} - \frac{1}{a_Y} \leq -\epsilon \right) =  0$ because $Y > 0$. Therefore,
  \begin{align} \nonumber
  & \mathbb{P}\left(\frac{1}{Y} - \frac{1}{a_Y} \leq -\epsilon \right)
   = \mathbb{P}\left(Y \geq \frac{1}{ \frac{1}{a_Y} - \epsilon }  \right) \\\nonumber 
   &=  \mathbb{P}\left(Y - a_Y\geq \frac{1}{\frac{1}{a_Y}-\epsilon} - a_Y \right)\\ \nonumber
  &=\mathbb{P}\left(Y - a_Y \geq a_Y\left(\frac{\epsilon a_Y}{1 - \epsilon  a_Y}\right)\right)\\ \nonumber
  & \leq K_2 e^{-nk_2\min\left(\left(\frac{\epsilon(a_Y)^2}{1 - \epsilon a_Y}\right)^2,1\right)} \leq K_2 e^{-nk_2\min\left(\epsilon^2(a_Y)^4,1\right)}\\ \label{lem1_app_eq2} 
  &\leq K_2 e^{-nk_2^\prime\min\left(\epsilon^2,1\right)}
 \end{align}
where $k_2^\prime = k_2\min\left((a_Y)^4,1\right)$. Using \eqref{lem1_app_eq1} and \eqref{lem1_app_eq2}, we obtain, for any $\epsilon >0$,
 \begin{equation*}
  \mathbb{P}\left( \left\vert \frac{1}{Y} - \frac{1}{a_Y} \right\vert \geq \epsilon \right) \leq Ke^{-nk \min\left(\epsilon^2,1\right)} 
  \end{equation*}
  where $k = \min(k_1^\prime, k_2^\prime)$.
\end{proof}

\begin{lemma}\label{app_lemma_g}
 Let $\{X_n \}_{n=1}^\infty$ be a sequence of random variables and $X$ be another random variable (or a constant) such that for any $\epsilon > 0 $, $\mathbb{P}\left(\vert X_n - X \vert \geq \epsilon \right) \leq Ke^{-nk\min\left(\epsilon^2,1\right)}$ for positive constants $K$ and $k$. Then, for the function $g(x) \vcentcolon= \max(\sigma^2,x)$, we have
 \begin{equation*}
  \mathbb{P}\left(\vert g(X_n) - g(X) \vert \geq \epsilon \right) \leq Ke^{-nk\min\left(\epsilon^2,1\right)}.
 \end{equation*}
\end{lemma}
\begin{proof}
Let $\mathcal{A}_n \vcentcolon = \{ \omega ~\big \vert~ \vert X_n(\omega) - X(\omega) \vert \geq \epsilon \}$. 
 We have 
 \begin{align} \nonumber
  &\mathbb{P}\left(\vert g(X_n) - g(X) \vert \geq \epsilon \right) \\ \nonumber
  & =  \mathbb{P}\left( \mathcal{A}_n \right) \mathbb{P}\left(\vert g(X_n) - g(X) \vert \geq \epsilon ~\big\vert~ \mathcal{A}_n \right)\\ \nonumber
  & \quad + \mathbb{P}\left(\mathcal{A}_n^c  \right) \mathbb{P}\left(\vert g(X_n) - g(X) \vert \geq \epsilon ~\big\vert~ \mathcal{A}_n^c \right)\\ \label{eq1_app_lemma_g}
  &\leq Ke^{-nk\min\left(\epsilon^2,1\right)} +  \mathbb{P}\left(\vert g(X_n) - g(X) \vert \geq \epsilon ~\big\vert~\mathcal{A}_n^c \right).\
 \end{align}
Now, when $X_n \geq \sigma^2$ and $X \geq \sigma^2$, it follows that $g(X_n) - g(X) = X_n -X$, and the second term of the RHS of \eqref{eq1_app_lemma_g} equals $0$, as it also does when $X_n < \sigma^2$ and $X < \sigma^2$. Let us consider the case where $X_n \geq \sigma^2$ and $X < \sigma^2$. Then, $g(X_n) - g(X) = X_n - \sigma^2 < X_n - X < \epsilon$, as we condition on the fact that $\vert X_n - X\vert < \epsilon$; hence in this case $\mathbb{P}\left(\vert g(X_n) - g(X) \vert \geq \epsilon ~\big\vert~ \mathcal{A}_n^c\right) = 0$. Finally, when $X_n < \sigma^2$ and $X \geq \sigma^2$, we have $   g(X) - g(X_n)  = X - \sigma^2 < X - X_n < \epsilon$; hence in this case also we have $\mathbb{P}\left(\vert g(X_n) - g(X) \vert \geq \epsilon ~\big\vert~ \mathcal{A}_n^c\right) = 0$. This proves the lemma.
\end{proof}

\begin{lemma}
(Hoeffding's Inequality \cite[Thm. 2.8]{boucheron2}). Let $X_1,\cdots, X_n$ be independent random variables such that $X_i \in [a_i, b_i]$ almost surely, for all $i \leq n$. Let $S_n = \sum_{i=1}^n (X_i - \mathbb{E}[X_i])$. Then for any $\epsilon > 0$, 
$\mathbb{P}\left( \frac{\abs{S_n}}{n} \geq \epsilon\right) \leq 2 e^{-\frac{2n^2 \epsilon^2}{\sum_{i=1}^n(b_i-a_i)^2}}$.
\end{lemma}

\begin{lemma}\label{lem_app_conc_ineq3} (Chi-squared concentration \cite{boucheron2}).
For i.i.d. Gaussian random variables $w_1, \ldots, w_n \sim \mathcal{N}(0, \sigma^2)$, we have for any $\epsilon > 0$,
\begin{align*}
\mathbb{P}\left(\left\vert\frac{1}{n}\sum_{i=1}^n w_i^2 - \sigma^2 \right \vert \geq \epsilon\right) \leq 2 e^{-nk\min(\epsilon,\epsilon^2)},
\end{align*}
where $k = \min\left(\frac{1}{4\sigma^4},\frac{1}{2\sigma^2}\right)$.
\end{lemma}

\begin{lemma}\label{lem_app_conc_ineq1}
For $i = 1, \cdots,n$, let $w_i \sim (0, \sigma^2)$ be independent, and $a_i$ be real-valued and finite constants.  We have for any $\epsilon > 0$,
\begin{align}
\nonumber
&\mathbb{P}\left(\frac{1}{n}\left\vert\sum_{i=1}^n w_i\mathsf{1}_{\left\{w_i > a_i\right\}} -\frac{\sigma}{\sqrt{2\pi}}\sum_{i=1}^ne^{-\frac{a_i^2}{2\sigma^2}}\right \vert \geq \epsilon\right) \\ \label{lem_app_conc_ineq1_eq1}
& \leq 2e^{-nk_1\min(\epsilon,\epsilon^2)} ,\\ 
\nonumber
&\mathbb{P}\left(\frac{1}{n}\left\vert\sum_{i=1}^n w_i\mathsf{1}_{\left\{w_i \leq a_i\right\}} +\frac{\sigma}{\sqrt{2\pi}}\sum_{i=1}^ne^{-\frac{a_i^2}{2\sigma^2}}\right \vert  \geq \epsilon\right) \\ \label{lem_app_conc_ineq1_eq2} 
&\leq 2e^{-nk_2\min(\epsilon,\epsilon^2)}
\end{align}
where $k_1$ and $k_2$ are positive constants.
\end{lemma}
The proof is given in Appendix \ref{pf_lem_app_conc_ineq1}.

\begin{lemma} \label{lem_app_small_O_P}
 Let $\mbf{y} \sim \mc{N}\left(\bst,\sigma^2\mbf{I}\right)$, and let $f: \mathbb{R}^n \to \mathbb{R}$ be a function such that for any $\e >0$,
 \[\mathbb{P}\left(\left \vert f({\mathbf{y}}) - a \right\vert \geq \epsilon \right) \leq 2e^{-{nk\epsilon^2}} \]
 for some constants $a,k,$ such that $k >0$.  Then for any $\e >0$, we have
\begin{align}
\label{eq:1yi_conc}
 &\mathbb{P}\left(\frac{1}{n}\left \vert\sum_{i=1}^n \mathsf{1}_{\left\{ y_i > f({\mathbf{y}})\right\}} - \sum_{i=1}^n \mathsf{1}_{\left\{ y_i > a\right\}}\right\vert  \geq \epsilon  \right)  \leq 4e^{-{nk_1\epsilon^2}},\\
 \label{eq:thetaiyi_conc}
  &\mathbb{P}\left(\frac{1}{n}\left \vert\sum_{i=1}^n  \theta_i \mathsf{1}_{\left\{ y_i > f({\mathbf{y}})\right\}} - \sum_{i=1}^n \theta_i \mathsf{1}_{\left\{y_i > a\right\}}\right\vert  \geq \epsilon  \right)  \leq 4e^{-\frac{nk_1 \epsilon^2}{\Vert \bst\Vert^2/n}}, \\ \nonumber 
  &\mathbb{P}\left(\frac{1}{n}\left \vert\sum_{i=1}^n w_i\mathsf{1}_{\left\{ y_i > f({\mathbf{y}}) \right\}} -\sum_{i=1}^n w_i\mathsf{1}_{\left\{ y_i > a\right\}}\right\vert  \geq \epsilon  \right) \\  \label{eq:wiyi_conc}
  & \leq 4e^{-{nk_1 \min(\epsilon^2,\epsilon)}}
\end{align}
where $k_1$ is a positive constant.
\end{lemma}
 The proof is given in Appendix \ref{pf_lem_app_small_O_P}.

\begin{lemma}
With the assumptions of Lemma \ref{lem_app_small_O_P},  let $h: \mathbb{R}^n \to \mathbb{R}$ be a function such that $b > a$ and $\mathbb{P}\left(\left \vert h({\mathbf{y}}) - b \right\vert \geq \epsilon \right) \leq 2e^{-nl\epsilon^2}$ for some  $l>0$.
Then for any $\e >0$, we have
\begin{align*}
 & \mathbb{P}\left(\frac{1}{n}\left \vert\sum_{i=1}^n \mathsf{1}_{\left\{h({\mathbf{y}}) \geq y_i > f({\mathbf{y}})\right\}} - \sum_{i=1}^n \mathsf{1}_{\left\{b \geq  y_i > a\right\}}\right\vert  \geq \epsilon  \right) \\
 & \leq 8e^{-{nk \epsilon^2}{}}. 
\end{align*}
 \label{lem:fhy_conc}
\end{lemma}
\begin{proof}
The result follows from Lemma \ref{lem_app_small_O_P} by noting that  
$\mathsf{1}_{\left\{h({\mathbf{y}}) \geq y_i > f({\mathbf{y}})\right\}}  =  \mathsf{1}_{\left\{ y_i > f({\mathbf{y}})\right\}}  - \mathsf{1}_{\left\{ y_i >h({\mathbf{y}})\right\}}$, and  $\mathsf{1}_{\left\{ b \geq y_i >  a \right\}}  =  \mathsf{1}_{\left\{ y_i > a \right\}}  - \mathsf{1}_{\left\{ y_i > b\right\}}$.
\end{proof}

\subsection{Proof of Theorem \ref{thm1} } \label{subsec:thm1_proof}

We have, 
  \begin{align}\nonumber
  &\frac{1}{n} \Vert\boldsymbol{\theta} - \hat{\boldsymbol{\theta}}_{JS_{2}} \Vert^2\\ \nonumber
  & =\frac{1}{n} \left\Vert \boldsymbol{\theta} - \boldsymbol{\nu}_{2} - \left[1-\frac{\sigma^2}{g\left({\Vert\mathbf{y}-\boldsymbol{\nu}_{2} \Vert^2}/{n}\right)} \right](\mathbf{y}-\boldsymbol{\nu}_{2})  \right\Vert^2 \\ \nonumber
  &= \frac{1}{n}\left\Vert \mathbf{y} - \boldsymbol{\nu}_{2} - \left[1-\frac{\sigma^2}{g\left({\Vert\mathbf{y}-\boldsymbol{\nu}_{2} \Vert^2}/{n}\right)} \right](\mathbf{y}-\boldsymbol{\nu}_{2})  - \mathbf{w}\right\Vert^2 \\\nonumber 
  & = \frac{1}{n}\left\Vert \left(\frac{\sigma^2}{g\left({\Vert\mathbf{y}-\boldsymbol{\nu}_{2} \Vert^2}/{n}\right)} \right)(\mathbf{y}-\boldsymbol{\nu}_{2})  - \mathbf{w}\right\Vert^2 \\ \nonumber
  &=\frac{\sigma^4{\Vert\mathbf{y}-\boldsymbol{\nu}_{2} \Vert^2}/{n}}{\left(g\left({\Vert\mathbf{y}-\boldsymbol{\nu}_{2} \Vert^2}/{n}\right)\right)^2}  +\frac{\left\Vert \mathbf{w}\right\Vert^2 }{n} \\ \label{eq1_thm1} 
  &\quad - \frac{2}{n} \left(\frac{\sigma^2}{g\left({\Vert\mathbf{y}-\boldsymbol{\nu}_{2} \Vert^2}/{n}\right)} \right)\left\langle\mathbf{y}-\boldsymbol{\nu}_{2}, \mathbf{w}\right \rangle.
  \end{align}
  We also have
\begin{align*}
\frac{1}{n}\left\Vert \boldsymbol{\theta} -  \boldsymbol{\nu}_{2} \right\Vert^2  &= \frac{1}{n}\left\Vert \mathbf{y} -  \boldsymbol{\nu}_{2} - \mathbf{w} \right\Vert^2 \\
&=\frac{1}{n}\left\Vert \mathbf{y} -  \boldsymbol{\nu}_{2} \right\Vert^2  + \frac{1}{n}\left\Vert \mathbf{w} \right\Vert^2  - \frac{2}{n}\left\langle \mathbf{y}-\boldsymbol{\nu}_{2}, \mathbf{w}\right \rangle 
\end{align*}
and so,
\begin{equation}\label{eq2_thm1}
- \frac{2}{n}\left\langle \mathbf{y}-\boldsymbol{\nu}_{2}, \mathbf{w}\right \rangle =  \frac{1}{n}\left\Vert \boldsymbol{\theta} -  \boldsymbol{\nu}_{2} \right\Vert^2  - \frac{1}{n}\left\Vert \mathbf{y} -  \boldsymbol{\nu}_{2} \right\Vert^2  - \frac{1}{n}\left\Vert \mathbf{w} \right\Vert^2 .
\end{equation}
Using \eqref{eq2_thm1} in \eqref{eq1_thm1}, we obtain
\begin{align} \nonumber
& \frac{\Vert\boldsymbol{\theta} - \hat{\boldsymbol{\theta}}_{JS_{2}} \Vert^2}{n}   \\ \nonumber
&= \frac{\sigma^4{\Vert\mathbf{y}-\boldsymbol{\nu}_{2} \Vert^2}/{n}}{\left(g\left({\Vert\mathbf{y}-\boldsymbol{\nu}_{2} \Vert^2}/{n}\right)\right)^2}   +\frac{\left\Vert \mathbf{w}\right\Vert^2 }{n} +  \left(\frac{\sigma^2}{g\left({\Vert\mathbf{y}-\boldsymbol{\nu}_{2} \Vert^2}/{n}\right)} \right)\\ \label{eq3_thm1}
 & \quad  \times  \left(\frac{1}{n}\left\Vert \boldsymbol{\theta} -  \boldsymbol{\nu}_{2} \right\Vert^2  - \frac{1}{n}\left\Vert \mathbf{y} -  \boldsymbol{\nu}_{2} \right\Vert^2  - \frac{1}{n}\left\Vert \mathbf{w} \right\Vert^2\right).
\end{align}
We now use the following results whose proofs are given in Appendix \ref{pf_lem3} and Appendix \ref{pf_lem_theta}.

\begin{lemma}\label{lem3}
\begin{equation}
\label{eq5_thm1}
\begin{split}
 \frac{1}{n}\left\Vert \mathbf{y}-\boldsymbol{\nu}_{2} \right\Vert^2  &\doteq \alpha_n + \sigma^2  +  \kappa_n\delta + o(\delta).
\end{split}
\end{equation}
where  $\alpha_n$ is given by \eqref{eq_alpha2}.
\end{lemma}

\begin{lemma}\label{lem_theta}
 \begin{equation}
 \label{eq6_thm1}
\begin{split}
 \frac{1}{n}\left\Vert \boldsymbol{\theta}-\boldsymbol{\nu}_{2} \right\Vert^2  \doteq \beta_n + \kappa_n\delta+o(\delta)\\
\end{split}
\end{equation}
where $\beta_n$ is given by \eqref{eq_beta2}.
\end{lemma}

Using Lemma \ref{app_lemma_g} together with \eqref{eq5_thm1}, we have 
\begin{equation}\label{eq7_thm1}
  g\left(\frac{\left\Vert \mathbf{y} -  \boldsymbol{\nu}_{2} \right\Vert^2}{n}\right) \doteq g(\alpha_n + \sigma^2) + \kappa_n\delta + o(\delta).
\end{equation}
Using \eqref{eq5_thm1}, \eqref{eq6_thm1} and \eqref{eq7_thm1} together with Lemmas \ref{lem1_app2}, \ref{lem1_app}, and \ref{lem_app_conc_ineq3}, we obtain 
\begin{align} \nonumber
 \frac{\Vert\boldsymbol{\theta} - \hat{\boldsymbol{\theta}}_{JS_{2}} \Vert^2}{n} & \doteq  \frac{\sigma^4\left(\alpha_n +\sigma^2\right)}{\left(g\left(\alpha_n +\sigma^2\right)\right)^2}  +\sigma^2 +  \left(\frac{\sigma^2}{g\left(\alpha_n +\sigma^2\right)} \right)\\ \label{eq4_thm1}
 & \times\left(\beta^2  - \left(\alpha_n +\sigma^2\right)  - \sigma^2\right) + \kappa_n\delta + o(\delta)\\ \nonumber
 & = \left\{ \begin{array}{cc}
              \beta_n + \kappa_n\delta + o(\delta), & \textrm{if }\alpha_n < 0,\\
              \frac{\beta_n\sigma^2}{\alpha_n+\sigma^2} + \kappa_n\delta + o(\delta) & \textrm{otherwise}.\\
             \end{array}\right.
\end{align}
Therefore, for any $\epsilon > 0$,
\begin{align*}
 & \mathbb{P}\left(\left\vert \frac{\Vert\boldsymbol{\theta} - \hat{\boldsymbol{\theta}}_{JS_{2}} \Vert^2}{n} - \min\left(\beta_n,\frac{\beta_n\sigma^2}{\alpha_n+\sigma^2}\right) \right. \right.\\
 &\hspace{0.2 in}+ \kappa_n\delta + o(\delta) \bigg\vert \geq \epsilon \bigg) \leq Ke^{-\frac{nk \min(\epsilon^2,1)}{\Vert \bst\Vert^2/n}}.
\end{align*}
This proves \eqref{eq1_thm1_statement} and hence, the first part of the theorem.

 To prove the second part of the theorem, we use the following definition and result.
\begin{defi}\label{def_UI}
(Uniform Integrability \cite[p. 81]{wasserman}) A sequence $\{X_n\}_{n=1}^{\infty}$ is said to be uniformly integrable (UI) if
\begin{equation}\label{eq_UI}
 \lim_{K \to \infty}\left(\limsup_{n \to \infty} \mathbb{E}\left[\vert X_n \vert \mathbf{1}_{\{\vert X_n \vert \geq K\}} \right]\right) = 0.
\end{equation}
\end{defi}

\begin{fact}\label{fact_UI}\cite[Sec. 13.7]{williams}
Let  $\{X_n\}_{n=1}^{\infty}$ be a sequence in $\mathcal{L}^1$, equivalently $\mathbb{E}\vert X_n \vert < \infty$, $\forall n$. Also, let $X\in \mathcal{L}^1$. Then $X_n \overset{\mathcal{L}^1}{\longrightarrow}  X$, i.e., $\mathbb{E}(\vert X_n - X\vert) \to 0$, if and only if the following two conditions are satisfied:
\begin{enumerate}
 \item $X_n \overset{P}{\longrightarrow}  X$,
 \item The sequence $\{X_n\}_{n=1}^{\infty}$ is UI.
\end{enumerate}
\end{fact}

Now, consider the individual terms of the RHS of \eqref{eq3_thm1}. Using Lemmas \ref{lem1_app2}, \ref{lem1_app} and \ref{app_lemma_g}, we obtain
\begin{equation*}
 \frac{\sigma^4{\Vert\mathbf{y}-\boldsymbol{\nu}_{2} \Vert^2}/{n}}{\left(g\left({\Vert\mathbf{y}-\boldsymbol{\nu}_{2} \Vert^2}/{n}\right)\right)^2} \doteq \frac{\sigma^4\left(\alpha_n +\sigma^2\right)}{\left(g\left(\alpha_n +\sigma^2\right)\right)^2} + \kappa_n\delta + o(\delta),
\end{equation*}
and so, from Lemma  \ref{prop_as},
\begin{align}\nonumber 
 S_n & \vcentcolon=\frac{\sigma^4{\Vert\mathbf{y}-\boldsymbol{\nu}_{2} \Vert^2}/{n}}{\left(g\left({\Vert\mathbf{y}-\boldsymbol{\nu}_{2} \Vert^2}/{n}\right)\right)^2} \\ \label{eq_S_n}
 &\quad - \left[\frac{\sigma^4\left(\alpha_n +\sigma^2\right)}{\left(g\left(\alpha_n +\sigma^2\right)\right)^2} + \kappa_n\delta + o(\delta)\right] \overset{a.s.}{\longrightarrow} 0.
\end{align}
Similarly, we obtain
\begin{align}
\begin{split} \label{eq_thm1_part2}
 T_n  &\vcentcolon = \frac{\sigma^2{\left\Vert \boldsymbol{\theta} -  \boldsymbol{\nu}_{2} \right\Vert^2}/{n}}{g\left({\Vert\mathbf{y}-\boldsymbol{\nu}_{2} \Vert^2}/{n}\right)} \\ 
 & \quad -\left[\frac{\beta_n\sigma^2}{g(\alpha_n+\sigma^2)}+ \kappa_n\delta + o(\delta)\right] \overset{a.s.}{\longrightarrow} 0, \\  
  U_n  &\vcentcolon = \frac{\sigma^2{\left\Vert \mathbf{y} -  \boldsymbol{\nu}_{2} \right\Vert^2}/{n}}{g\left({\Vert\mathbf{y}-\boldsymbol{\nu}_{2} \Vert^2}/{n}\right)}\\
   & \quad - \left[\frac{\sigma^2(\alpha_n+\sigma^2)}{g(\alpha_n+\sigma^2)} + \kappa_n\delta + o(\delta)\right] \overset{a.s.}{\longrightarrow} 0,\\
   V_n  &\vcentcolon =\frac{\sigma^2{\left\Vert \mathbf{w} \right\Vert^2}/{n}}{g\left({\Vert\mathbf{y}-\boldsymbol{\nu}_{2} \Vert^2}/{n}\right)}\\
    & \quad - \left[\frac{\sigma^4}{g(\alpha_n+\sigma^2)}+ \kappa_n\delta + o(\delta)\right] \overset{a.s.}{\longrightarrow} 0.
    \end{split}
 \end{align}
Now, using \eqref{eq3_thm1} and \eqref{eq4_thm1}, we can write
\begin{align*} 
 &\frac{\Vert\boldsymbol{\theta} - \hat{\boldsymbol{\theta}}_{JS_{2}} \Vert^2}{n} - \left(\min\left(
              \beta_n , \frac{\beta_n\sigma^2}{\alpha_n+\sigma^2} \right) + \kappa_n\delta + o(\delta)\right) \\
              &= S_n + T_n - U_n - V_n + \left(\frac{\left\Vert \mathbf{w} \right\Vert^2}{n}- \sigma^2\right). 
\end{align*}
 Note from Jensen's inequality that $\vert \mathbb{E}[X_n] - \mathbb{E}X \vert \leq \mathbb{E}(\vert X_n - X\vert)$.  We therefore have
\begin{align}\nonumber
  &\left \vert\frac{1}{n}R( \boldsymbol{\theta}, \hat{\boldsymbol{\theta}}_{JS_{2}})- \left[\min\left(\beta_n,\frac{\beta_n\sigma^2}{\alpha_n+\sigma^2}\right) + \kappa_n\delta + o(\delta)\right]\right \vert \\ \nonumber
 &\leq  \mathbb{E}\left \vert\frac{\| \boldsymbol{\theta} - \hat{\boldsymbol{\theta}}_{JS_{2}} \|^2}{n}- \left[\min\left(\beta_n,\frac{\beta_n\sigma^2}{\alpha_n+\sigma^2}\right) + \kappa_n\delta + o(\delta)\right]\right \vert \\  \nonumber
 &= \mathbb{E}\left \vert S_n + T_n - U_n-  V_n +\frac{\| \mathbf{w} \|^2}{n} - \sigma^2 \right \vert \\ \label{thm1_eq1} 
 &\leq  \mathbb{E}\left \vert S_n  \right \vert + \mathbb{E}\left \vert T_n\right \vert - \mathbb{E}\left \vert U_n \right \vert - \mathbb{E}\left \vert V_n\right \vert + \mathbb{E}\left \vert \frac{\| \mathbf{w} \|^2}{n} - \sigma^2 \right \vert.
\end{align}

We first show that $\frac{\left\Vert \mathbf{w} \right\Vert^2}{n} \overset{\mathcal{L}^1}{\longrightarrow} \sigma^2$, i.e., 
\begin{equation}\label{eq_chi_ui}
 \lim_{n \to \infty} \mathbb{E}\left[\left\vert \frac{\| \mathbf{w} \|^2}{n} - \sigma^2 \right\vert \right] = 0.
\end{equation}
 This holds because
\begin{align*}
&  \mathbb{E}\left[\left\vert \frac{\| \mathbf{w} \|^2}{n} - \sigma^2 \right\vert \right] = \int_{0}^\infty \mathbb{P}\left( \left\vert \frac{\| \mathbf{w} \|^2}{n} - \sigma^2 \right\vert > x\right)dx \\
 &\overset{(i)}{\leq} \int_{0}^1 2e^{-nkx^2}dx + \int_{1}^\infty 2e^{-nkx}dx\\
 & \leq \int_{0}^{\infty}2e^{-nkx^2}dx + \int_{0}^\infty 2e^{-nkx}dx \\
 &  = \frac{2}{\sqrt{nk}}\int_{0}^{\infty}e^{-t^2}dt + \frac{2}{nk}\int_{0}^\infty e^{-t}dt  \overset{n\to \infty}{\longrightarrow}0,
\end{align*} 
where inequality $(i)$ is due to Lemma \ref{lem_app_conc_ineq3}.

Thus, from  \eqref{thm1_eq1}, to prove \eqref{eq2_thm1_statement}, it is sufficient to show that $\mathbb{E}\left \vert S_n  \right \vert$, $\mathbb{E}\left \vert T_n\right \vert$, $\mathbb{E}\left \vert U_n \right \vert$ and $\mathbb{E}\left \vert V_n\right \vert$ all converge to $0$ as $n \to \infty$.  From Fact \ref{fact_UI} and \eqref{eq_S_n}, \eqref{eq_thm1_part2}, this implies that we need to show that $\{S_n\}_{n=1}^{\infty}, \{T_n\}_{n=1}^{\infty}, \{U_n\}_{n=1}^{\infty}, \{V_n\}_{n=1}^{\infty}$ are UI.  Considering $S_n$, we have
\begin{equation*}
\frac{\sigma^4{\Vert\mathbf{y}-\boldsymbol{\nu}_{2} \Vert^2}/{n}}{\left(g\left({\Vert\mathbf{y}-\boldsymbol{\nu}_{2} \Vert^2}/{n}\right)\right)^2} \leq \sigma^2,~~~~~  \frac{\sigma^4\left(\alpha_n +\sigma^2\right)}{\left(g\left(\alpha_n +\sigma^2\right)\right)^2} \leq \sigma^2,~~~\forall n,
\end{equation*}
and since the sum of the terms in \eqref{eq_S_n} that involve $\delta$ have bounded absolute value for a chosen and fixed $\delta$ (see Note \ref{note_delta_bound}), there exists $M > 0$ such that $\forall n$, $\vert S_n \vert \leq 2\sigma^2 +M$. Hence, from Definition \ref{def_UI}, $\{S_n\}_{n=1}^{\infty}$ is UI. By a similar argument, so is $\{U_n\}_{n=1}^{\infty}$. Next, considering $V_n$, we have 
\begin{equation*}
 \frac{\sigma^2{\left\Vert \mathbf{w} \right\Vert^2}/{n}}{g\left({\Vert\mathbf{y}-\boldsymbol{\nu}_{2} \Vert^2}/{n}\right)} \leq \frac{\left\Vert \mathbf{w} \right\Vert^2}{n}, ~~~~~~ \frac{\sigma^4}{g(\alpha_n+\sigma^2)} \leq \sigma^2, ~~~~\forall n,
 \end{equation*}
and hence, $\vert V_n \vert \leq \frac{\left\Vert \mathbf{w} \right\Vert^2}{n} + \sigma^2 + M$, $\forall n$. Note from \eqref{eq_chi_ui} and Fact \ref{fact_UI}  that $\{\Vert \mathbf{w} \Vert^2/n\}_{n=1}^\infty$ is UI. To complete the proof, we use  the following result whose proof is provided in Appendix \ref{pf_lem_ui}.

\begin{lemma}\label{lem_ui}
 Let $\{Y_n \}_{n=1}^{\infty}$ be a UI sequence of positive-valued random variables, and let $\{X_n \}_{n=1}^{\infty}$ be a sequence of random variables such that $\vert X_n \vert \leq c Y_n + a$, $\forall n$, where $c$ and $a$ are positive constants. Then, $\{X_n \}_{n=1}^{\infty}$ is also UI.
\end{lemma}

\noindent Hence, $\{V_n \}_{n=1}^{\infty}$ is UI. Finally, considering $T_n$ in \eqref{eq_thm1_part2}, we see that
 \begin{align*}
  \frac{\sigma^2\frac{\Vert\boldsymbol{\theta}-\boldsymbol{\nu}_{2} \Vert^2}{n}}{g\left(\frac{\Vert\mathbf{y}-\boldsymbol{\nu}_{2} \Vert^2}{n}\right)} = \frac{\sigma^2\frac{\Vert\mathbf{y}-\boldsymbol{\nu}_{2} - \mathbf{w} \Vert^2}{n}}{g\left(\frac{\Vert\mathbf{y}-\boldsymbol{\nu}_{2} \Vert^2}{n}\right)}  &\leq \frac{2\sigma^2\left(\frac{\Vert\mathbf{y}-\boldsymbol{\nu}_{2}\Vert^2}{n} + \frac{\Vert \mathbf{w} \Vert^2}{n}\right)}{g\left(\frac{\Vert\mathbf{y}-\boldsymbol{\nu}_{2} \Vert^2}{n}\right)}\\
  &\leq 2\left(\sigma^2 + \frac{\Vert \mathbf{w} \Vert^2}{n} \right), \\
  \frac{\beta_n\sigma^2}{g(\alpha_n+\sigma^2)} & \leq \beta_n < \infty.
 \end{align*}
 Note that the last inequality is due to the assumption that $\limsup_{n \to \infty}\Vert \bst \Vert^2/n < \infty$. Therefore, $\vert T_n \vert \leq 2{\left\Vert \mathbf{w} \right\Vert^2}/{n} + 2\sigma^2 + M$, $\forall n$, where $M$ is some finite constant.  Thus, by Lemma  \ref{lem_ui}, $T_n$ is UI. Therefore, each of the terms of the RHS of \eqref{thm1_eq1} goes to $0$ as $n\to \infty$, and this completes the proof of the theorem.

\subsection{Proof of Theorem \ref{thm2} } \label{subsec:thm2_proof}

Let 
\begin{align*}
 \mathcal{E}_n &\vcentcolon= \left\{ \mathbf{y} \in \mathbb{R}^n  \ : \  \Vert \bst - \bsth_{JS_1} \Vert^2 <  \Vert \bst - \bsth_{JS_2} \Vert^2 \right\}, \\
  \Delta_n &\vcentcolon= \frac{\norm{\boldsymbol{\theta} - \hat{\boldsymbol{\theta}}_{JS_{2}}}^2}{n} -\frac{\norm{\boldsymbol{\theta} - \hat{\boldsymbol{\theta}}_{JS_{1}}}^2}{n}.
\end{align*}
Without loss of generality, for a given $\epsilon > 0$, we can assume that $\vert \Delta_n \vert > \epsilon$ because if not, it is clear that 
\[ \frac{\norm{\boldsymbol{\theta} - \hat{\boldsymbol{\theta}}_{JS_{H}}}^2}{n} -  \min\left(\frac{\norm{\boldsymbol{\theta} - \hat{\boldsymbol{\theta}}_{JS_{1}}}^2}{n}, \frac{\norm{\boldsymbol{\theta} - \hat{\boldsymbol{\theta}}_{JS_{2}}}^2}{n}\right) \leq \epsilon. \]

From \eqref{eq:gyybar_conv} and Lemma \ref{lem1_app}, we obtain the following concentration inequality for the loss estimate in \eqref{eq_est_1_partition_risk}:
\ben 
 \label{eq_1_est1_part_risk}
 \frac{1}{n}\hat{L}(\boldsymbol{\theta},  \hat{\boldsymbol{\theta}}_{JS_{1}}) \doteq \frac{\rho_n\sigma^2}{\rho_n+\sigma^2}. 
 \een
Using this together with Corollary \ref{cor_pp_JS_Lindley}, we obtain
 \be 
 \label{eq_est1_part_risk}
 \frac{1}{n}\hat{L}(\boldsymbol{\theta},  \hat{\boldsymbol{\theta}}_{JS_{1}}) \doteq \frac{1}{n}\norm{\boldsymbol{\theta} - \hat{\boldsymbol{\theta}}_{JS_{1}}}^2. 
 \ee
Following steps similar to those in the proof of Lemma \ref{lem3}, we obtain the following  for the loss estimate in \eqref{eq_est_2_partition_risk}:
\begin{equation}\label{eq_2_est_2_partition_risk1}
 \frac{1}{n}\hat{L}(\boldsymbol{\theta},  \hat{\boldsymbol{\theta}}_{JS_{2}}) \doteq \frac{\beta_n\sigma^2}{g(\alpha_n+\sigma^2)} + \kappa_n\delta + o(\delta).
\end{equation}
Combining this with Theorem \ref{thm1}, we have
\begin{equation}\label{eq_est_2_partition_risk1}
 \frac{1}{n}\hat{L}(\boldsymbol{\theta},  \hat{\boldsymbol{\theta}}_{JS_{2}}) \doteq \frac{1}{n}\norm{\boldsymbol{\theta} - \hat{\boldsymbol{\theta}}_{JS_{2}}}^2.
\end{equation}

\noindent Then, from \eqref{eq_est1_part_risk}, \eqref{eq_est_2_partition_risk1}, and Lemma \ref{lem_4}, we have $
\frac{1}{n}\hat{L}(\boldsymbol{\theta},  \hat{\boldsymbol{\theta}}_{JS_{1}}) - \frac{1}{n}\hat{L}(\boldsymbol{\theta},  \hat{\boldsymbol{\theta}}_{JS_{2}}) \doteq -\Delta_n
$. We therefore have, for any $\epsilon > 0$,
\begin{align}\nonumber 
 \mathbb{P}\left(\frac{1}{n}\hat{L}(\boldsymbol{\theta},  \hat{\boldsymbol{\theta}}_{JS_{1}}) - \frac{1}{n}\hat{L}(\boldsymbol{\theta},  \hat{\boldsymbol{\theta}}_{JS_{2}}) -(-\Delta_n) \geq \epsilon \right) \\ \label{eq_3_thm2}
 \leq Ke^{-\frac{nk\min(\epsilon^2,1)}{\max(\Vert \bst\Vert^2/n,1)}}
\end{align}
for some positive constants $k$ and $K$. Let $\mathbb{P}\left( \gamma_\by =0, \Delta_n > \e \right)$ denote the probability that $\gamma_\by = 0$ and $\Delta_n > \e$ for a chosen $\e > 0$. Therefore,
\begin{align}\nonumber 
 &\mathbb{P}\left( \gamma_\by =0, \Delta_n >\e   \right) \\ \nonumber
&=  \mathbb{P}\left(  \frac{1}{n}\hat{L}(\boldsymbol{\theta},  \hat{\boldsymbol{\theta}}_{JS_{1}}) > \frac{1}{n}\hat{L}(\boldsymbol{\theta},  \hat{\boldsymbol{\theta}}_{JS_{2}}) \ , \ \Delta_n >\e   \right) \\ \nonumber 
 & \leq \mathbb{P}\left( \frac{1}{n}\hat{L}(\boldsymbol{\theta},  \hat{\boldsymbol{\theta}}_{JS_{1}}) +  \Delta_n > \frac{1}{n}\hat{L}(\boldsymbol{\theta},  \hat{\boldsymbol{\theta}}_{JS_{2}}) + \e  \right) \\ \label{eq1_thm3}
  & \leq Ke^{-\frac{nk\min(\epsilon^2,1)}{\max(\Vert \bst\Vert^2/n,1)}},
\end{align}
where the last inequality is obtained from  \eqref{eq_3_thm2}. So for any $\epsilon > 0$, we have
\begin{align*}
 &\mathbb{P}\left( \left. \frac{\norm{\boldsymbol{\theta} - \hat{\boldsymbol{\theta}}_{JS_{H}}}^2}{n} - \frac{\norm{\boldsymbol{\theta} - \hat{\boldsymbol{\theta}}_{JS_{1}}}^2}{n}\geq  \epsilon ~ \right \vert ~ \by \in  \mathcal{E}_n\right) \\
 & \leq \mathbb{P}\left( \gamma_\by =0 , \Delta_n >\e \right) \leq Ke^{-\frac{nk\min(\epsilon^2,1)}{\max(\Vert \bst\Vert^2/n,1)}}.
 \end{align*}
In a similar manner, we obtain for any $\epsilon > 0$,
\begin{align*}
 &  \mathbb{P}\left( \left. \frac{\norm{\boldsymbol{\theta} - \hat{\boldsymbol{\theta}}_{JS_{H}}}^2}{n} - \frac{\norm{\boldsymbol{\theta} - \hat{\boldsymbol{\theta}}_{JS_{2}}}^2}{n}\geq  \epsilon ~ \right \vert ~ \by \in \mathcal{E}_n^c\right) \\
 &  \leq \mathbb{P}\left( \gamma_\by =1 , -\Delta_n >\e \right) \leq Ke^{-\frac{nk\min(\epsilon^2,1)}{\max(\Vert \bst\Vert^2/n,1)}}.
 \end{align*}
Therefore, we arrive at
\begin{align*}
  & \mathbb{P}\left( \frac{\norm{\boldsymbol{\theta} - \hat{\boldsymbol{\theta}}_{JS_{H}}}^2}{n} - \min_{i=1,2}\left(\frac{\norm{\boldsymbol{\theta} - \hat{\boldsymbol{\theta}}_{JS_{i}}}^2}{n}\right)\geq  \epsilon \right) \\
  &  \leq Ke^{-\frac{nk\min(\epsilon^2,1)}{\max(\Vert \bst\Vert^2/n,1)}}.
 \end{align*}
This proves the first part of the theorem.

For the second part, fix $\epsilon > 0$.  First suppose that $\hat{\boldsymbol{\theta}}_{JS_{1}}$ has lower risk. 
For a given  $\bst$, let 
\begin{align*} 
\mathcal{A}_{JS_{1}}(\bst) & := \{ \by \in \mathbb{R}^n: \gamma_\by =1  \},  \\
\mathcal{A}_{JS_{2a}}(\bst) & := \left \{ \by \in \mathbb{R}^n: \gamma_\by =0,  \ \text{ and} \ 
\Delta_n \leq \e \right\}, 
\\
\mathcal{A}_{JS_{2b}}(\bst) &:=\mathbb{R}^n \backslash (\mathcal{A}_{JS_{1}}(\bst) \cup  \mathcal{A}_{JS_{2a}}(\bst)) \\
&= \left \{ \by \in \mathbb{R}^n: \gamma_\by =0,  \ \text{and} \ \Delta_n > \e   \right \}.
\end{align*} Denoting $\frac{1}{\sqrt{2\pi \sigma^2}}\exp\left(-\frac{\Vert \by - \bst \Vert^2}{2\sigma^2}\right)$ by $\phi(\by;\bst)$, we have
\begin{align}\nonumber
 R& (\boldsymbol{\theta}, \hat{\bst}_{JS_{H}})  = \int_{\mathcal{A}_{JS_{1}}(\bst)} \hspace{-1pt} \phi(\by; \bst) \norm{\bsth_{JS_{1}} - \bst}^2 \ d \by  \\ \nonumber
  & \quad +  \int_{ \mathcal{A}_{JS_{2a}}(\bst) \,  \cup \, \mathcal{A}_{JS_{2b}}(\bst) }   \hspace{-1pt}  \phi(\by; \bst) \norm{\bsth_{JS_{2}} - \bst}^2 \ d\by \\ \nonumber 
  & \overset{(a)}{\leq}  \int_{\mathcal{A}_{JS_{1}}(\bst)} \hspace{-1pt} \phi(\by; \bst) \norm{\bsth_{JS_{1}} - \bst}^2 \ d \by \\ \nonumber
  & \quad + \int_{\mathcal{A}_{JS_{2a}} (\bst)} \hspace{-1pt} \phi(\by; \bst) \, (\norm{\bsth_{JS_{1}} - \bst}^2 + n \e) \ d \by \\ \nonumber 
  & \quad + \int_{\mathcal{A}_{JS_{2b}}(\bst)} \hspace{-1pt} \phi(\by; \bst) \norm{\bsth_{JS_{2}} - \bst}^2 \ d \by   \\ \nonumber
  & \overset{(b)}{\leq} R(\boldsymbol{\theta}, \hat{\bst}_{JS_{1}}) + n\e+ \left(\mathbb{P}\left( \gamma_\by =0, \Delta_n >\e \right) \right)^{1/2}\\ \nonumber
  & \quad \times \left(\int_{\mathcal{A}_{JS_{2b}}(\bst)}   \hspace{-5pt}  \phi(\by; \bst) \norm{\bsth_{JS_{2}} - \bst}^4 \ d\by \right)^{1/2} \nonumber \\ \nonumber 
  & \overset{(c)}{\leq} R(\boldsymbol{\theta}, \hat{\bst}_{JS_{1}}) + n\e \\ \label{eq_5_thm2}
  & \quad + Ke^{-\frac{nk\min(\epsilon^2,1)}{\max(\Vert \bst\Vert^2/n,1)}} \left(\mathbb{E} \norm{\bsth_{JS_{2}} - \bst}^4  \right)^{1/2}
  \end{align}
where step $(a)$ uses the definition of $\mathcal{A}_{JS_{2a}}$, in step $(b)$ the last term  is obtained using  the Cauchy-Schwarz inequality on the product of the functions  $\sqrt{\phi(\by; \bst)}$, and $\sqrt{\phi(\by; \bst)} \norm{\bsth_{JS_{2}} - \bst}^2$. Step $(c)$ is from \eqref{eq1_thm3}.

Similarly, when $\hat{\boldsymbol{\theta}}_{JS_{2}}$ has lower risk, we get
\begin{align}\nonumber
 R( \boldsymbol{\theta}, \hat{\boldsymbol{\theta}}_{JS_{H}}) & \leq  R( \boldsymbol{\theta}, \hat{\boldsymbol{\theta}}_{JS_{2}})  + n\e \\ \label{eq_6_thm2}
 &+Ke^{-\frac{nk\min(\epsilon^2,1)}{\max(\Vert \bst\Vert^2/n,1)}} \left(\mathbb{E} \norm{\bsth_{JS_{1}} - \bst}^4  \right)^{1/2}.
\end{align}
Hence, from \eqref{eq_5_thm2}-\eqref{eq_6_thm2}, we obtain
\begin{align*}
  &\frac{1}{n}R( \boldsymbol{\theta}, \hat{\boldsymbol{\theta}}_{JS_{H}}) \leq  \frac{1}{n}\left[\min_{i=1,2}\left(R( \boldsymbol{\theta}, \hat{\boldsymbol{\theta}}_{JS_{i}})\right) + n\e \right.\\
 & +Ke^{-\frac{nk\min(\epsilon^2,1)}{\max(\Vert \bst\Vert^2/n,1)}} \max_{i=1,2}\left(\left(\mathbb{E} \norm{\bsth_{JS_{i}} - \bst}^4  \right)^{1/2}\right) \bigg].
\end{align*}
Now, noting that by assumption, $\limsup_{n \to \infty}\left(\mathbb{E} \norm{\bsth_{JS_{i}} - \bst}^4  \right)^{1/2}/n$ is finite, we get 
\begin{align*}
 \limsup_{n \to \infty}\frac{1}{n}\left[R( \boldsymbol{\theta}, \hat{\boldsymbol{\theta}}_{JS_{H}}) - \min_{i=1,2}\left(R( \boldsymbol{\theta}, \hat{\boldsymbol{\theta}}_{JS_{i}})\right) - \e \right] &\leq 0.
\end{align*}
Since this is true for every $\epsilon > 0$, we therefore have
\begin{equation}\label{eq5_thm2}
 \limsup_{n \to \infty}\frac{1}{n}\left[R( \boldsymbol{\theta}, \hat{\boldsymbol{\theta}}_{JS_{H}})- \min_{i=1,2}\left(R( \boldsymbol{\theta}, \hat{\boldsymbol{\theta}}_{JS_{i}})\right)\right] \leq 0.
\end{equation}
\noindent This completes the proof of the theorem.
\begin{note}
 Note that in the best case scenario, $\Vert \bst -\bsth_{JS_{H}} \Vert^2 = \min \left(\Vert \bst -\bsth_{JS_{1}} \Vert^2, \Vert \bst -\bsth_{JS_{2}} \Vert^2 \right)$, which occurs when for each realization of $\by$, the hybrid estimator picks the better of the two rival estimators $\bsth_{JS_{1}}$ and $\bsth_{JS_{2}}$. In this case, the inequality in \eqref{eq5_thm2} is strict, provided that there are realizations of $\by$ with non-zero probability measure for which one estimator is strictly better than the other.
\end{note}

\subsection{Proof of Theorem \ref{thm3} } \label{subsec:thm3_proof}

The proof is similar to that of Theorem \ref{thm1}, so we only provide a sketch. Note that for $a_i,b_i$, real-valued and finite, $i = 1, \cdots,n$, with $a_i < b_i$,
\begin{align*}
 \frac{1}{n}\sum_{i=1}^n \mathbb{E}\left[\mathsf{1}_{\left\{a_i < w_i \leq b_i\right\}} \right] = \frac{1}{n}\sum_{i=1}^n\left[Q\left(\frac{a_i}{\sigma}\right) -Q\left(\frac{b_i}{\sigma}\right)\right], \\
 \frac{1}{n}\sum_{i=1}^n  \mathbb{E}\left[w_i\mathsf{1}_{\left\{a_i < w_i \leq b_i\right\}}\right] = \frac{\sigma}{n\sqrt{2\pi}}\sum_{i=1}^n\left(e^{-\frac{a_i^2}{2\sigma^2}} - e^{-\frac{b_i^2}{2\sigma^2}}\right).
\end{align*}
Since $\mathsf{1}_{\left\{a_i < w_i \leq b_i\right\}} \in [0,1]$, it follows that $w_i\mathsf{1}_{\left\{a_i < w_i \leq b_i\right\}} \in [m_i,n_i]$ where $m_i = \min(0,a_i)$, $n_i = \max(0,b_i)$. So, from Hoeffding's inequality, we obtain
\begin{align*}
&\mathbb{P}\left(\frac{1}{n}\left\vert\sum_{i=1}^n \mathsf{1}_{\left\{a_i < w_i \leq b_i\right\}} -\sum_{i=1}^n\left[Q\left(\frac{a_i}{\sigma}\right) -Q\left(\frac{b_i}{\sigma}\right)\right]\right \vert \geq \epsilon\right)\\ & \leq 2e^{- 2n\epsilon^2}, \\
 & \mathbb{P}\left(\frac{1}{n}\left\vert\sum_{i=1}^n w_i\mathsf{1}_{\left\{a_i < w_i \leq b_i\right\}} -\frac{\sigma}{\sqrt{2\pi}}\sum_{i=1}^n\left(e^{-\frac{a_i^2}{2\sigma^2}} - e^{-\frac{b_i^2}{2\sigma^2}}\right)\right \vert \right. \\ 
 & \hspace{0.2 in } \geq \epsilon \bigg) \leq 2e^{- \frac{2n\epsilon^2}{\sum_{i=1}^n\left(n_i-m_i\right)^2}} .
\end{align*}
Subsequently, the steps of Lemma \ref{lem3} are used to obtain 
\begin{align} \nonumber
  &\frac{1}{n}\| \mathbf{y}- \boldsymbol{\nu}_{\mathbf{y}_{L}} \|^2 \\ \nonumber 
  & \doteq \frac{\Vert \boldsymbol{\theta} \Vert^2}{n} - \sum_{j=0}^L\frac{ c_j^2}{n}\sum_{i=1}^n\left[Q\left(\frac{\mu_{j}-\theta_i}{\sigma}\right) -Q\left(\frac{\mu_{j-1}-\theta_i}{\sigma}\right)\right] \\ \nonumber
 &\quad- \left(\frac{2}{n}\right)\left(\frac{\sigma}{\sqrt{2\pi}} \right)\sum_{j=1}^Lc_j\sum_{i=1}^n\left[e^{-\frac{\left(\mu_{j}-\theta_i\right)^2}{2\sigma^2}}-e^{-\frac{\left(\mu_{j-1}-\theta_i\right)^2}{2\sigma^2}}\right]\\ \label{eq_thm3}
 &\quad +\kappa_n\delta + o(\delta).
\end{align}
Finally,  employing the steps of Lemma \ref{lem_theta}, we get 
\begin{align*}
 & \frac{1}{n}\left\Vert \boldsymbol{\theta}- \boldsymbol{\nu}_{\mathbf{y}_{L}} \right\Vert^2 \\
& \doteq  \frac{\Vert \boldsymbol{\theta} \Vert^2}{n}  - \sum_{j=0}^L\frac{c_j^2}{n}\sum_{i=1}^n\left[Q\left(\frac{\mu_{j}-\theta_i}{\sigma}\right) -Q\left(\frac{\mu_{j-1}-\theta_i}{\sigma}\right)\right] \\
& \quad + \kappa_n \delta + o(\delta).
\end{align*}
The subsequent steps of the proof are along the lines of that of Theorem \ref{thm1}. 

\section{Concluding remarks}\label{sec_conc}

In this paper, we presented a class of shrinkage estimators that take advantage of the large dimensionality  to infer the clustering structure of the parameter values from the data. This structure is then used to construct an  attracting vector for the shrinkage estimator. A good cluster-based attracting vector enables significant risk reduction over the ML-estimator even when $\bst$ is composed of several inhomogeneous quantities.

We obtained concentration bounds for the squared-error loss of the constructed estimators and convergence results for the risk.  The estimators have significantly smaller risks than the regular JS-estimator for a wide range of $\boldsymbol{\theta} \in \mathbb{R}^n$, even though they do not dominate the regular (positive-part) JS-estimator for finite $n$. 

An important next step is to test the performance of the proposed estimators on real data sets.  It would be interesting to  adapt these estimators and analyze their risks when the sample values are bounded by a known value, i.e., when $\vert \theta_i \vert \leq \tau$, $\forall i = 1,\cdots,n$, with $\tau$ known. Another open question is how one should decide the maximum number of clusters to be considered for the hybrid estimator.

An interesting direction for future research is to study confidence sets centered on the estimators in this paper, and compare them to  confidence sets  centered on the positive-part JS-estimator, which were studied in \cite{hwangC82,samworth05small}.

The James-Stein estimator for colored Gaussian noise, i.e., for $\mathbf{w} \sim \mathcal{N}\left(\mathbf{0},\boldsymbol{\varSigma} \right)$ with $\boldsymbol{\varSigma}$ known, has been studied in \cite{bock}, and  variants have been proposed in \cite{manton}, \cite{eldar1}. It would be interesting to extend the ideas in this paper to the case of colored Gaussian noise, and to noise that has a general sub-Gaussian distribution.  Yet another research direction is to construct multi-dimensional target subspaces from the data that are more general than the cluster-based subspaces proposed here. The goal is to obtain greater risk savings for a wider range of $\boldsymbol{\theta} \in \mathbb{R}^n$, at the cost of having a more complex attractor. 


\appendix

\section{Proofs of General Lemmas}\label{sec_appendix}

\subsection{Proof of Lemma \ref{lem_app_conc_ineq1}}\label{pf_lem_app_conc_ineq1}
 Note that $\mathbb{E}\left[ w_i\mathsf{1}_{\left\{w_i > a_i\right\}}\right] = \frac{\sigma}{\sqrt{2\pi}}e^{-\frac{a_i^2}{2\sigma^2}}$. So, with
 \begin{equation*}
 X \vcentcolon= \sum_{i=1}^n w_i\mathsf{1}_{\left\{w_i > a_i\right\}} -\frac{\sigma}{\sqrt{2\pi}}\sum_{i=1}^ne^{-\frac{a_i^2}{2\sigma^2}},
 \end{equation*}
 we have $\mathbb{E}X =0$. Let $m_i \vcentcolon = \frac{\sigma}{\sqrt{2\pi}}e^{-\frac{a_i^2}{2\sigma^2}}$, and consider the moment generating function (MGF) of $X$. We have
 \begin{equation}
 \begin{split}
 &  \mathbb{E}\left[e^{\lambda X}\right] = \prod_{i=1}^n\frac{e^{-\lambda m_i}}{\sqrt{2\pi\sigma^2}}\int_{-\infty}^{\infty}e^{\left(\lambda w_i\mathsf{1}_{\left\{w_i > a_i\right\}}\right)}e^{-\frac{w_i^2}{2\sigma^2}}dw_i \\
 & = \prod_{i=1}^n\frac{e^{-\lambda m_i}}{\sqrt{2\pi\sigma^2}}\left[\int_{a_i}^{\infty}e^{\lambda w_i}e^{-\frac{w_i^2}{2\sigma^2}}dw_i + \int_{-\infty}^{a_i}e^{-\frac{w_i^2}{2\sigma^2}}dw_i\right] \\
  & = \prod_{i=1}^ne^{-\lambda m_i}\left[\frac{1}{\sqrt{2\pi\sigma^2}}\int_{a_i}^{\infty}e^{\lambda w_i}e^{-\frac{w_i^2}{2\sigma^2}}dw_i + 1- Q\left(\frac{a_i}{\sigma}\right)\right] \\
  &  = \prod_{i=1}^ne^{-\lambda m_i}\left[e^{\frac{\lambda^2\sigma^2}{2}}Q\left(\frac{a_i}{\sigma}-\lambda\sigma\right) +1- Q\left(\frac{a_i}{\sigma}\right)\right]. 
  \end{split}
    \label{eq_lem_con_ineq0}
 \end{equation}
 Now, for any positive real number $b$, consider the function
 \begin{equation*}
  f(x;b) = e^{-\frac{b}{\sqrt{2\pi}}e^{-\frac{x^2}{2}}}\left[e^{\frac{b^2}{2}}Q\left(x-b\right) +1- Q\left(x\right)\right].
 \end{equation*}
 Note that the RHS of \eqref{eq_lem_con_ineq0} can be written as $\prod_{i=1}^nf( \frac{a_i}{\sigma}; \lambda \sigma)$. We will bound the MGF in \eqref{eq_lem_con_ineq0} by bounding $f(x;b)$. 
 
Clearly, $f(-\infty;b) = e^{\frac{b^2}{2}}$, and since $b > 0$, we have for $x \leq 0$, 
\begin{align*}
&  f(x;b) <  e^{-\frac{b}{\sqrt{2\pi}}e^{-\frac{x^2}{2}}}\left[e^{\frac{b^2}{2}}Q\left(x-b\right) +e^{\frac{b^2}{2}}\left(1- Q\left(x\right)\right)\right]\\
 & = \left(e^{-\frac{b}{\sqrt{2\pi}}e^{-\frac{x^2}{2}}}\right)\left(e^{\frac{b^2}{2}}\right)\left[Q\left(x-b\right) +1- Q\left(x\right)\right] \\
 & = \left(e^{-\frac{b}{\sqrt{2\pi}}e^{-\frac{x^2}{2}}}\right)\left(e^{\frac{b^2}{2}}\right)\left(1 + \int_{x-b}^{x}\frac{1}{\sqrt{2\pi}}e^{-\frac{u^2}{2}}du\right)\\
 &\overset{(i)}{=} \left(e^{-\frac{b}{\sqrt{2\pi}}e^{-\frac{x^2}{2}}}\right)\left(e^{\frac{b^2}{2}}\right)\left(1 + \frac{b}{\sqrt{2\pi}}e^{-\frac{c^2}{2}}\right) \\
 & \overset{(j)}{ \leq} \left(e^{-\frac{b}{\sqrt{2\pi}}e^{-\frac{x^2}{2}}}\right)\left(e^{\frac{b^2}{2}}\right)\left(e^{\frac{b}{\sqrt{2\pi}}e^{-\frac{c^2}{2}}}\right) \overset{(k)}{ <}e^{\frac{b^2}{2}}
\end{align*}
where $(i)$ is from the first mean value theorem for integrals for some $c \in (x-b,x)$, $(j)$ is because $e^x \geq 1 + x$ for $x \geq 0$, and $(k)$ is because for $x \leq 0$, $e^{-x^2} > e^{-(x-b)^2}$ for $b>0$. Therefore,
\begin{equation}\label{eq_lem_con_ineq1}
 \sup_{x \in (-\infty,0]}f(x;b) = e^{\frac{b^2}{2}}.
\end{equation}
Now, for $x \geq 0$, consider 
\begin{align*}
h(x) & \vcentcolon= \frac{f(-x;b) - f(x;b)}{e^{-\frac{b}{\sqrt{2\pi}}e^{-\frac{x^2}{2}}}} \\
 &= e^{\frac{b^2}{2}}\left[Q\left(-x-b\right) -Q\left(x-b\right) \right] + Q\left(x\right) - Q\left(-x\right).
\end{align*}
We have $h(0) =0$ and 
\begin{align*}
 \sqrt{2\pi}\frac{dh(x)}{dx} &= e^{\frac{b^2}{2}}\left[e^{-\frac{(x+b)^2}{2}} + e^{-\frac{(x-b)^2}{2}}\right] - 2e^{-\frac{x^2}{2}}\\
 & = e^{\frac{-x^2}{2}}\left[e^{-bx} + e^{bx}\right] - 2e^{-\frac{x^2}{2}} \\
 & =  e^{\frac{-x^2}{2}}\left[e^{-bx} + e^{bx}-2\right] \\
 & = 2e^{\frac{-x^2}{2}}\left[\cosh (bx) -1\right] \geq 0 
\end{align*}
because $\cosh(bx) \geq 1$. This establishes that $h(x)$ is monotone non-decreasing in $[0,\infty)$ with $h(0) = 0$, and hence, for $x \in [0,\infty)$, 
\begin{equation}\label{eq_lem_con_ineq2}
 h(x) \geq 0 \Rightarrow  f(-x;b) \geq f(x;b).
\end{equation}
Finally, from \eqref{eq_lem_con_ineq1} and \eqref{eq_lem_con_ineq2}, it follows that 
\begin{equation}\label{eq_lem_con_ineq3}
 \sup_{x \in (-\infty,\infty)}f(x;b) = e^{\frac{b^2}{2}}.
\end{equation}
\noindent Using \eqref{eq_lem_con_ineq3} in \eqref{eq_lem_con_ineq0}, we obtain 
$ \mathbb{E}\left[e^{\lambda X}\right] \leq e^{\frac{n\lambda^2\sigma^2}{2}}$.  Hence, applying the Chernoff trick, we have for $\lambda >0$:
 \begin{equation*}
 \begin{split}
 \mathbb{P}\left( X \geq \epsilon  \right) & = \mathbb{P}\left( e^{\lambda X}  \geq e^{\lambda \epsilon}  \right)  \leq \frac{ \mathbb{E}\left[e^{\lambda X}\right]}{ e^{\lambda \epsilon}}  \leq e^{-\left(\lambda \epsilon -  \frac{n\lambda^2\sigma^2}{2} \right)}.
 \end{split}
 \end{equation*}
Choosing $\lambda = \frac{\epsilon}{n\sigma^2}$ which minimizes $ e^{-\left(\lambda \epsilon -  \frac{n\lambda^2\sigma^2}{2} \right)}$, we get $\mathbb{P}\left( X \geq \epsilon  \right) \leq e^{-\frac{\epsilon^2}{2n\sigma^2}}$ and so, 
\begin{align} \nonumber 
& \mathbb{P}\left( \frac{X}{n} \geq \epsilon  \right) \\ \nonumber
&= \mathbb{P}\left(\frac{1}{n}\left(\sum_{i=1}^n w_i\mathsf{1}_{\left\{w_i > a_i\right\}} -\frac{\sigma}{\sqrt{2\pi}}\sum_{i=1}^ne^{-\frac{a_i^2}{2\sigma^2}}\right )\geq \epsilon\right)\\ \label{eq_app_upp1}
& \leq e^{-\frac{n\epsilon^2}{2\sigma^2}}.
\end{align}
To obtain the lower tail inequality, we  use  the following result:
\begin{fact}
 \cite[Thm. 3.7]{chung}. For independent random variables $X_i$ satisfying $X_i \geq -M$, for $1 \leq i \leq n$, we have for any $\epsilon > 0$,
\begin{equation*}
\mathbb{P}\left(\sum_{i=1}^nX_i - \sum_{i=1}^n\mathbb{E}[X_i] \leq -\epsilon \right)  \leq e^{- \frac{\epsilon^2}{2\left(\sum_{i=1}^n\mathbb{E}\left[X_i^2 \right] + \frac{M\epsilon}{3}\right)}  }.
\end{equation*}
\end{fact}

So, for $X_i = w_i\mathsf{1}_{\left\{w_i > a_i\right\}}$, we have $X_i \geq \min\{0,a_i,i=1,\cdots,n\}$, and $\mathbb{E}\left[X_i^2 \right] \leq \sigma^2$, $\forall i = 1,\cdots,n$. Clearly, we can take $M= -\min\{0,a_i,i=1,\cdots,n\} < \infty$. Therefore, for any $  \epsilon > 0$, 
\begin{align*}
\mathbb{P}\left(\sum_{i=1}^nX_i - \sum_{i=1}^n\mathbb{E}[X_i] \leq -\epsilon \right) & \leq e^{- \frac{\epsilon^2}{2\left(\sum_{i=1}^n\mathbb{E}\left[X_i^2 \right] + \frac{M\epsilon}{3}\right)}} \\
& \leq e^{- \frac{\epsilon^2}{2\left(n\sigma^2  + \frac{M\epsilon}{3}\right)}}
\end{align*}
and hence,
\begin{align}\nonumber 
& \mathbb{P}\left(\frac{1}{n}\left(\sum_{i=1}^n w_i\mathsf{1}_{\left\{w_i > a_i\right\}} -\frac{\sigma}{\sqrt{2\pi}}\sum_{i=1}^ne^{-\frac{a_i^2}{2\sigma^2}}\right ) \leq -\epsilon \right) \\ \label{eq_app_low1}
& \leq e^{- \frac{n\epsilon^2}{2\left(\sigma^2  + \frac{M\epsilon}{3}\right)}}.
\end{align}
Using  the upper and lower tail inequalities obtained in \eqref{eq_app_upp1} and \eqref{eq_app_low1}, respectively, we get 
\begin{align*}
 & \mathbb{P}\left(\frac{1}{n}\left\vert\sum_{i=1}^n w_i\mathsf{1}_{\left\{w_i > a_i\right\}} -\frac{\sigma}{\sqrt{2\pi}}\sum_{i=1}^ne^{-\frac{a_i^2}{2\sigma^2}}\right \vert \geq \epsilon\right)  \\
 & \leq 2e^{- \frac{n\epsilon^2}{2\left(\sigma^2  + \frac{M\epsilon}{3}\right)}} \leq 2e^{-nk\min(\epsilon,\epsilon^2)}
\end{align*}
where $k$ is a positive constant (this is due to $M$ being finite). This proves \eqref{lem_app_conc_ineq1_eq1}. The concentration inequality in \eqref{lem_app_conc_ineq1_eq2} can be similarly proven, and will not be detailed here. 

\subsection{Proof of Lemma \ref{lem_app_small_O_P}} \label{pf_lem_app_small_O_P}

We first prove  \eqref{eq:thetaiyi_conc}.  Then  \eqref{eq:1yi_conc} immediately follows by setting $\theta_i =1$, $\forall i$.

Let us denote the event whose probability we want to bound by $\mc{E}$. In our case, 
\[ \mc{E} = \left\{ \frac{1}{n}\bigg \vert \sum_{i=1}^n \theta_i \mathsf{1}_{\left\{y_i > f({\mathbf{y}})\right\}} - \sum_{i=1}^n \theta_i \mathsf{1}_{\left\{y_i > a\right\}} \bigg\vert \geq \epsilon \right\}. \]
Then, for any $t >0$, we have
\begin{align}
\nonumber
&  \mathbb{P}(\mc{E}) \\ \nonumber 
&  \leq \mathbb{P}(\mc{E}, \, \{a < f(\mbf{y}) \leq a+t\} )  +\mathbb{P}(\mc{E}, \, \{ a-t \leq f(\mbf{y}) \leq a\} )\\ \nonumber 
& \quad + \mathbb{P}(\abs{f(\mbf{y}) -a} > t) \\ \nonumber 
& =  \mathbb{P}\left(  \frac{1}{n} \bigg\lvert  \sum_{i=1}^n {\theta_i} 1_{\{ a < y_i \leq  f(\mbf{y})\}}\bigg\rvert \geq \e, \, \{a < f(\mbf{y}) \leq a+t\} \right)  \\ \nonumber 
&   \quad +  \mathbb{P}\left(  \frac{1}{n}  \bigg\lvert \sum_{i=1}^n {\theta_i} 1_{\{ f(\mbf{y}) < y_i \leq  a \}}\bigg \rvert \geq \e, \, \{a -t < f(\mbf{y}) \leq a \} \right) \\ \nonumber
& \quad +
 \mathbb{P}(\abs{f(\mbf{y}) -a} > t)  \\ \nonumber
 & \leq \mathbb{P}\left(  \bigg[  \frac{1}{n} \sum_{i=1}^n \abs {\theta_i} 1_{\{ a  < y_i \leq a+t \}}\bigg] \geq \e\right) \\ \label{eq_lem_app_O_p1}
 & \quad +  \mathbb{P}\left(  \bigg[ \frac{1}{n} \sum_{i=1}^n \abs{\theta_i} 1_{\{ a-t < y_i \leq  a \}}\bigg] \geq \e \right) + 2e^{-nkt^2}.
\end{align}

\noindent Now, 
\begin{align} 
       \mathbb{P}( 1_{\left\{a  < y_i \leq a+t\right\}} = 1)  = \int_{a}^{a+t} \hspace{-5pt} \frac{1}{\sqrt{2\pi\sigma^2}}e^{-\frac{(y_i-\theta_i)^2}{2\sigma^2}}dy_i \leq  \frac{t}{\sqrt{2\pi\sigma^2}}  \label{eq_lem11}
      \end{align}
where we have used $e^{-\frac{(y_i-\theta_i)^2}{2\sigma^2}} \leq 1$. Let $Y \vcentcolon = \frac{1}{n}\sum_{i=1}^n Y_i$ where $Y_i \vcentcolon= \abs{\theta_i} \mathsf{1}_{\left\{ a < y_i \leq  a+t\right\}}$. Then, from \eqref{eq_lem11}, we have  
\begin{align*}
  0 \leq \mathbb{E}Y & = \frac{1}{n}\sum_{i=1}^n \abs{\theta_i}  \,  \mathbb{P}( 1_{\left\{a  < y_i \leq a+t\right\}} = 1) \\
  & \leq \frac{t}{n \sqrt{2\pi\sigma^2}}\sum_{i=1}^n \abs{\theta_i}. 
\end{align*}
  Since $Y_i \in [0,\abs{\theta_i}]$, from Hoeffding's inequality, for any $\epsilon_1 > 0$, we have $\mathbb{P}\left( Y - \mathbb{E}Y \geq \epsilon_1 \right)  \leq  \exp\{ -\frac{2n\epsilon_1^2}{\norm{\bst}^2/n} \}$, which implies
 \begin{align*}
\mathbb{P}\left( Y \geq \epsilon_1 + \frac{t \norm{\bst}_1}{n\sqrt{2\pi\sigma^2}}\right) & \leq  e^{-\frac{2n\epsilon_1^2}{\norm{\bst}^2/n}},
 \end{align*}
where $\norm{\bst}_1 := \sum_{i=1}^n \abs{\theta_i}$.  Now, set $\e_1 = \e/2$ and 
 \be 
 \label{eq:teps_def}
 t= \frac{\e  \sqrt{\pi \sigma^2/2}}{\norm{\bst}_1/n} \ee to obtain
 \begin{align} \nonumber 
 \mathbb{P}\left( Y  \geq \e  \right) &  \leq  e^{-\frac{n\epsilon^2}{2\norm{\bst}^2/n}} \\ \label{eq_lem_app_O_p2}
 \Rightarrow \mathbb{P}\left( \frac{1}{n}\sum_{i=1}^n \abs{\theta_i} 
 \mathsf{1}_{\left\{ a  < y_i \leq a+t \right\}}  \geq \epsilon \right) & \leq   e^{-\frac{n\epsilon^2}{2\norm{\bst}^2/n}}.
\end{align}
 A similar analysis yields 
  \begin{align}\label{eq_lem_app_O_p3}
 \mathbb{P}\left( \frac{1}{n}\sum_{i=1}^n \abs{\theta_i} \mathsf{1}_{\left\{a -t< y_i \leq a \right\}}  \geq \epsilon \right) \leq   e^{-\frac{n\epsilon^2}{2\norm{\bst}^2/n}}.
\end{align}
Using \eqref{eq_lem_app_O_p2} and \eqref{eq_lem_app_O_p3} in \eqref{eq_lem_app_O_p1} and recalling that $t$ is given by \eqref{eq:teps_def}, we obtain
\begin{align*}
 & \mathbb{P}\left(\frac{1}{n}\left\vert \sum_{i=1}^n \theta_i \mathsf{1}_{\left\{y_i > f({\mathbf{y}})\right\}} - \sum_{i=1}^n  \theta_i \mathsf{1}_{\left\{y_i > a\right\}} \right\vert \geq \epsilon \right) \\
 & \leq  2\left(e^{-\frac{n\epsilon^2k \pi\sigma^2}{2 \norm{\bst}^2_1/n^2}} +  e^{-\frac{n\epsilon^2}{2\norm{\bst}^2/n}} \right) \leq 4e^{-\frac{nk \epsilon^2}{\Vert \bst\Vert^2/n}}
\end{align*}
where $k$ is a positive constant. The last inequality holds because $\Vert \bst \Vert_1^2/n^2 < \Vert \bst \Vert^2/n$ (by the Cauchy-Schwarz inequality), and   $\limsup_{n \to \infty} \Vert \bst \Vert^2/n < \infty$ (by assumption). This proves \eqref{eq:thetaiyi_conc}.

Next, we prove \eqref{eq:wiyi_conc}. Using steps very similar to \eqref{eq_lem_app_O_p1}, we have, for $t > 0$, $\epsilon >0$,
\begin{align}\nonumber
 &\mathbb{P}\left(\frac{1}{n}\left\vert \sum_{i=1}^n w_i\mathsf{1}_{\left\{y_i > f({\mathbf{y}})\right\}} - \sum_{i=1}^n w_i\mathsf{1}_{\left\{y_i > a\right\}} \right\vert \geq \epsilon \right) \\ \nonumber  
  & \leq2e^{-nkt^2} + \mathbb{P}\left(\frac{1}{n}\sum_{i=1}^n \abs{w_i} \mathsf{1}_{\left\{a < y_i \leq  a+t \right\}}  \geq \epsilon \right) \\  \label{eq_lem_app_O_p4} 
   & \quad + \mathbb{P}\left(\frac{1}{n}\sum_{i=1}^n \abs{w_i} \mathsf{1}_{\left\{a-t <   y_i \leq a\right\}}  \geq \epsilon \right).
\end{align}
Now, let  $Y \vcentcolon = \frac{1}{n}\sum_{i=1}^n Y_i$ where 
\[ Y_i \vcentcolon= \abs{w_i} \mathsf{1}_{\left\{ a < y_i \leq  a+t\right\}} =  \abs{w_i} \mathsf{1}_{\left\{ a -\theta_i < w_i \leq  a -\theta_i +t\right\}} . \]
Noting that $\abs{w_i} \leq t + \abs{a-\theta_i}$ when $w _i \in [a-\theta_i, a-\theta_i+t ]$, we have
\begin{align*}
\expec[Y_i] & = \int_{a-\theta_i}^{a-\theta_i+t} \frac{\abs{w}}{\sqrt{2 \pi \sigma^2}} e^{-w^2/2\sigma^2}dw \\
& \overset{(i)}{=}t \left(\frac{\abs{c}}{\sqrt{2 \pi \sigma^2}} e^{-c^2/2\sigma^2}\right) 
\overset{(j)}{\leq} \frac{t}{\sqrt{2\pi e}}.
\end{align*}
Note that $(i)$ is from the mean value theorem for integrals with $c \in \left(a-\theta_i, a-\theta_i+t \right)$, and $(j)$ is because $xe^{-x^2} \leq 1/\sqrt{2e}$ for $x \geq 0 $. Hence
\[ 0 \leq \expec[Y] =  \frac{1}{n}\sum_{i=1}^n \expec[Y_i] \leq \frac{t}{\sqrt{2\pi e}}. \]
As each $Y_i$ takes values in an interval of length at most $t$, by Hoeffding's inequality we have for any $\e_1 >0$
\begin{align} \nonumber 
 \mathbb{P}(Y \geq \e_1 + \expec[Y]) &\leq 2e^{-2n \e_1^2/t^2} \\ \label{eq:Ye1_bnd1}
 \Rightarrow \mathbb{P}\left(Y \geq \e_1 + \frac{t}{\sqrt{2\pi e}}\right) & \leq 2e^{-2n \e_1^2/t^2}.
\end{align}
Now, set $\frac{t}{\sqrt{2\pi e}} = \sqrt{\epsilon_1}$. Using this value of $t$ in the RHS of \eqref{eq:Ye1_bnd1}, we obtain
\ben 
\mathbb{P}\left(\frac{1}{n}\sum_{i=1}^n \abs{w_i} \mathsf{1}_{\left\{a < y_i \leq  a+t \right\}}  \geq \epsilon_1 + \sqrt{\epsilon_1} \right)  \leq 2 e^{-n k_1\e_1  } 
\een
where $k_1 = 1/(\pi e)$. Setting $\epsilon_1 + \sqrt{\epsilon_1} = \epsilon$, we get $\sqrt{\epsilon_1} = \frac{\sqrt{4\epsilon+1}-1}{2}$. Using the following inequality for $x > 0$:
\[ \left(\sqrt{1+x}-1\right)^2  \geq 
\left\{ 
\begin{array}{cc}
{x^2}/{32}, & 0 \leq x \leq 3 \\
3x/4,& x >3,
\end{array}
\right.
\]
we obtain, 
\be \label{eq:wibnd1}
\mathbb{P}\left(\frac{1}{n}\sum_{i=1}^n \abs{w_i} \mathsf{1}_{\left\{a < y_i \leq  a+t \right\}}  \geq \epsilon \right)  \leq 2 e^{-n k \min(\e^2, \e) } 
\ee
where $k$ is a positive constant. Using similar steps, it can be shown that the third term on the RHS of \eqref{eq_lem_app_O_p4} can also be bounded as 
\be \label{eq:wibnd2}
 \mathbb{P}\left(\frac{1}{n}\sum_{i=1}^n \abs{w_i} \mathsf{1}_{\left\{a-t <   y_i \leq a\right\}}  \geq \epsilon \right) 
\leq2 e^{-n k \min(\e^2, \e) }. 
\ee
This completes the proof of \eqref{eq:wiyi_conc}.

\subsection{Proof of Lemma \ref{lem_ui}}\label{pf_lem_ui}
 Since $\{Y_n \}_{n=1}^{\infty}$ is UI, from Definition \ref{def_UI}, we have
 $\lim_{K \to \infty} \big(\limsup_{n \to \infty} \mathbb{E}\left[Y_n \mathbf{1}_{\{ Y_n \geq K\}} \right]\big) = 0$.
Therefore,
\begin{align*}
&\mathbb{E}\left[\vert X_n \vert \mathbf{1}_{\{\vert X_n \vert \geq K\}} \right] \\
& \leq  \mathbb{E}\left[c\vert Y_n \vert \mathbf{1}_{\{\vert X_n \vert \geq K\}} \right] +  \mathbb{E}\left[a \mathbf{1}_{\{\vert X_n \vert \geq K\}} \right]\\
 & \leq c\mathbb{E}\left[ Y_n  \mathbf{1}_{\left\{c Y_n  + a \geq K\right\}} \right] +  a\mathbb{E}\left[ \mathbf{1}_{\{c Y_n  + a \geq K\}} \right]\\
&= c\mathbb{E}\left[ Y_n  \mathbf{1}_{\left\{ Y_n  \geq \frac{K-a}{c}\right\}} \right] +  a\mathbb{E}\left[ \mathbf{1}_{\left\{ Y_n \geq \frac{K-a}{c}\right\}} \right]\\
& = c\mathbb{E}\left[Y_n \mathbf{1}_{\left\{ Y_n  \geq \frac{K-a}{c}\right\}} \right] +  a\mathbb{P}\left(Y_n \geq \frac{K-a}{c}\right).
\end{align*}
So,
\begin{align*}
& \lim_{K \to \infty}\left(\limsup_{n \to \infty} \mathbb{E}\left[\vert X_n \vert \mathbf{1}_{\{\vert X_n \vert \geq K\}} \right]\right) \\
& \leq c\lim_{K \to \infty}\left(\limsup_{n \to \infty} \mathbb{E}\left[ Y_n  \mathbf{1}_{\left\{Y_n  \geq \frac{K-a}{c}\right\}} \right]\right) \\
& \quad + a\lim_{K \to \infty}\left(\limsup_{n \to \infty} \mathbb{P}\left(Y_n \geq \frac{K-a}{c}\right) \right)= 0.
\end{align*}

\section{Proofs of Lemmas related to JS-estimators}\label{sec_appendixB}

\subsection{Proof of Lemma \ref{prop_y1_y2}}\label{pf_prop_y1_y2}
We first prove \eqref{eq:lem1_theti_yi} and \eqref{eq:lem1_theti_yi1}.   Then, \eqref{eq:lem1_1yi} and \eqref{eq:lem1_1yi1} immediately follow by setting $\theta_i=1$, for $1 \leq i \leq n$.

 From Lemma \ref{lem_app_small_O_P}, for any $\epsilon > 0$,
\begin{equation}\label{prop_y1_y2_eq1}
\mathbb{P}\left(\frac{1}{n}\left \vert\sum_{i=1}^n\theta_i \mathsf{1}_{\left\{y_i > \bar{y}\right\}} -\sum_{i=1}^n \theta_i\mathsf{1}_{\left\{y_i > \bar{\theta}\right\}}\right \vert\right)  \leq  4 e^{-\frac{nk \epsilon^2}{\Vert \bst\Vert^2/n}}.
 \end{equation}
  Since $\theta_i \mathsf{1}_{\left\{y_i > \bar{\theta} \right\}} \in \{0,\theta_i\}$  are independent for $1 \leq i \leq n$, from Hoeffding's inequality, we have, for any $\epsilon > 0$,
\begin{align} \nonumber
& \mathbb{P}\left( \left \vert \frac{1}{n}\sum_{i=1}^n \theta_i \mathsf{1}_{\left\{y_i > \bar{\theta} \right\}} - \frac{1}{n}\sum_{i=1}^n \theta_i \mathbb{E}\left[\mathsf{1}_{\left\{y_i > \bar{\theta} \right\}}\right] \right \vert > \epsilon \right) \\ \label{prop_y1_y2_eq2}
& \leq 2e^{-\frac{2n\epsilon^2}{\norm{\bst}^2/n}} .
\end{align}
 Also for each $i$,
\begin{align*}
\mathbb{E}\left[ \mathsf{1}_{\left\{y_i > \bar{\theta} \right\}} \right] & = \mathbb{P}\left( y_i > \bar{\theta} \right)  = \mathbb{P}\left( w_i > \bar{\theta} - \theta_i\right ) \\
&=  Q\left(\frac{\bar{\theta}-\theta_i}{\sigma} \right).
\end{align*}

\noindent Therefore, from \eqref{prop_y1_y2_eq1} and \eqref{prop_y1_y2_eq2}, we obtain
\begin{equation}\label{eq_cow_ineq_1}
\frac{1}{n}\sum_{i=1}^n \theta_i \mathsf{1}_{\left\{y_i > \bar{y} \right\}} \doteq \frac{1}{n}\sum_{i=1}^n  \theta_i Q\left(\frac{\bar{\theta}-\theta_i}{\sigma} \right).
\end{equation}
 The concentration result in \eqref{eq:lem1_theti_yi1} immediately follows by writing
$\mathsf{1}_{\left\{y_i \leq \bar{y} \right\}}  = 1 -  \mathsf{1}_{\left\{y_i > \bar{y} \right\}}$.

To  prove \eqref{eq:lem1_yigeq}, we write
\[  \frac{1}{n} \sum_{i=1}^n y_i \mathsf{1}_{\{y_i > \bar{y}\}}  =  \frac{1}{n} \sum_{i=1}^n\theta_i\mathsf{1}_{\{y_i > \bar{y}\}}+ 
 \frac{1}{n} \sum_{i=1}^n w_i \mathsf{1}_{\{y_i > \bar{y}\}}. \]
Hence, we have to show that 
 \be
  \frac{1}{n} \sum_{i=1}^n w_i \mathsf{1}_{\{y_i > \bar{y}\}} \doteq \frac{\sigma}{n \sqrt{2\pi}}\sum_{i=1}^n e^{-\frac{\left(\bar{\theta}-\theta_i\right)^2}{2\sigma^2}}.
  \label{eq:wts_wi}
 \ee
From Lemma \ref{lem_app_small_O_P}, for any $\epsilon > 0$, we have
\begin{equation}\label{eq:wiyi}
\mathbb{P}\left(\frac{1}{n}\left \vert\sum_{i=1}^n w_i \mathsf{1}_{\left\{y_i > \bar{y}\right\}} -\sum_{i=1}^n w_i \mathsf{1}_{\left\{y_i > \bar{\theta}\right\}}\right \vert\right)  \leq  4 e^{-\frac{nk \epsilon^2}{\Vert \bst\Vert^2/n}}.
 \end{equation}
 Now,
\begin{align*}
 \mathbb{E}\left[w_i \mathsf{1}_{\{y_i > \bar{\theta} \}} \right] & = \int_{-\infty}^{\infty}w_i \mathsf{1}_{\{y_i > \bar{\theta} \}}\left(\frac{1}{\sqrt{2\pi\sigma^2}}e^{-\frac{w_i^2}{2\sigma^2}} \right)dw_i\\
 & = \int_{\bar{\theta}- \theta_i}^{\infty}\frac{w_i}{\sqrt{2\pi\sigma^2}}e^{-\frac{w_i^2}{2\sigma^2}}dw_i= \frac{\sigma}{\sqrt{2\pi}}e^{-\frac{\left(\bar{\theta}-\theta_i\right)^2}{2\sigma^2}}. 
\end{align*}
Using Lemma \ref{lem_app_conc_ineq1}, we get, for any $\epsilon > 0$, 
\begin{align}\nonumber 
& \mathbb{P}\left(\frac{1}{n}\left\vert\sum_{i=1}^n w_i\mathsf{1}_{\left\{w_i > \bar{\theta}\right\}} -\frac{\sigma}{\sqrt{2\pi}}\sum_{i=1}^ne^{-\frac{\left(\bar{\theta}-\theta_i\right)^2}{2\sigma^2}}\right \vert \geq \epsilon\right) \\ \label{eq_cow_ineq_3}
& \leq 2e^{-nk\min(\epsilon,\epsilon^2)}.
\end{align}
We obtain  \eqref{eq:wts_wi} by combining \eqref{eq:wiyi} and \eqref{eq_cow_ineq_3}.

Similarly, \eqref{eq:lem1_yileq} can be shown using Lemma \ref{lem_app_small_O_P} and Lemma \ref{lem_app_conc_ineq1}
 to establish that 
\begin{equation*}
\frac{1}{n}\sum_{i=1}^n w_i \mathsf{1}_{\left\{w_i \leq \bar{y}\right\}} \doteq - \frac{\sigma}{n \sqrt{2\pi}}\sum_{i=1}^n
e^{-\frac{\left(\bar{\theta}-\theta_i\right)^2}{2\sigma^2}}.
\end{equation*}

\subsection{Proof of Lemma \ref{prop_bias}}\label{pf_prop_bias}
From Lemma \ref{lem:fhy_conc}, we have, for any $\epsilon >0$,
\begin{align} \nonumber
 & \mathbb{P}\left(\frac{1}{n}\frac{\sigma^2}{2\delta}\left\vert\sum_{i=0}^{n} \mathsf{1}_{\left\{\left\vert y_i - \bar{{y}}  \right\vert \leq \delta \right\}} - \sum_{i=0}^{n} \mathsf{1}_{\left\{\left\vert y_i - \bar{\theta} \right\vert \leq \delta \right\}} \right\vert \geq \epsilon \right) \\ \label{prop_bias_eq1}
 &\leq 8e^{-\frac{nk\epsilon^2\delta^2}{\sigma^4}}.
\end{align}
Further, from Hoeffding's inequality, 
\begin{align} \nonumber 
& \mathbb{P}\left(\frac{1}{n}\left\vert\frac{\sigma^2}{2\delta}\sum_{i=0}^{n} \mathsf{1}_{\left\{\left\vert y_i - \bar{\theta} \right\vert \leq \delta \right\}}  -  \frac{\sigma^2}{2\delta}\sum_{i=0}^{n}\mathbb{E}\left[ \mathsf{1}_{\left\{\left\vert y_i - \bar{\theta}\right \vert \leq \delta \right\}}\right]\right \vert \geq \epsilon \right) \\ \label{prop_bias_eq2}
& \leq 2e^{-\frac{8n\delta^2\epsilon^2}{\sigma^4}}.
\end{align}
Also,
\begin{align} \nonumber 
 & \frac{\sigma^2}{2\delta}\sum_{i=0}^{n} \mathbb{E}\left[ \mathsf{1}_{\left\{\left\vert y_i - \bar{\theta}  \right\vert \leq \delta \right\}}\right] = \frac{\sigma^2}{2\delta}\sum_{i=0}^{n} \mathbb{P}\left(\left\vert y_i - \bar{\theta} \right \vert \leq \delta \right) \\ \label{prop_bias_eq3}
 & = \frac{\sigma^2}{2\delta}\sum_{i=0}^{n} \int_{\bar{\theta}-\delta}^{\bar{\theta}+\delta}\frac{1}{\sqrt{2\pi\sigma^2}}e^{-\frac{\left(y_i-\theta_i\right)^2}{2\sigma^2}}dy_i .
\end{align}
\noindent From the first mean value theorem for integrals, $\exists \varepsilon_i \in \left(-{\delta},{\delta}\right)$ such that \begin{equation*}
\int_{\bar{\theta}-\delta}^{\bar{\theta}+\delta}\frac{1}{\sqrt{2\pi\sigma^2}}e^{-\frac{\left(y_i-\theta_i\right)^2}{2\sigma^2}}dy_i = 2\delta\left(\frac{1}{\sqrt{2\pi\sigma^2}}e^{-\frac{\left(\bar{\theta} + \varepsilon_i-\theta_i\right)^2}{2\sigma^2}}\right)
\end{equation*}
and so the RHS of \eqref{prop_bias_eq3} can be written as 
\begin{equation*}
\begin{split}
&\frac{\sigma^2}{2n\delta}\sum_{i=0}^{n} \int_{\bar{\theta}-\delta}^{\bar{\theta}+\delta}\frac{1}{\sqrt{2\pi\sigma^2}}e^{-\frac{\left(y_i-\theta_i\right)^2}{2\sigma^2}}dy_i \\
 & =  \frac{\sigma}{n\sqrt{2\pi}}\sum_{i=0}^{n} e^{-\frac{\left(\bar{\theta} + \varepsilon_i -\theta_i\right)^2}{2\sigma^2}} .
\end{split}
\end{equation*}
Now, let $x_i \vcentcolon= \bar{\theta} - \theta_i$. Then, since $\vert \varepsilon_i \vert \leq \delta$, we have
\begin{align*}
  &\frac{1}{n\sqrt{2\pi\sigma^2}}\sum_{i=0}^{n} \left\vert e^{-\frac{x_i^2}{2\sigma^2}} - e^{-\frac{\left(x_i + \varepsilon_i \right)^2}{2\sigma^2}}\right \vert \\
 & \leq \delta \max \,  \frac{d}{dx}\left( \frac{1}{\sqrt{2\pi\sigma^2}}e^{-\frac{x^2}{2\sigma^2}}\right) = \frac{\delta}{\sigma^2\sqrt{2\pi e}}.
\end{align*}

\noindent Therefore, 
\begin{align} \nonumber 
 \frac{\sigma^2}{2n\delta}\sum_{i=0}^{n} \int_{\bar{\theta}-\delta}^{\bar{\theta}+\delta}\frac{1}{\sqrt{2\pi\sigma^2}}e^{-\frac{\left(y_i-\theta_i\right)^2}{2\sigma^2}}dy_i \\ \label{prop_bias_eq4}
  = \frac{\sigma}{n\sqrt{2\pi}}\sum_{i=0}^{n} e^{-\frac{\left(\bar{\theta} -\theta_i\right)^2}{2\sigma^2}} + \kappa_n\delta 
\end{align}
where $\vert \kappa_n \vert \leq \frac{1}{\sqrt{2\pi e}} $.  Using \eqref{prop_bias_eq4} in \eqref{prop_bias_eq3}, and then the obtained result in \eqref{prop_bias_eq2} and \eqref{prop_bias_eq1}, the proof of the lemma is complete.

\subsection{Proof of Lemma \ref{lem3}}\label{pf_lem3}

We have
 \begin{align}  \nonumber
\frac{1}{n}\left\Vert \mathbf{y}-\boldsymbol{\nu}_{2} \right\Vert^2 & = \frac{1}{n} \sum_{i=1}^n\left(y_i - a_1 \right)^2\mathsf{1}_{\left\{y_i > \bar{y}\right\}} \\ \label{lem3_eq0}
& + \frac{1}{n} \sum_{i=1}^n\left(y_i - a_2 \right)^2\mathsf{1}_{\left\{y_i \leq \bar{y}\right\}}.
\end{align}
Now,
\begin{align*}
 &\frac{1}{n}\sum_{i=1}^n\left(y_i - a_1 \right)^2\mathsf{1}_{\left\{y_i > \bar{y}\right\}}\\
 & = \frac{1}{n}\left[ \sum_{i=1}^ny_i^2 \mathsf{1}_{\left\{y_i > \bar{y}\right\}}  + \sum_{i=1}^na_1^2 \mathsf{1}_{\left\{y_i > \bar{y}\right\}}-2\sum_{i=1}^na_1y_i \mathsf{1}_{\left\{y_i > \bar{y}\right\}}\right]\\
  & = \frac{1}{n}\left[ \sum_{i=1}^n w_i^2 \mathsf{1}_{\left\{w_i > \bar{y} -\theta_i\right\}} + \sum_{i=1}^n \theta_i^2 \mathsf{1}_{\left\{y_i > \bar{y}\right\}}  \right.\\
  & \quad + 2\sum_{i=1}^n \theta_i w_i \mathsf{1}_{\left\{w_i > \bar{y} -\theta_i\right\}} + \sum_{i=1}^na_1^2 \mathsf{1}_{\left\{y_i > \bar{y}\right\}} \\
  & \quad -2\sum_{i=1}^na_1y_i \mathsf{1}_{\left\{y_i > \bar{y}\right\}}\bigg]
\end{align*}
and similarly,
\begin{align*}
&  \frac{1}{n}\sum_{i=1}^n\left(y_i - a_2 \right)^2\mathsf{1}_{\left\{y_i \leq \bar{y}\right\}} =\frac{1}{n}\bigg[ \sum_{i=1}^n w_i^2 \mathsf{1}_{\left\{w_i \leq \bar{y} -\theta_i\right\}}  \\
 & \quad  + \sum_{i=1}^n \theta_i^2 \mathsf{1}_{\left\{y_i \leq \bar{y}\right\}}  + 2\sum_{i=1}^n\theta_i w_i \mathsf{1}_{\left\{w_i\leq\bar{y} -\theta_i\right\}} \\
 & \quad  + \sum_{i=1}^na_2^2 \mathsf{1}_{\left\{y_i \leq \bar{y}\right\}}  -2\sum_{i=1}^na_2y_i \mathsf{1}_{\left\{y_i \leq \bar{y}\right\}}\bigg].
\end{align*}

\noindent Therefore, from \eqref{lem3_eq0}
\begin{align}
\nonumber
  & \frac{1}{n}\left\Vert \mathbf{y}-\boldsymbol{\nu}_{2} \right\Vert^2  =  \frac{1}{n}\sum_{i=1}^n w_i^2 + \frac{\Vert \boldsymbol{\theta} \Vert^2}{n} +\frac{2}{n}\sum_{i=1}^n \theta_i w_i \\ \nonumber
  & \qquad+  \frac{1}{n}\left(\sum_{i=1}^na_1^2 \mathsf{1}_{\left\{y_i > \bar{y}\right\}} -2\sum_{i=1}^na_1y_i \mathsf{1}_{\left\{y_i > \bar{y}\right\}} \right.\\ \label{lem3_eq1}
  & \qquad + \sum_{i=1}^na_2^2 \mathsf{1}_{\left\{y_i \leq \bar{y}\right\}} -2\sum_{i=1}^na_2y_i \mathsf{1}_{\left\{y_i \leq \bar{y}\right\}}\bigg).
\end{align}

\noindent Since $\frac{1}{n}\sum_{i=1}^n \theta_i w_i \sim \mathcal{N}\left(0,\frac{\Vert\boldsymbol{\theta} \Vert^2}{n^2}\right)$, 
\begin{equation}\label{lem3_eq2}
 \mathbb{P}\left( \left\vert \frac{1}{n}\sum_{i=1}^n \theta_i w_i \right\vert \geq \epsilon\right) \leq e^{-\frac{n\epsilon^2}{2\Vert\boldsymbol{\theta}\Vert^2/n}}.
\end{equation}

\noindent From Lemma \ref{lem_app_conc_ineq3}, we have, for any $\epsilon >0$,
\begin{equation*}
 \mathbb{P}\left(\left\vert\frac{1}{n}\sum_{i=1}^n w_i^2 - \sigma^2 \right \vert \geq \epsilon\right) \leq 2 e^{-nk\min(\epsilon,\epsilon^2)} 
\end{equation*}
where $k$ is a positive constant. Next, we claim that
\be a_1 \doteq c_1 + \kappa_n \delta + o(\delta), \quad a_2 \doteq c_2 + \kappa_n \delta + o(\delta), \label{eq:a1a2_conc} \ee
where $c_1, c_2$ are defined in \eqref{eq_c_1_c_2}. The concentration in \eqref{eq:a1a2_conc} follows from Lemmas \ref{prop_y1_y2} and \ref{prop_bias}, together with the results on concentration of products and reciprocals in Lemmas \ref{lem1_app2} and \ref{lem1_app}, respectively. Further, using \eqref{eq:a1a2_conc} and Lemma \ref{lem1_app2} again, we obtain
$a_1^2 \doteq c_1^2 + \kappa_n \delta + o(\delta)$ and
\begin{equation}\label{lem3_eq4}
a_1^2\left(\frac{1}{n}\sum_{i=1}^n \mathsf{1}_{\left\{y_i > \bar{y} \right\}}\right) \doteq \frac{c_1^2}{n}\sum_{i=1}^nQ\left(\frac{\bar{\theta}-\theta_i}{\sigma} \right) + \kappa_n\delta+o(\delta).
\end{equation}
Similarly, 
\begin{align}
 \nonumber
 &a_1\left(\frac{2}{n}\sum_{i=1}^ny_i \mathsf{1}_{\left\{y_i > \bar{y}\right\}}\right)\\ \nonumber 
 &\doteq \frac{2c_1}{n}\left(\sum_{i=1}^n \theta_iQ\left(\frac{\bar{\theta}-\theta_i}{\sigma}\right)+ \frac{\sigma}{\sqrt{2\pi}}\sum_{i=1}^ne^{-\frac{\left(\bar{\theta}-\theta_i\right)^2}{2\sigma^2}} \right) \\ 
 \nonumber 
 & \quad + \kappa_n\delta+o(\delta) \\ \nonumber
 &\doteq\frac{2c_1}{n}\left(c_1\sum_{i=1}^n Q\left(\frac{\bar{\theta}-\theta_i}{\sigma}\right)+ \frac{\sigma}{\sqrt{2\pi}}\sum_{i=1}^ne^{-\frac{\left(\bar{\theta}-\theta_i\right)^2}{2\sigma^2}} \right) \\ \nonumber 
  & \quad + \kappa_n\delta+o(\delta) \\  \nonumber 
  &\doteq \frac{2c_1^2}{n}\sum_{i=1}^n Q\left(\frac{\bar{\theta}-\theta_i}{\sigma}\right) + \frac{2c_1\sigma}{n\sqrt{2\pi}}\sum_{i=1}^ne^{-\frac{\left(\bar{\theta}-\theta_i\right)^2}{2\sigma^2}} \\ \label{lem3_eq5}
  & \quad + \kappa_n\delta+o(\delta) .
\end{align}

\noindent Employing the same steps as above, we get 
\begin{align}
\label{lem3_eq6}
&a_2^2\left(\frac{1}{n}\sum_{i=1}^n \mathsf{1}_{\left\{y_i \leq \bar{y} \right\}}\right) \doteq \frac{c_2^2}{n}\sum_{i=1}^nQ^c\left(\frac{\bar{\theta}-\theta_i}{\sigma} \right) + \kappa_n\delta+o(\delta),\\
\nonumber
& a_2\left(\frac{2}{n}\sum_{i=1}^ny_i \mathsf{1}_{\left\{y_i \leq \bar{y}\right\}}\right) 
   \doteq \frac{2c_2^2}{n}\sum_{i=1}^n Q^c\left(\frac{\bar{\theta}-\theta_i}{\sigma}\right)\\ \label{lem3_eq7}
  &  \hspace{0.4in} - \frac{2c_2\sigma}{n\sqrt{2\pi}}\sum_{i=1}^ne^{-\frac{\left(\bar{\theta}-\theta_i\right)^2}{2\sigma^2}}  + \kappa_n\delta+o(\delta).
\end{align}

\noindent Therefore, using \eqref{lem3_eq2}-\eqref{lem3_eq7} in \eqref{lem3_eq1}, we finally obtain
 \begin{equation*}
\begin{split}
& \frac{1}{n}\left\Vert \mathbf{y}-\boldsymbol{\nu}_{2} \right\Vert^2 \\
&\doteq  \frac{\Vert \boldsymbol{\theta} \Vert^2}{n} + \sigma^2 - \frac{ c_1^2}{n}\sum_{i=1}^nQ\left(\frac{\bar{\theta}-\theta_i}{\sigma}\right) - \frac{c_2^2}{n}\sum_{i=1}^nQ^c\left(\frac{\bar{\theta}-\theta_i}{\sigma}\right) \\
 & \quad- \left(\frac{2}{n}\right)\left(\frac{\sigma}{\sqrt{2\pi}} \right)\left(\sum_{i=1}^ne^{-\frac{\left(\bar{\theta}-\theta_i\right)^2}{2\sigma^2}} \right)\left(c_1 - c_2\right) \\
 & \quad + \kappa_n\delta+o(\delta) ,
\end{split}
\end{equation*}
which completes the proof of the lemma.

\subsection{Proof of Lemma \ref{lem_theta}}\label{pf_lem_theta}
 The proof is along the same lines as that of Lemma \ref{lem3}. We have
 \begin{equation*}
\begin{split}
  & \frac{1}{n}\Vert \boldsymbol{\theta}-\boldsymbol{\nu}_{2} \Vert^2 \\
  &  = \frac{1}{n}\left[ \sum_{i=1}^n\left(\theta_i - a_1 \right)^2\mathbf{1}_{\left\{y_i > \bar{y}\right\}} + \sum_{i=1}^n\left(\theta_i - a_2 \right)^2\mathbf{1}_{\left\{y_i \leq \bar{y}\right\}}\right]\\
  &=   \frac{\Vert \boldsymbol{\theta} \Vert^2}{n} + \frac{1}{n}\left(\sum_{i=1}^na_1^2 \mathbf{1}_{\left\{y_i > \bar{y}\right\}} -2\sum_{i=1}^na_1\theta_i \mathbf{1}_{\left\{y_i > \bar{y}\right\}} \right.\\
  & \quad + \sum_{i=1}^na_2^2 \mathbf{1}_{\left\{y_i \leq \bar{y}\right\}} -2\sum_{i=1}^na_2\theta_i \mathbf{1}_{\left\{y_i \leq \bar{y}\right\}}\Bigg)\\
  & \doteq \frac{\Vert \boldsymbol{\theta} \Vert^2}{n} + \frac{c_1^2}{n}\sum_{i=1}^n Q\left(\frac{\bar{\theta}-\theta_i}{\sigma}\right) -\frac{2c_1}{n}\sum_{i=1}^n\theta_i Q\left(\frac{\bar{\theta}-\theta_i}{\sigma}\right)\\
  &\quad +\frac{1}{n}\left( c_2^2\sum_{i=1}^n Q^c\left(\frac{\bar{\theta}-\theta_i}{\sigma}\right) -2c_2\sum_{i=1}^n\theta_i Q^c\left(\frac{\bar{\theta}-\theta_i}{\sigma}\right)\right) \\
  & \quad + \kappa_n\delta +o(\delta) \\
  &\doteq\frac{\Vert \boldsymbol{\theta} \Vert^2}{n} - \frac{ c_1^2}{n}\sum_{i=1}^nQ\left(\frac{\bar{\theta}-\theta_i}{\sigma}\right) - \frac{c_2^2}{n}\sum_{i=1}^nQ^c\left(\frac{\bar{\theta}-\theta_i}{\sigma}\right) \\
  & \quad + \kappa_n\delta+o(\delta).
 \end{split}
\end{equation*}

\subsection*{Acknowledgement}
 The authors thank R. Samworth for useful discussions on James-Stein estimators, and A. Barron and an anonymous referee for their comments which led to a much improved manuscript.



\end{document}